\documentclass[a4paper,10pt]{amsart}

\usepackage{amsrefs}
 \usepackage{amsfonts,mathrsfs}
 \usepackage{amsthm}
 \usepackage{amssymb}
 \usepackage{enumerate}
 \usepackage{tikz-cd}
 \usepackage{cleveref}
 \usepackage{enumitem}
\theoremstyle{definition}
\newtheorem{Def}{Definition}[section]
\newtheorem{ex}[Def]{Example}
\newtheorem{rem}[Def]{Remark}

\theoremstyle{plain}
\newtheorem{prop}[Def]{Proposition}
\newtheorem{thm}[Def]{Theorem}
\newtheorem*{thm*}{Theorem}
\newtheorem{lem}[Def]{Lemma}

\newtheorem{cor}[Def]{Corollary}
\newtheorem*{cor*}{Corollary}

\newtheorem*{con*}{Conjecture}

\newtheorem*{verm*}{Vermutung}
\newtheorem{notation}[Def]{Notation}

\newcommand{\logX}{Y(\R)}
\newcommand{\bm}{m}

\usepackage[draft]{fixme}
\fxsetup{theme=color, mode=multiuser, noinline}
\FXRegisterAuthor{la}{LA}{\color{red}Lior}
\FXRegisterAuthor{pk}{PK}{Pavel}
\FXRegisterAuthor{mk}{MK}{\color{orange}Mario}
\FXRegisterAuthor{cv}{CV}{\color{teal}Cynthia}

\usepackage{mathrsfs}

\newcommand{\var}{\operatorname{var}}
\newcommand{\varbar}{\operatorname{\overline{var}}}

\newcommand{\Ima}{\operatorname{Im}}

\newcommand{\Hom}{\operatorname{Hom}} 

\newcommand{\Vol}{\operatorname{Vol}}

\newcommand{\Gr}{\operatorname{Gr}}

\newcommand{\T}{\operatorname{T}}

\newcommand{\GL}{\operatorname{{\mathbf GL}}}

\newcommand{\cA}{{\mathcal A}}

\newcommand{\cL}{{\mathcal L}}
\newcommand{\cM}{{\mathcal M}}
\newcommand{\cN}{{\mathcal N}}

\newcommand{\cS}{{\mathcal S}}

\newcommand{\C}{{\mathbb C}}
\newcommand{\R}{{\mathbb R}}
\newcommand{\pp}{\mathbb{P}}
\newcommand{\Q}{{\mathbb Q}}
\newcommand{\N}{{\mathbb N}}
\newcommand{\Z}{{\mathbb Z}}

\title[Higher Dimensional Fourier Quasicrystals]{Higher Dimensional Fourier Quasicrystals from Lee--Yang Varieties}
\author{Lior Alon}
\address{Massachusetts Institute of Technology, USA} 
\email{lioralon@mit.edu}
\author{Mario Kummer}
\address{Technische Universit\"at, Dresden, Germany} 
\email{mario.kummer@tu-dresden.de}
\author{Pavel Kurasov}
\address{Stockholm University, Sweden} 
\email{kurasov@math.su.se}
\author{Cynthia Vinzant}
\address{University of Washington, Seattle, USA} 
\email{vinzant@uw.edu}

\thanks{LA was supported by the Simons Foundation Grant 601948, DJ. MK has been partially supported by the Deutsche Forschungsgemeinschaft (DFG, German Research Foundation), Grant number 502861109. PK was partially supported by the Swedish Research Council Grant 2020-03780.
CV was partially supported by NSF grant No.~DMS-2153746 and a Sloan Research Fellowship. }

\begin{document}


\begin{abstract}
In this paper, we construct Fourier quasicrystals with unit masses in arbitrary dimensions. 
This generalizes a one-dimensional construction of Kurasov and Sarnak.
To do this, we employ a class of complex algebraic varieties avoiding certain regions in $\C^n$, which generalize hypersurfaces defined by Lee--Yang polynomials.   We show that these are Delone almost periodic sets that have at most finite intersection with every discrete periodic set. 
\end{abstract}
\maketitle

\section{Introduction}
{ In this paper, we describe a general construction of Delone sets $\Lambda\subseteq\R^{d}$, in every dimension $d$, whose counting measures $\mu_{\Lambda} = \sum_{x \in \Lambda} \delta_{x}$ are Fourier quasicrystals with unit masses. Our construction is based on the interplay between real algebraic geometry and discrete geometry, using a new class of complex algebraic varieties that avoid certain regions in $\mathbb{C}^{n}$.} This directly generalizes previous work by Kurasov and Sarnak \cite{kurasovsarnak}, who constructed one-dimensional Fourier quasicrystals with unit masses using Lee--Yang polynomials. The notion of a Fourier quasicrystal was introduced and used in the context of measures, see e.g. \cites{lev2015quasicrystals, favorov2016fourier, kurasovsarnak, lagarias2000mathematical, lev2017fourier}. 

{ A \emph{crystalline measure} is a tempered measure of the form $\mu = \sum_{x \in \Lambda} a_{x} \delta_{x}$, with discrete support $\Lambda$, such that the distributional Fourier transform is also a (complex) measure with discrete support $\hat{\mu}=\sum_{\xi\in\Lambda'}c_{\xi}\delta_{\xi},\ c_{\xi}\in\C$, see \cite{meyerquestion}. A measure $\mu $ is a \emph{Fourier quasicrystal} if it is crystalline and the absolute value measures $|\mu|$ and $|\hat{\mu}|$ are tempered.\footnote{Not all crystalline measures are Fourier quasicrystals, see e.g. \cite{favorov2024crystalline}. } We say $\mu$ has unit masses if $a_{x} = 1$ for all $x \in \Lambda$, in which case $\mu$ is the counting measure of $\Lambda$. To ease the reading, we say that a set $\Lambda \subseteq \mathbb{R}^{d}$ is a Fourier quasicrystal if its counting measure $\mu_{\Lambda} = \sum_{x \in \Lambda} \delta_{x}$ is a Fourier quasicrystal. Explicitly, we have the following.  }
\begin{Def}[FQ]\label{def:FQ}
	A set $\Lambda\subseteq\R^{d}$ is a \emph{Fourier quasicrystal} (FQ) if there exists a set $\Lambda'\subseteq\R^{d}$ and non-zero complex coefficients $(c_{\xi})_{\xi\in\Lambda'}$, called the \emph{spectrum} and the \emph{Fourier coefficients} of $\Lambda$, respectively, such that
\begin{enumerate}
	\item for every Schwartz function $f\in\mathcal{S}(\R^{d})$, 
	\begin{equation*}
		\sum_{x\in\Lambda}\hat{f}(x)=  \sum_{\xi\in\Lambda'}c_{\xi}f(\xi),
	\end{equation*}
	\item both $\Lambda$ and $\Lambda'$ are discrete (locally finite) with polynomial growth bound: there exist $C>0$ and $m\in\N$ such that 
	\[\left|\Lambda\cap B_{R}(0)\right|+\sum_{\xi\in\Lambda'\cap B_{R}(0)}|c_{\xi}|\le CR^{m}\quad\text{for all }\ R>1,\]
	where $B_{R}(0)$ is the ball of radius $R$ around the origin.
\end{enumerate}

\begin{figure}
	\begin{center}
		\includegraphics[height=1.5in]{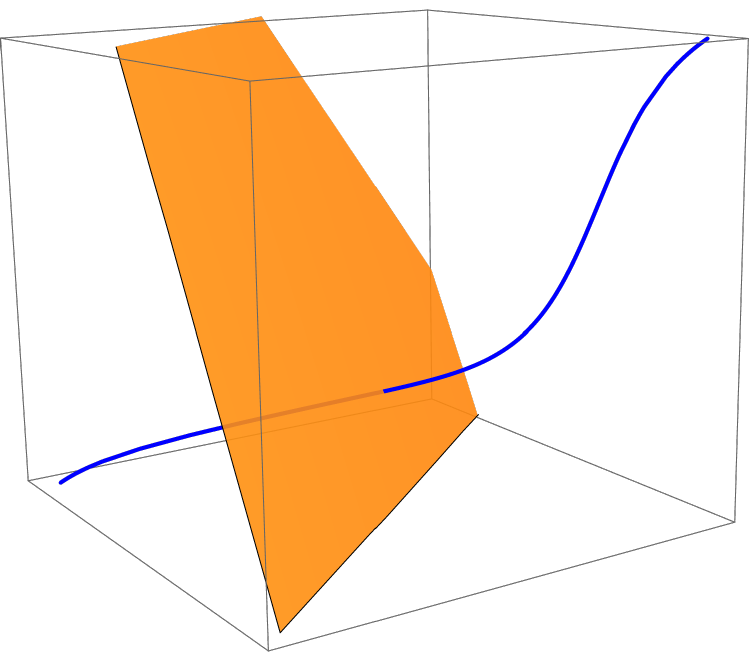} \ \ \ \  \includegraphics[height=1.5in]{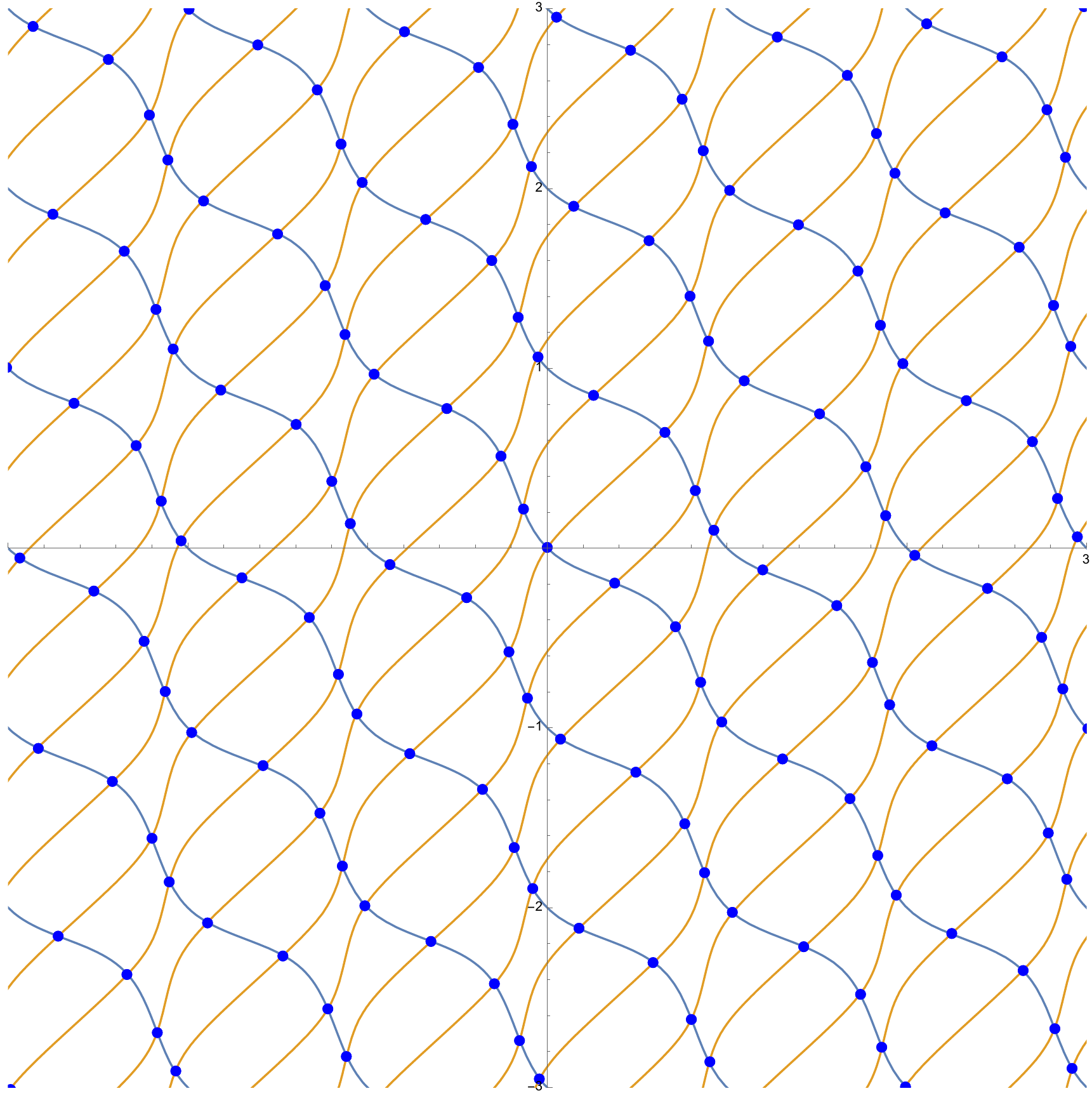}
		\  \includegraphics[height=1.5in]{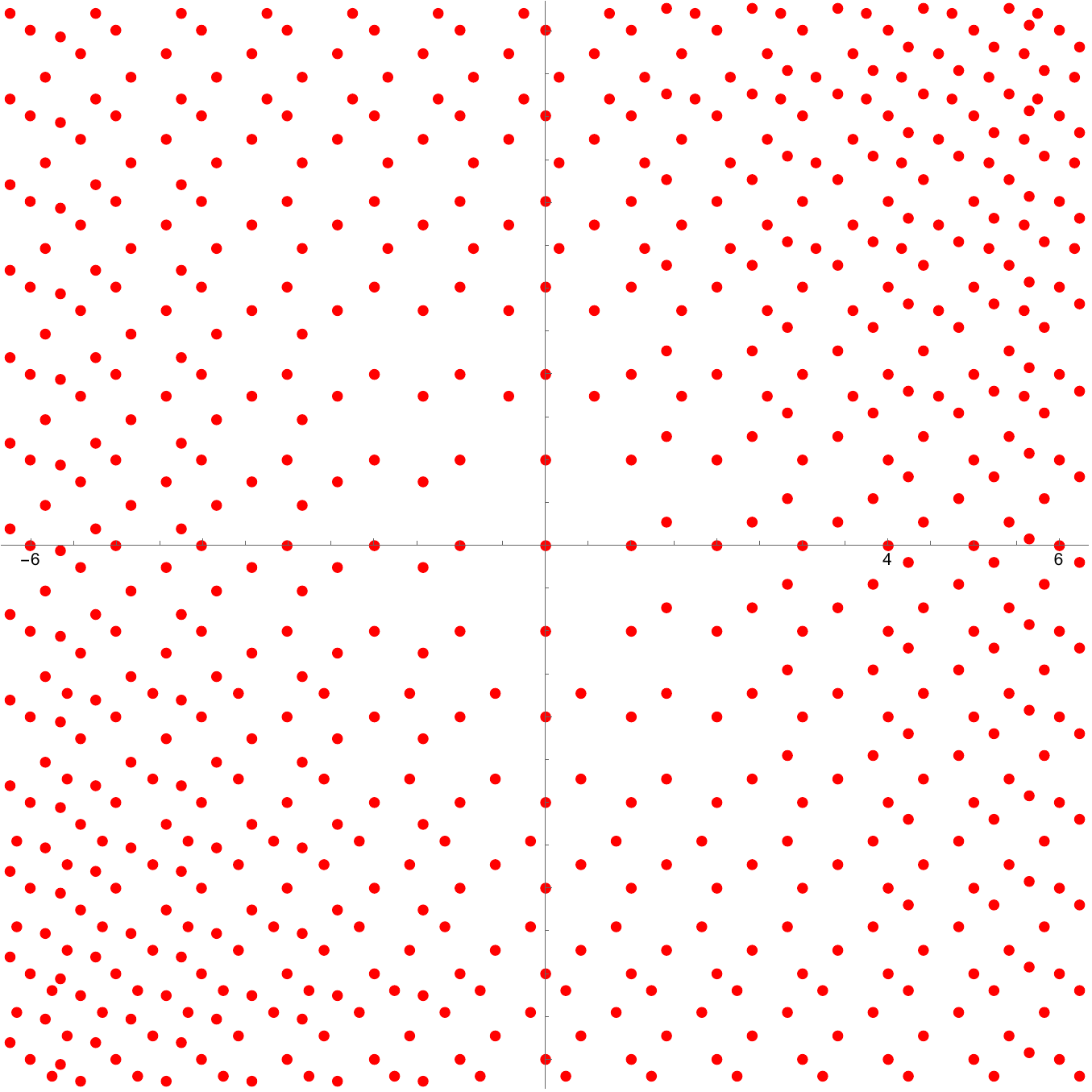}
	\end{center}
	\caption{From \Cref{ex:runningRational}, the curve $\{\theta \mid \exp(2\pi i\theta)\in X \}$, the Fourier quasicrystal $\Lambda$ (blue dots), and a discrete set which contains the spectrum of $\Lambda$ (red dots).} 
	\label{fig:RationalCurveIntro}
\end{figure}

\end{Def}
{ Throughout the paper, the term \emph{Fourier quasicrystal} (FQ) is always understood in the set-theoretic sense of \Cref{def:FQ}, except for the introductory section. In the rest of the introduction, we use the terms \emph{FQ-sets} and \emph{FQ-measures} to make a clear distinction between sets in $\R^d$ that are Fourier quasicrystals in the sense of \Cref{def:FQ} and measures on $\R^d$ that are Fourier quasicrystals.}

\begin{Def}[Delone]
	A set $\Lambda \subseteq \R^d$ is \emph{Delone} if it is uniformly discrete and relatively dense. Recall that $\Lambda$ is uniformly discrete if there exists $r > 0$ such that $|x - x'| \ge r$ for every pair of distinct points $x, x' \in \Lambda$, and relatively dense if there exists $R > 0$ such that the balls of radius $R$ around points of $\Lambda$ cover $\R^d$:
	\begin{equation*}
		\bigcup_{x \in \Lambda} B_R(x) = \R^d.
	\end{equation*}
\end{Def}

Lagarias states in \cite{lagarias2000mathematical}*{Corollary~3.3} (following \cite{cordoba1989dirac}) that if $\Lambda \subseteq \R^d$ is a Delone FQ-set whose Fourier transform is translation bounded, then $\Lambda$ is periodic.  Lagarias conjectured that if an FQ-measure has Delone support and Delone spectrum, then it is periodic. The conjecture was proven by Lev and Olevskii in \cite{lev2015quasicrystals} for $d=1$, and also for $d > 1$ under an additional positivity assumption. Favorov showed that the positivity assumption is needed by constructing a non-positive two-dimensional counterexample \cite{favorov2016fourier}. Since counting measures are positive, the conjecture is true for FQ-sets in arbitrary dimension. {This leads naturally to the following question: 
	\begin{center}
		Are there non-periodic Delone FQ-sets?
\end{center}
We provide a positive answer in every dimension $d$, which also answers questions asked by Meyer in \cite{meyerquestion}*{p.~3158}.}
The first construction of a non-periodic Delone FQ-set was given by Kurasov and Sarnak \cite{kurasovsarnak} for $d=1$. They provided a general construction of FQ-sets in $\R$ using multivariate stable (Lee--Yang) polynomials.  
We call $p\in \C[z_1, \hdots, z_n]$ a Lee--Yang polynomial if $p(z)\neq 0$ whenever $z=(z_1, \hdots, z_n)\in \C^n$ has $|z_i|<1$ for all $i$ or $|z_i|>1$ for all $i$. 
{Lee--Yang polynomials first appeared in the work of Lee and Yang \cite{lee1952statistical} in which they showed that certain univariate polynomials\footnote{{The grand partition functions of two dimensional Ising models}} have all their roots on the circle, by showing that each of those is the restriction $s\mapsto p(s,s,\ldots,s)$ of a Lee--Yang polynomial $p$.
More generally, for every Lee--Yang polynomial $p$ and every} positive vector  $ \ell \in \R_+^n$ the trigonometric polynomial 
\begin{equation*}
f(x) = p(\exp(2\pi i \ell x)) \quad \text{ with } \quad x\in \C
\end{equation*} has all of its roots in $\R$, where $\exp\colon\C^{n}\to (\C^{*})^{n}$ is taken entry-wise. That is, $\exp(2\pi i \ell x)=(e^{2\pi i \ell_{1}x},e^{2\pi i \ell_{2}x},\ldots,e^{2\pi i \ell_{n}x})$.
They show that when $p(z)$ and $\nabla p(z)$ have no common zeros in the torus $\T^n = \{z\in \C^n\mid |z_i|=1 \text{ for all }i\}$, the set $\Lambda = \{x\in \R \mid f(x)=0\}$ is a Delone FQ-set. 
Moreover, under mild conditions on $\ell$ and $p$, the set $\Lambda$ intersects every arithmetic progression in at most finitely many points. Properties of this construction are further investigated in \cite{AlonVinzant}. 

It turns out that every {FQ-set} in $\R$ arises from this construction. Namely, Olevskii and Ulanovskii showed that every FQ-set $\Lambda\subseteq\R$
is the zero set of a real-rooted  trigonometric polynomial \cite{olevskii2020fourier} and  Alon, Cohen and Vinzant showed that every real-rooted trigonometric polynomial is the restriction of a Lee--Yang polynomial
 \cite{AlonCohenVinzant}. Arbitrary one-dimensional {FQ-measures} are less well-understood.
 
The Kurasov--Sarnak construction can be seen geometrically as follows. 
The vector $\ell\in \R_+^n$ defines a map from $\C$ to $\C^n$ given by $x\mapsto \exp(2\pi i \ell x)$ that maps $\R$ into $\T^n$. 
When the entries of $\ell$ are linearly independent over $\Q$, the image of $\R$ under this map 
is dense in $\T^n$. 
The FQ-set $\Lambda$ is the set of $x\in \R$ whose image under this map belongs to the 
algebraic variety $\{z\in \T^n \mid p(z)=0\}$.

In this paper we generalize this geometric construction by considering algebraic varieties of arbitrary 
codimension $d$ and maps $\C^d\to \C^n$ given by $x\mapsto  \exp(2\pi iLx)$, where $L$ is a $n\times d$ matrix with real entries and positive $d\times d$ minors. 
The higher codimensional analogue of the zero set of a Lee--Yang polynomial will be a Lee--Yang variety, defined in \Cref{def:leeyang}.  {The intimate relation between real-rootedness and Lee--Yang polynomials extends to Lee--Yang varieties.
Namely,} for every Lee--Yang variety $X$ of codimension $d$, $\exp(2\pi iLx)\in X$ implies that $x\in \R^d$ and analogously to the Kurasov--Sarnak construction, we 
consider the set $\Lambda = \{x\in \R^d \mid \exp(2\pi iLx)\in X\}$.
Similar notions of ``real-rootedness'' for algebraic varieties have been studied in 
\cites{SV18, KV19, KS20a, KS20b, RVY21}. This existing theory enables us to provide a large collection of Lee--Yang varieties and resulting Delone FQ-sets with various properties in $\R^d$ for arbitrary $d\in\N$. {Our proofs are mainly based on the interplay between real algebraic geometry and discrete geometry, rather than on complex analysis.}

If $\Lambda\subseteq\R$ is an FQ-set, then by \cite{lev2015quasicrystals} it is the zero set of a of a trigonometric polynomial in one complex variable, all of whose zeroes are real.
Similarly, the FQ-sets we construct in this paper are solution sets to systems of multivariate trigonometric equations in several complex variables, all of whose solutions are real. More precisely, if our Lee--Yang variety is defined by equations $p_1(z)=0, \hdots, p_m(z)=0$ where $p_j(z)\in \C[z_1, \hdots, z_n]$, then we obtain the system of trigonometric equations
\[
p_1(\exp(2\pi iLx))=0, \hdots, p_m(\exp(2\pi iLx))=0
\]
all of whose solutions in $\C^d$ lie in $\R^d$. The FQ-set $\Lambda$ is the set of solutions.

\begin{ex}\label{ex:runningRational}
Consider the algebraic curve $X\subseteq \C^3$ defined by equations 
  \begin{align*}
 p_1(z) =    (1 - 2 i) - z_1 - z_2 + (1 + 2 i) z_1z_2 &=0,\\ 
 p_2(z) =  1 - (1 +  i) z_1 - (1-  i) z_3+  z_1z_3&=0.
  \end{align*}
  This curve can be parametrized (up to closure) by the map 
  \begin{equation*}
   t\mapsto \left(\frac{-1 + i + t}{-1 - i + t},  \frac{-i + t}{i + t}, \frac{1 + i + t}{1 - i + t} \right).   
  \end{equation*}
Take the linear map $x=(x_{1},x_{2})\mapsto Lx = (x_1, x_2, -\sqrt{2}x_1 + \sqrt{3}x_2)$, given by a $3\times 2$ matrix $L$ whose $2\times 2$ minors are $1$, $\sqrt{2}$, and $\sqrt{3}$. 
We will see that a compactification of $X$ is a Lee--Yang variety and that $\Lambda=\{x\in\R^2\mid \exp(2\pi iLx)\in X\}$ is a Delone FQ-set.
This set $\Lambda$ consists of the solutions to the two trigonometric equations 
  \begin{align*} 
-4 \sin (\pi \theta_1) \cos (\pi  \theta_2)-4 \sin (\pi  \theta_2) (\sin (\pi  \theta_1)+\cos (\pi \theta_1)) &=0,\\
 2 \sin (\pi  \theta_1) \cos (\pi  \theta_3)-2 \sin (\pi  \theta_3)  (2 \sin (\pi \theta_1)+\cos (\pi \theta_1))&=0.
  \end{align*}
where $(\theta_1, \theta_2, \theta_3)  = (x_1, x_2, -\sqrt{2}x_1 + \sqrt{3}x_2)$.
The left hand side of these equations are $e^{-\pi i (\theta_1+\theta_2)} p_1(\exp({2\pi i \theta}))$ and $e^{-\pi i (\theta_1+\theta_3)} p_2(\exp(2\pi i \theta))$.
The density of $\Lambda$ is
    \begin{equation*}
		\lim_{R\to\infty}\frac{\left|\Lambda\cap B_{R}\right|}{\Vol(B_{R})}= 1+\sqrt{2}+\sqrt{3}.
	\end{equation*}
The spectrum of $\Lambda$, as defined in \Cref{def:FQ}, is contained in the discrete set  
\[\Lambda' = \left\{(k_1-\sqrt{2}k_3, k_2+\sqrt{3}k_3) \mid k\in \Z^3, {\rm sign}(k)\not\in \{(+,-,+), (-,+,-)\}\right\}, \]
where ${\rm sign}(k)$ records the signs in $\{-, 0,+\}$ of the coordinates of $k$. 
As shown in \Cref{ex:PolyComplex}, the number of points in $\Lambda'\cap[-R,R]^2$ is 
$\frac{1}{9}(10 \sqrt{2} + 3\sqrt{3}) R^3 + O(R^2)$ as $R\to \infty$. {We will also see that there is no affine line $W \subsetneq \R^2$ such that $W\cap\Lambda$ contains a set isometric to an FQ-set in $\R$.}
The torus points of the curve $X$, the FQ-set $\Lambda$, and set $\Lambda'$ are shown in \Cref{fig:RationalCurveIntro}.
\end{ex}
Meyer in \cites{Meyer2023-2,Meyer2023} and de Courcy-Ireland and Kurasov \cite{KudCI} have used real-rooted systems of trigonometric polynomials to construct FQ-measures. They obtain weights that are in general not integer valued. In particular, their FQ-measures are not counting measures of subsets of $\R^d$. In a recent preprint \cite{lawton2024fourier}, Lawton and Tsikh construct FQ-sets $\Lambda\subseteq\R^d$ using an approach that is similar to ours in that $\Lambda$ is also the solution set to a system of multivariate trigonometric equations with only real solutions. We prove that every Delone FQ-set arising from their method can also be obtained from our construction (see \Cref{sec:reduct}). On the other hand, their method only applies to systems of trigonometric equations where the number of equations equals the number of variables while our construction is not subject to this restriction.

\subsection{Notations and conventions} Throughout the paper, we use $d,n\in \N$ and $c=n-d\ge 0$ as dimensions.  The Fourier quasicrystals in this paper will be subsets of  $d$-dimensional Euclidean space, constructed using algebraic varieties of dimension $c=n-d\ge 0 $ in an ambient space of dimension $n\ge d$.

We consider $\C^*=\C\setminus\{0\}$ as a subset of the projective line $\pp^1=\C\cup\{\infty\}$. 
For a subset $X\subseteq(\pp^1)^n$ we denote $X(\T)=X\cap\T^n$, where  
	\begin{equation*}
		\T^n =\{z\in(\C^*)^n\mid |z_1|=\cdots=|z_n|=1\}.
	\end{equation*}
We use entrywise functions, for $z=(z_{1},\ldots,z_{n})$
	\[\exp:\C^{n}\to (\C^{*})^{n}\quad \exp(z)=(e^{z_{1}},\ldots,e^{z_n}),\]
	and similarly, $\log|z|=(\log|z_1|,\ldots,\log|z_n|)$, with the convention $\log|0|:=-\infty$ and $\log|\infty|:=\infty$, so that $\log|\cdot|$ is a map from  $(\pp^1)^{n}\to(\R\cup\{\pm\infty\})^{n}$.
	
	We denote $[n]:=\{1,2,\ldots,n\}$ and $\binom{[n]}{d}=\{I\subseteq [n]\mid\  |I|=d\}$ the collection of $d$-tuples of distinct indices in $[n]$.
Given $I\subseteq[n]$, a set $S$ and the product $S^{n}$, the canonical projection $\pi_I\colon S^{n} \to S^I$ is given by $\pi_I(s_1, \hdots, s_n) = (s_i)_{i\in I}$, where $S^I$ denotes the set of tuples indexed by $I$.  

We use $R^{n\times d}$ to denote the set of $n\times d$ matrices with entries in a ring $R$. 
Given a matrix $M\in R^{n\times d}$ (or $R^{d\times n}$) and $I\in\binom{[n]}{d}$, we use $M_I$ to denote the determinant of the $d\times d$ submatrix of $M$ obtained by restricting to the rows (or columns) indexed by $I$. We use $M^t$ to denote the transpose of $M$. By convention, all vectors are assumed to be column vectors. 

Finally, for a function $f$, we use the convention that its Fourier transform is  
	\begin{equation*}
		\hat{f}(y) = \int_{\mathbb{R}^n} f(x)e^{- 2\pi i \langle x,y\rangle}\,dx.
	\end{equation*}

\section{The construction and main results}
Here we present a rather general construction of multidimensional Fourier quasicrystals in detail
and formulate main theorems describing their properties.

\subsection{General construction and main theorem}
Our Fourier quasicrystals will result from the following construction. The input to our construction is:
\begin{itemize}
	\item an algebraic variety $X\subseteq(\pp^1)^n$ of dimension $ c = n-d$, and 
	\item a real matrix $L\in\R^{n\times d} $.
\end{itemize}
We consider the points whose image under the map $x\mapsto \exp(2\pi iLx)$ belongs to $X$:
\begin{equation} \label{lambda}
	\Lambda=\Lambda(X,L): =\{x\in\C^d\mid \exp(2\pi iLx)\in X\}.
\end{equation}
If all $d\times d$ minors of $L$ are positive and $X$ is a strict Lee--Yang variety in the sense of \Cref{def:strictLY} below, then $\Lambda\subseteq\R^d$, and our main result states that $\Lambda$ is a $d$-dimensional Fourier quasicrystal. In \Cref{ssec:lyex} we provide a method for explicit construction of one-dimensional strict Lee--Yang varieties in $(\pp)^{n}$, hence $d=n-1$ dimensional Fourier quasicrystals, for every $n\ge 2$. In \Cref{sec:reduct}, we will see that many algebraic varieties can be transformed to a strict Lee--Yang variety by a change of coordinates. Strict Lee--Yang varieties are defined as follows.

\begin{Def}\label{def:var}
		For a vector $b=(b_1,\ldots,b_n)\in\left(\R\cup\{\pm\infty\}\right)^{n}$ we let $\var(b)$ be the number of sign changes in the sequence $b_1,\ldots,b_n$ after discarding the zeroes.
\end{Def}

\begin{Def}[Strict Lee--Yang variety]\label{def:strictLY}
	An equidimensional\footnote{An algebraic variety is equidimensional if it is a finite union of irreducible algebraic varieties of the same dimension.} closed algebraic subvariety $X$ in $ (\pp^1)^n $, of dimension $ c=n-d\ge 0 $, is called a \emph{strict Lee--Yang variety} if it satisfies the following three conditions:
	\begin{enumerate}
		\item $X$ is invariant under the coordinate-wise involution $z\mapsto 1/\overline{z}$;
		\item for every $z\in X$ we either have $z\in X(\T)=X\cap\T^{n}$ or $\var(\log|z|)\geq d$;
        \item $X(\T)$ is contained in the smooth part of $X$.
	\end{enumerate}
\end{Def}
Now we can formulate our main result:
\begin{thm}[Construction of Delone FQ]\label{thm:main1}
	Let $X\subseteq(\pp^1)^n$ be a strict Lee--Yang variety of dimension $c=n-d$ and let $L\in\R^{n\times d}$ be a real matrix all of whose $d\times d$ minors are positive. Then the set $\Lambda=\Lambda(X,L)$ is real, $\Lambda\subseteq\R^{d}$, it is a Delone set, and it is a Fourier quasicrystal. Its spectrum is contained in the discrete set
	\begin{equation}\label{eq: lambda'}
		 			\Lambda'=\Lambda'(L)=\{L^{t}k\mid k\in\Z^{n},\ \var(k)<d\}.
				\end{equation} 
\end{thm}
The proof of \Cref{thm:main1} is given at the end of \Cref{sec:measuresandrestrictions}. Note that neither \Cref{def:strictLY} nor the positivity of the minors of $L$ are invariant under arbitrary permutations of the coordinates.
\begin{rem}[Real-rooted system of trigonometric equations]
	The statement that $\Lambda(X,L)$ is real is equivalent to the system of trigonometric equations $$ \C^{d}\ni x\mapsto P(\exp(2\pi iLx)) $$ having only real solutions, where $ \C^{n}\ni z\mapsto P(z) $ is a system of (possibly more than $ d $) polynomial equations that determine $ X$. An independent recent work \cite{lawton2024fourier} provided a construction of $d$-dimensional FQs from sufficiently generic real-rooted systems of $d$ trigonometric equations. In \Cref{sec:reduct} we show that every such system is obtained via our construction if the FQ is Delone.
\end{rem}
\subsection{Properties of the constructed Fourier quasicrystals} Similarly to the one-dimensional case \cites{kurasovsarnak,AlonVinzant}, under mild conditions on $X$ and $L$ the set $\Lambda(X,L)$ is far from being periodic. Let $\dim_{\Q}(A)$ of a set $A\subseteq\R^{d}$ denote the dimension, as a rational vector space, of the linear span of its elements over $\Q$. For example, the set $\{1,\pi,2+\pi\}\subseteq \R$ has $\dim_{\Q}=2$. A lattice in $\R^{d}$ has $\dim_{\Q}\le d$, and a set contained in the projection of a lattice in $\R^{d}\times\R^{N}$ to $\R^{d}$ has $\dim_{\Q}\le d+N$.  
\begin{thm}[Quantitative non-periodicity]\label{thm: non-periodicity}
Let $\Lambda=\Lambda(X,L)$ be defined as in \Cref{thm:main1}, and further assume that 
\begin{enumerate}
	\item the $d\times d$ minors of $L$ are $\Q$-linearly independent, and
	\item $X$ is irreducible and $X\cap (\C^{*})^{n}$ is not a coset of an algebraic subtorus\footnote{A coset of a $c$-dimensional algebraic subtorus of $(\C^{*})^{n}$ is a set $\{\exp(x+My)\mid y\in\C^{c}\}$ with $x\in\C^{n}$ and $M\in\Z^{n\times c}$ of full rank.} of $(\C^{*})^{n}$.
\end{enumerate}
Then $\dim_{\Q}(\Lambda)=\infty$, and there is a constant $r=r(X)>0$, that only depends on $n$ and the degree of $X$, such that the following uniform bounds hold 
\[|\Lambda\cap A|<r^{m+1},\ \text{  for every set  }\ A\subseteq\R^{d}\ \text{with}\ \dim_{\Q}(A)=m.\]
In particular, for every $a,b\in\R^{d} $ there are at most $r^{3}$ points $a+jb\in\Lambda$ with $j\in\Z$.
\end{thm}
This statement, including an explicit expression for $r$, is a special case of \Cref{thm: non-periodicity2}. 

\begin{thm}[Properties of our construction]\label{thm:main2} Let $\Lambda=\Lambda(X,L)\subseteq\R^{d}$ and $\Lambda'$ as in \Cref{thm:main1}.
Let $(c_{\xi})_{\xi\in\Lambda'}$ be the coefficients such that
 for every Schwartz function $f\in\mathcal{S}(\R^{d})$, we have $\sum_{x\in\Lambda}\hat{f}(x)=  \sum_{\xi\in\Lambda'}c_{\xi}f(\xi)$.
Denote  the $d\times d$ minors of $L$ by $(L_{I})_{I\in \binom{[n]}{d}}$. Then:
	\begin{enumerate}
		\item The set $\Lambda$ is Bohr almost periodic with density equal to 
		\begin{equation*}
			\lim_{R\to\infty}\frac{\left|\Lambda\cap B_{R}\right|}{\Vol(B_{R})}= c_{0}=\sum_{I\in\binom{[n]}{d}} L_I\cdot d_{I}
		\end{equation*}
		where $d({X})=(d_J)_{J\in\binom{[n]}{d}}$ is the multidegree of the strict Lee--Yang variety ${X}$, and $B_{R}$ is the ball of radius $R$ around the origin.
		\item The set $\Lambda'=\Lambda'(L)$ as in \eqref{eq: lambda'} has rational dimension $\dim_{\Q}(\Lambda')\le n$ and growth rate $ \left|\Lambda'\cap B_{R}\right|=C R^{n}+O(R^{n-1}) $, for $ R \rightarrow \infty $ and some $ C \geq 0 $. If the rows of $L$ are linearly independent over $\Q$, then $C>0$.
		\item There are coefficients $c_{L,k}\in\C$, for $k\in\Z^n$, such that for all $\xi\in\Lambda'$
	\[c_{\xi}=\sum_{k\in\Z^{n},\ L^{t}k=\xi}c_{L,k}\]where almost all summands, and at least those with $\var(k)\geq d$, in the sum on the right-hand side are zero.
  The coefficients $c_{L,k}$ are bounded, $|c_{L,k}|\le c_{0}$, and depend linearly on $(L_{I})_{I\in \binom{[n]}{d}}$. Specifically, there is a collection of measures $\bm_{J}$ on $ X(\T)$, indexed by $J\in\binom{[n]}{n-d}$, such that 
		\[c_{L,k}=\sum_{I\in\binom{[n]}{d}}L_{I}\widehat{\bm}_{[n]\setminus I}(k).\]
	\end{enumerate}
\end{thm}
The measures $\bm_{J}$ and their Fourier  transform are defined in \Cref{sec:measuresandrestrictions} where we also prove \Cref{thm:main2}, except for the statement that the density of $\Lambda$ is given by $c_0$, which is proven independently in \Cref{thm: HU and AC}. One possibility to construct a multidimensional FQ is by taking the product of one-dimensional FQs, or a rotation and translation thereof.
In \Cref{Sec:non-triv} we give sufficient criteria on $X$ and $L$ that guarantee that $\Lambda(X,L)$ is not of this form.
{More precisely, we prove that, under mild assumptions, our Fourier quasicrystals are genuinely high-dimensional in the sense that they do not intersect any affine line in a set that is isometric to a one dimensional Fourier quasicrystal.} 
\begin{thm}
Let $\Lambda=\Lambda(X,L)\subseteq\R^{d}$ be as in \Cref{thm:main1} and assume in addition that:  
\begin{enumerate}
	\item The variety $X$ is a curve, i.e. $\dim(X)=1$, and so $d=n-1$. 
	\item $X$ is irreducible and $X(\T)=X\cap\T^{n}$ is not contained in a coset of a proper subtorus.
	\item The $d\times d$ minors of $L$ are $\Q$-linearly independent.
\end{enumerate}
Then, there is no affine line $W\subsetneq\R^d$ such that $W\cap\Lambda$ contains a set isometric to a Fourier quasicrystal in $\R$. 
\end{thm}
In \Cref{sec:concreteexample} we construct explicit examples of Fourier quasicrystals that have all the properties described in the above theorems.

\subsection{Diffraction and hyperuniformity} 
Next, we provide a rather straight forward application, showing that every discrete set $\Lambda\subseteq\R^{d}$ whose counting measure is a Fourier quasicrystal is \emph{hyperuniform}, and in fact \emph{stealthy hyperuniform}. See \cite{torquato2018hyperuniform} for a thorough review of hyperuniform materials and the physical implications of a hyperuniform state of matter. See also \cites{bjorklund2023hyperuniformity,ouguz2017hyperuniformity} where it is shown that crystals (i.e.~lattices) and certain quasicrystals\footnote{not to be confused with Fourier quasicrystals} are stealthy hyperuniform and hyperuniform respectively. What we show next suggests that the physical properties of Fourier quasicrystals should be closer to crystals than to quasicrystals. In the physics literature there are two equivalent definitions of hyperuniform sets, and in \cite{bjorklund2023hyperuniformity}*{Theorem~1.1} this equivalence is proven rigorously. 
\begin{Def}
	A discrete set $ \Lambda\subseteq \R^{d} $ is called \emph{hyperuniform} if one of the following equivalent conditions hold.
	\begin{enumerate}
		\item Hyperuniformity via physical space: Suppose $ \Lambda $ has density $C>0$, and let $ N_{R}(x)=|\Lambda\cap B_{R}(x)| $ for $ x\in\R^{d} $. Then $\Lambda$ is called \emph{hyperuniform} if the variance of $ \frac{N_{R}(x)}{\Vol(B_{R})} $ over $x\in B_{R}(0)$ goes to zero as $ R\to \infty $. That is, 
		\[\lim_{R\to\infty}\frac{1}{\Vol(B_{R})}\int_{x\in B_{R}}\left|\frac{N_{R}(x)}{\Vol(B_{R})}-C\right|^{2}dx=0.\]
		\item Hyperuniformity via Fourier space: Suppose that the following limit of tempered measures exists (in the vague topology) and is unique, 
		\[\gamma=\lim_{R\to\infty}\frac{1}{\Vol(B_{R})}\sum_{x,y\in\Lambda\cap B_{R}}\delta_{x-y}.\]
		It is called the \emph{auto-correlation} of $ \Lambda $. Its distributional Fourier transform $ \widehat{\gamma} $ is called the \emph{diffraction measure} of $ \Lambda $. In such case $ \Lambda $ is hyperuniform if 
		\[\lim_{\epsilon\to 0}\frac{\widehat{\gamma}(B_{\epsilon}\setminus\{0\})}{\epsilon^{d}}=0.\]
		Another related quantity in the physics literature is the \emph{structure factor} $S(\xi)$, which is the density of the diffraction measure $\widehat{\gamma}$, namely $d\widehat{\gamma}=S(\xi)d\xi$, and is well defined at $\xi$ only if $\widehat{\gamma}$ is absolutely continuous in a neighborhood of $\xi$. Hyperuniformity is sometimes defined by $S(\xi)\to 0$ when $0\ne \xi\to 0$.
	\end{enumerate}
	A hyperuniform set $ \Lambda $ is called \emph{stealthy hyperuniform} if $ \{0\} $ is an isolated point in the support of $ \widehat{\gamma} $.
\end{Def}

\begin{thm}\label{thm: HU and AC}
	If a discrete set $\Lambda\subseteq \R^{d} $ is a Fourier quasicrystal, then it is stealthy hyperuniform.
	In more details, suppose that $ \Lambda$ is a Fourier quasicrystal with spectrum $ \Lambda' $ and Fourier coefficients $ (c_{\xi})_{\xi\in\Lambda'} $.
Then:
	\begin{enumerate}
		\item The auto-correlation $ \gamma $ of $ \Lambda $ exists and the diffraction measure is equal to $ \widehat{\gamma}=\sum_{\xi\in\Lambda'}|c_{\xi}|^{2}\delta_{\xi} $. In particular the support of the diffraction measure is the discrete set $\Lambda'$. 
		\item The Fourier coefficient $c_{0}$ is real, positive, and is the density of $\Lambda$, with  	\[\sup_{x\in\R^{d}}\left|N_{R}(x)-c_{0}\Vol(B_{R})\right|=O(R^{d-1}) .\]		
	\end{enumerate} 
\end{thm}
The proof is in \Cref{sec: autocorrelation}.  

\subsection{Open questions}
We conclude this section with the following questions:
\begin{enumerate}
	\item \emph{Can all Delone FQ-sets in every dimension be obtained via our construction, namely as $\Lambda(X,L)$ with strict Lee--Yang $X$ and $L$ with positive minors?}
	\item \emph{If $X$ is Lee--Yang but $X(\T)$ is not smooth, is it still true that $\Lambda(X,L)$ is an FQ-set, possibly with multiplicities? If so, can every FQ-set be obtained in this way?}
	\item \emph{If there is a negative answer to question (1) or (2), can every (Delone) FQ-set be obtained from some suitable modification of our construction, namely as $\Lambda(X,L)$ with an algebraic variety $X$ and a matrix $L$?}
	\item \emph{Can all FQ-sets in every dimension be obtained as the common zeros of a real-rooted system of trigonometric equations?}
	\item \emph{Can the zero set of every real-rooted system of trigonometric equations be obtained as $\Lambda(X,L)$ with $X$ Lee--Yang and $L$ with positive minors?}
\end{enumerate}
Positive answers for these questions are known in some special cases. All these questions have positive answers for the one dimensional case. Furthermore, in \Cref{sec:reduct} we show that the answer to the first question is positive for Delone FQ-sets that arise from the construction in \cite{lawton2024fourier}.

\section{The positive Grassmannian}
We denote by ${\rm Gr}(d,n)$ the Grassmannian of $d$-dimensional linear subspaces of $\R^n$. 
Any subspace $V\in {\rm Gr}(d,n)$ can be expressed as the column span of an $n\times d$ matrix $L$, whose columns $v_1, \hdots, v_d$ form a basis of $V$. For $I\subseteq [n]$ with $|I|=d$, let $L_I$ denote the determinant of the $d\times d$ submatrix of $L$ obtained by restricting to rows indexed by $I$.  We use these minors to specify $V$ as follows. 

Let $1 \leq d \leq n$ and take $e_1, \ldots, e_n $ to be the standard basis for $\R^n$.
For each subset $I =\{i_1, \ldots, i_d\}$ of $[n]$ with $i_1 < \cdots < i_d$, 
we denote $e_I=e_{i_1} \wedge \cdots \wedge e_{i_d}$.
The collection of wedge products $\{e_I\mid I \in \binom{[n]}{d}\}$ forms a basis of $\bigwedge^d \R^n$.
For a linear space $V= {\rm span}\{v_1, \hdots, v_d\}$ in $\Gr(d,n)$, 
we can express $v_1 \wedge \hdots \wedge v_d$ as $\sum_{I\in \binom{[n]}{d}} L_{I}e_I$ in $\bigwedge^d\R^n$ where $L$ is the $n\times d$ matrix with columns $v_1, \hdots, v_d$. 
The coefficients, $(L_{I})_{I\in \binom{n}{d}}$, are known as the \emph{Pl\"ucker coordinates} of $V$ and 
are independent of the basis $\{v_1, \hdots, v_d\}$ of $V$ up to global scaling. 
Thus the map sending ${\rm span}\{v_1, \hdots, v_d\} $ to $[v_1 \wedge \hdots \wedge v_d]\in  \pp(\bigwedge^d \R^n)$ is well defined. 
This is the \emph{Pl\"ucker embedding} of the Grassmannian $\Gr(d,n)$ into $ \pp(\bigwedge^d \R^n)\cong \pp^{\binom{n}{d}-1}(\R)$.
Finally, we will define $s(I)\in\{\pm1\}$ such that $e_I\wedge e_{[n]\setminus I}=s(I)\cdot e_{[n]}$.

Let ${\rm Gr}_{+}(d,n)$ resp. ${\rm Gr}_{\geq}(d,n)$ denote the positive and nonnegative Grassmannian, respectively. 
This is the  collection of $d$-dimensional subspaces of $\R^n$ that can be written as the column span of an $n\times d$ matrix all of whose  $d\times d$ minors are strictly positive or nonnegative, respectively. 
In order to formulate a characterization of points lying on an element of the positive Grassmannian we recall and extend \Cref{def:var}.

\begin{Def}\label{def:varbar}
		For a vector $b=(b_1,\ldots,b_n)\in\left(\R\cup\{\pm\infty\}\right)^{n}$ we let $\var(b)$ be the number of sign changes in the sequence $b_1,\ldots,b_n$ after discarding the zeroes. We further let $\varbar(b)$ be the number of sign changes in the sequence $b_1,\ldots,b_n$ where the zeroes are assigned signs that maximize the number of sign changes.  
\end{Def}

\begin{lem}[\cite{Karp}*{Lemma~4.1}]\label{lem:Karp}
  Let $v \in \R^n\setminus \{0\}$.
  \begin{enumerate}
   \item There exists a linear subspace in ${\rm Gr}_{+}(d,n)$ containing $v$ iff $\overline{\var}(v) < d$.
   \item There exists a linear subspace in ${\rm Gr}_{\geq}(d,n)$ containing $v$ iff $\var(v) < d$.
  \end{enumerate}
\end{lem}

\begin{thm}[\cites{GK, KWZ}]\label{thm:GK}
  For a $d$-dimensional linear subspace $V\subseteq \R^n$, we have
  \begin{enumerate}
  \item $V \in \Gr_{\geq}(d,n) \Leftrightarrow \var(v) < d \text{ for all }v \in V \backslash\{0\} \Leftrightarrow \varbar(w) \geq d \text{ for all } w \in V^\perp \backslash\{0\}.$
  \item $V \in \Gr_{+}(d,n) \  \Leftrightarrow \varbar(v) < d \text{ for all }v \in V  \backslash\{0\} \Leftrightarrow \var(w) \geq d \text{ for all } w \in V^\perp \backslash\{0\}.$
  \end{enumerate}
\end{thm}

\begin{cor}\label{cor:intsgr}
    If $V\in \Gr_{\geq}(d,n)$ and $W\in \Gr_{+}(d,n)$, then $V\cap W^\perp=\{0\}$ and $V^\perp\cap W=\{0\}$.
\end{cor}

We conclude this section with a statement that will enable us to prove that the support of the Fourier transform of our Fourier quasicrystals is discrete.
For positive integers $d\leq n$, define
\begin{equation}\label{eq:orthants}
\R^{n}_{\var<d} = \{ y\in \R^n \mid \var(y)< d\} \ \ \ \text{ and  } \ \ \ \Z^{n}_{\var<d}= \R^{n}_{\var<d}\cap \Z^n.    
\end{equation}
This is a union of orthants in $\R^n$ and $\Z^n$, respectively. 

\begin{lem} \label{lem:bounded}
Let $L\in \R^{n\times d}$ be a $n\times d$ matrix all of whose $d\times d$ minors are positive. 
Then $\{y\in \R^{n}_{\var<d}\mid L^ty\in [-1,1]^d\}$
is closed, bounded, and has nonempty interior in the Euclidean topology on $\R^n$. 
\end{lem}
\begin{proof}
First, note that $\R^{n}_{\var<d}$ and $\{y \mid L^ty\in [-1,1]^d\}$ are closed subsets of $\R^n$, so their intersection is also closed. 
To see that its interior is nonempty, consider a point $y = \lambda (1,\hdots, 1)^t$. For sufficiently small $\lambda>0$, 
$L^ty= \lambda L^t(1,\hdots, 1)^t$ belongs to $(-1,1)^d$. Then for every $z$ in a sufficiently small neighborhood of $y$ in $\R^n$, $z\in \R_{>0}^n$ and $L^tz\in(-1,1)^d$. Therefore $y$ belongs to the interior of this set. 

To see that it is bounded, let  $W$ the column span of $L$. By assumption we have $W\in {\rm Gr}_{+}(d,n)$.
Suppose, for the sake of contradiction, that $\{y\in \R^{n}_{\var<d}\mid L^ty\in [-1,1]^d\}$ is unbounded. 
Note that $\R^{n}_{\var<d}$ is a union of orthants in $\R^n$. It follows that 
the set $\{y\in \R^{n}_{\var<d}\mid L^ty\in [-1,1]^d\}$ is a union of finitely many polyhedra. 
Since it is unbounded, it therefore must contain a ray 
$\{y+\lambda v\mid \lambda \in \R_{\geq 0}\}$ for some $y,v\in \R^n$ with $v\neq 0$. 
By assumption, $L^t(y+\lambda v) = L^ty+\lambda L^tv$ belongs to $[-1,1]^d$ for all $\lambda\geq 0$. 
As $[-1,1]^d$ is bounded, we must have $L^tv=0$ and  thus $v\in W^{\perp}\setminus\{0\}$. 
By \Cref{thm:GK}, it follows that $\var(v)\geq d$. 
Then for sufficiently large $\lambda\geq 0$, 
$\var(y+\lambda v) \geq d$. Indeed, since $\var$ does not count zero coordinates towards sign changes, the entries of $y$ can only increase $\var(y+\lambda v)$ for sufficiently large $\lambda$. This contradicts $y+\lambda v\in \R^{n}_{\var<d}$ for all $\lambda \geq 0$.
\end{proof}

\begin{cor}\label{cor:imagediscrete}
Let $L\in \R^{n\times d}$ be an $n\times d$ matrix all of whose $d\times d$ minors are positive. 
The image of $\Z^{n}_{\var<d}$ under the map $k\mapsto L^tk$ is discrete. 
More precisely, for $R>0$, the number points whose image lies in $[-R,R]^d$ is
\[
|\{k\in \Z^{n}_{\var<d} \mid L^tk \in [-R,R]^d\}| = {\rm vol}\cdot R^n + O(R^{n-1})
\]
where ${\rm vol}>0$ is the volume of $\{y\in \R^{n}_{\var<d} \mid L^ty \in [-1,1]^d\}$.
Moreover, if the rows of $L$ are linearly independent over $\Q$, then
\[
|\{L^tk \mid k\in \Z^{n}_{\var<d}\}\cap [-R,R]^d| = {\rm vol}\cdot R^n + O(R^{n-1}).
\]
\end{cor} 

\begin{proof}
Let $P = \{y\in \R^{n}_{\var<d} \mid L^ty \in [-1,1]^d\}$. By \Cref{lem:bounded}, $P$ is closed, bounded, and has nonempty interior. Note that $L^ty$ belongs to $[-R,R]^d$ if and only if $L^t((1/R)y)$ belongs to $[-1,1]^d$ and $\R^{n}_{\var<d}$ is invariant under positive scaling. 
Thus for every $R>0$,  
\[\{y\in \R^{n}_{\var<d} \mid L^ty \in [-R, R]^d\} = R\cdot P = \{R\cdot y \mid y\in P\}. \]
Since $R\cdot P$ is bounded, the number of integer points it contains is finite, showing that 
 the image of $\Z^{n}_{\var<d}$ under the map $k\mapsto L^tk$ is discrete. 

To get the statement on the asymptotic growth, we write $P$ as a union of 
polyhedra $P_\sigma = \{y\in \R^n_\sigma \mid L^ty \in [-1,1]^d\}$ where $\sigma$ 
runs over all vectors in $\{\pm 1\}^n$ with $\var(\sigma)<d$ and $\R^n_\sigma = \{y\in \R^n: \sigma_jy_j\geq 0 \text{ for }j=1, \hdots, n\}$. 
Since $P_{\sigma}$ is bounded, it is a polytope. 
By \cite{BordaThesis}*{Corollary~1.2}, for every polytope $Q\subseteq \R^n$, 
the number of integer points in $R\cdot Q$ is 
\[
|(R\cdot Q)\cap \Z^n| = {\rm vol}(Q)R^n + O(R^{n-1}).
\]
For every $\sigma\neq \sigma'$, the intersection of $P_{\sigma}$ and $P_{\sigma'}$ is a polytope of dimension $\leq n-1$ and so the number of integer points in $R\cdot(P_{\sigma}\cap P_{\sigma'})$
is $O(R^{n-1})$. It follows that 
\begin{align*}
|(R\cdot P)\cap \Z^n| &= \sum_{\sigma}|(R\cdot P_{\sigma})\cap \Z^n| +O(R^{n-1})\\
& = \sum_{\sigma}{\rm vol}(P_{\sigma})R^n +O(R^{n-1}) ={\rm vol}(P)R^n +O(R^{n-1}).
\end{align*}
Since $P$ has nonempty interior, ${\rm vol} = {\rm vol}(P)>0$. 

Finally, if the rows of $L$ are linearly independent over $\Q$, then the map $k\mapsto L^tk$ is injective on $\Z^n$.
The second equality then follows from the first. 
\end{proof}

\begin{ex}[\Cref{ex:runningRational} continued]\label{ex:PolyComplex}
The $2\times 2$ minors of the matrix $L^t = {\small \begin{pmatrix} 1 & 0 & -\sqrt{2} \\ 0 & 1 & \sqrt{3} \end{pmatrix}}$ are all positive. 
\Cref{lem:bounded} then implies that the set 
$$P = \{y\in \R^3\mid \var(y)<2, (y_1 - \sqrt{2}y_3, y_2 +\sqrt{3}y_3)\in [-1,1]^2\}$$
is closed and bounded. We can write $P$ as a union of its intersections, $P_{\sigma}$, with the orthants $\R^3_\sigma$, where $\sigma$ ranges over the six elements of $\{\pm 1\}^3\backslash \{(1,-1,1), (-1,1,-1)\}$. 
See \Cref{fig:polyComplex}.
The volume of $P$ is $\frac{1}{9}(10 \sqrt{2} + 3\sqrt{3}) \approx 2.1487$, so, by \Cref{cor:imagediscrete}, the number of integer points in  $R\cdot P$ is given by 
$$|(R\cdot P)\cap \Z^3|= \frac{1}{9}(10 \sqrt{2} + 3\sqrt{3}) R^3 + O(R^2).$$ 
Since the rows of $L$ are linearly independent over $\Q$, the map $\Z^3 \to \R^2$ given by $k\mapsto L^tk$ 
is injective. Therefore for $\Lambda' = \{L^tk \mid k\in \Z^{n}_{\var<d}\}$, the number of points in 
$\Lambda'\cap [-R,R]^2$ is also given by $\frac{1}{9}(10 \sqrt{2} + 3\sqrt{3}) R^3 + O(R^2)$.
The set $\Lambda'$ appears on the right in \Cref{fig:RationalCurveIntro}.
\end{ex}

\begin{figure}
\begin{center}
\includegraphics[height=1.5in]{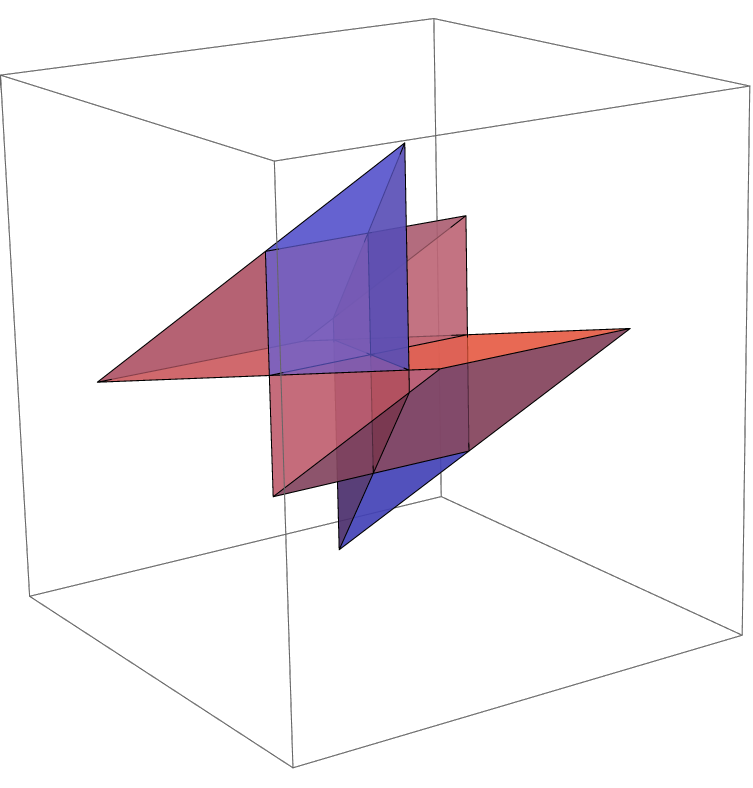} 
\end{center}
\caption{From \Cref{ex:PolyComplex}, the set of $y\in \R^3$ with $\var(y)<2$ and $(y_1 - \sqrt{2}y_3, y_2 +\sqrt{3}y_3)\in [-1,1]^2$. Points with $\var=0$ and $1$ are in blue and red, respectively. 
}
\label{fig:polyComplex}
\end{figure}

\section{Lee--Yang varieties}
In this section we introduce a class of algebraic varieties which we call Lee--Yang varieties. These arise from relaxing condition (2) in \Cref{def:strictLY} of strict Lee--Yang varieties. This section is technical and is partitioned into five subsections. First, we define Lee--Yang varieties and study their connection to real-rootedness. Second, we provide examples of Lee--Yang varieties and explicit construction methods. Third, we discuss the transversality of $x\mapsto(\exp(2\pi i Lx))$ and a Lee--Yang variety, when $L$ has positive minors. Then we focus on strict Lee--Yang varieties and the topology of their intersection with the torus. In the final subsection, we prove that certain integrals over the torus part of a strict Lee--Yang variety vanish. This is crucial for proving that the spectrum of the constructed Fourier quasicrystal $\Lambda(X,L)$ is discrete.

\subsection{Definitions and real-rootedness}  We recall and add to \Cref{def:strictLY}.
\begin{Def}\label{def:leeyang}
	An equidimensional closed subvariety $X$ in $ (\pp^1)^n $ of dimension $ c=n-d $ with $0\le d\le n$ is called a \emph{strict Lee--Yang variety} if it satisfies the following three conditions:
	\begin{enumerate}
		\item[(1)\hspace{2pt}] $X$ is invariant under the coordinate-wise involution $z\mapsto 1/\overline{z}$,
		\item[(2)\hspace{2pt}] for every $z\in X$ we either have $z\in X(\T)=X\cap\T^{n}$ or $\var(\log|z|)\geq d$, and 
        \item[(3)\hspace{2pt}] $X(\T)$ is contained in the smooth part of $X$.
	\end{enumerate}
	We call $X$ a \emph{Lee--Yang variety} if it satisfies the following two conditions:
	\begin{enumerate}
		\item[(1)\hspace{2pt}]  $X$ is invariant under the coordinate-wise involution $z\mapsto 1/\overline{z}$, and
		\item[(2$'$)] for every $z\in X$ we either have $z\in X(\T)$ or $\varbar(\log|z|)\geq d$. 
	\end{enumerate}
\end{Def}

\begin{rem}\label{lem:strictweak}$\,$
\begin{enumerate}
    \item  Every strict Lee--Yang variety is a Lee--Yang variety and one can check that not every Lee--Yang variety is a strict Lee--Yang variety.
    \item Lee--Yang varieties of codimension one are exactly the zeros sets of Lee--Yang polynomials, as considered in \cite{AlonVinzant}*{\S3}, or of stable pairs of polynomials, considered in \cite{kurasovsarnak}*{\S II}.
    \item The smoothness assumption (3) in \Cref{def:leeyang} will ensure that the support of the Fourier quasicrystals, which we will construct, is uniformly discrete, see \Cref{sec:delone}. Property~(2) in \Cref{def:leeyang} makes the definition robust with respect to small perturbations. Indeed, if $z\in(\pp^1)^n$ satisfies $\var(\log|z|)\geq d$, then this is also the case for every $z'$ in a sufficiently small neighbourhood of $z$.
    \item The conclusion of \Cref{thm:main1} is in general false for varieties that satisfy (1), (2) and (3) from \Cref{def:leeyang} but are not equidimensional. Indeed, let $X\subseteq(\pp^1)^n$ be a strict Lee--Yang variety of dimension $c=n-d<n$ and $L\in\R^{n\times d}$ a real matrix all of whose $d\times d$ minors are positive. For every 
           $y\not\in\Lambda(X,L)$
    the closed subset $X\cup\{\exp(2\pi iLy)\}\subseteq(\pp^1)^n$ satisfies (1), (2) and (3) of \Cref{def:leeyang} but the Fourier transform of $\delta_y+\sum_{x\in\Lambda}\delta_x$ is not discrete.
\end{enumerate}
\end{rem}

We now observe some first properties of Lee--Yang varieties.

\begin{lem}\label{lem:translate}
    Let $X\subseteq(\pp^1)^n$ be a closed subvariety. For $a\in\T^n$ we consider the automorphism
    \begin{equation*}
        m_a\colon (\pp^1)^n\to(\pp^1)^n
    \end{equation*}
    which multiplies the $i$-th coordinate by $a_i\in\T$. Then $X$ is a (strict) Lee--Yang variety if and only if $m_a(X)$ is a (strict) Lee--Yang variety.
\end{lem}

\begin{proof}
    This follows directly from the definitions because for every $z\in(\pp^1)^n$ one has $\log|z|=\log|m_a(z)|$ and ${1}/{\overline{m_a(z)}}=m_a({1}/{\overline{z}})$.
\end{proof}

\begin{prop}\label{prop:RR} Let $X\subseteq (\pp^1)^n$ be an equidimensional closed subvariety of codimension $d$ and let $L \in \R^{n\times d}$ be a matrix of full rank. Further let $x\in \C^d$ such that the point $\exp(2\pi i  Lx)$ belongs to $X$.
\begin{enumerate}
 \item If $X$ is a Lee--Yang variety and if all $d\times d$ minors of $L$ are positive, then $x\in \R^d$. 
 \item If $X$ is a strict Lee--Yang variety and if all $d\times d$ minors of $L$ are nonnegative, then $x\in \R^d$. 
\end{enumerate}
\end{prop}

\begin{proof}
We prove part (1).
Because $X$ is a Lee--Yang variety, we either have that $\exp(2\pi i  Lx)\in\T^n$, which means $L \Ima(x)=0$, or we have
\begin{equation*}
  d\leq\varbar(\log|\exp(2\pi i  Lx)|)=\varbar(-2\pi L \Ima(x))=\varbar(L \Ima(x)).
\end{equation*}
On the other hand, because the column span of $L$ is in ${\rm Gr}_+(d,n)$, \Cref{lem:Karp} then also implies that $L \Ima(x)=0$.  Since $L$ is full rank, we find that $\Ima(x) = 0$. 
The proof for part (2) is analogous.
\end{proof}

\begin{notation}\label{not:chim}
  For $M=(m_{ij})_{i,j}\in\Z^{m\times n}$ we consider the map 
 \begin{equation*}
  \chi_M\colon(\C^*)^n\to (\C^*)^m   ,\, (t_1,\ldots,t_n)\mapsto (t_1^{m_{i1}}\cdots t_n^{m_{in}})_{i=1,\ldots,m}.
 \end{equation*}
\end{notation}

\begin{lem}\label{cor:torus-fibered}
 Let $X\subseteq(\pp^1)^n$ be a Lee--Yang variety of dimension $c=n-d$ and $M\in \Z^{c\times n}$ be a matrix of rank $c$ whose (right) kernel belongs to ${\rm Gr}_+(d,n)$. For $X^*=X\cap(\C^*)^n$ we have $\chi_M^{-1}(\T^c)\cap X=X^*(\T)$.
\end{lem}

\begin{proof}
Since $\chi_M$ takes $\T^n$ to $\T^c$, it is clear that $X^*(\T)\subseteq \chi_M^{-1}(\T^c)$. We now show the other inclusion.  
Let $z = \exp(x+iy) \in X$ with $ x,y\in\R^{n} $. Since $X$ is Lee--Yang, either $x=0$, in which case we are done, or $\varbar(x)\geq d$. If $\chi_M(\exp(x+iy))=\exp(Mx+iMy)\in \T^c$ then $x\in \ker(M)$. 
 \Cref{lem:Karp} implies that $\varbar(v)< d$ for every non-zero $v\in\ker(M)$, so we conclude that $x=0$ and $z\in \T^n$. 
\end{proof}

\begin{lem}\label{cor:torus-fiberednonneg}
 Let $X\subseteq(\pp^1)^n$ be a strict Lee--Yang variety of dimension $c=n-d$. For every $I\in\binom{[n]}{c}$ the projection $\pi_I\colon X\to (\pp^1)^I$ satisfies $\pi_I^{-1}(\T^I)=X(\T)$.
\end{lem}
\begin{proof}
 Since $\pi_I$ takes $\T^n$ to $\T^I$, it is clear that $X(\T)\subseteq \pi_I^{-1}(\T^I)$. We now show the other inclusion. Let $z\in X$ such that $\pi_I(z)\in\T^I$ . This means that $\log|z_i|=0$ for all $i\in I$. Because $|I|=n-d$, this implies that $\var(\log|z|)< d$. Property (2) in \Cref{def:leeyang} then implies that $z\in X(\T)$. 
\end{proof}

\begin{prop}\label{prop:smoothdense}
    Let $X\subseteq(\pp^1)^n$ be a Lee--Yang variety of dimension $c=n-d$ and denote by $X_{\rm sm}$ the smooth part of $X$. The closure of $X_{\rm sm}(\T)$ is $X(\T)$.
\end{prop}

\begin{proof}
    Let $M\in\GL_n(\Z)$ such that the right kernel of the matrix consisting of the first $c$ rows of $M$ is in $\Gr_{+}(d,n)$. Let $X_0'$ denote the image of $X\cap(\C^*)^n$ under the map $\chi_M$. Because $M$ is invertible over $\Z$, the map $\chi_M$ is an isomorphism. Therefore, it suffices to show that the closure of $(X_0')_{\rm sm}(\T)$ is $X_0'(\T)$. Letting $X_1'$ be the closure of $X_0'$ in $(\pp^1)^n$, we have that $X_1'(\T)=X_0'(\T)$. Thus, it suffices to show that the closure of $(X_1')_{\rm sm}(\T)$ is $X_1'(\T)$. By \Cref{cor:torus-fibered} and our choice of $M$ the projection
    \begin{equation*}
        \pi\colon X_1'\to(\pp^1)^c
    \end{equation*}
    onto the first $c$ coordinates satisfies $\pi^{-1}(\T^c)=X_1'(\T)$. Let $Z\subsetneq(\pp^1)^c$ be the branch locus of $\pi$. Note that $Z$ contains in particular the image of the singular locus of $X_1'$ under $\pi$.  Let $x\in X_1'(\T)$ and $X_2'$ an irreducible component of $X_1'$ containing $x$. It suffices to show that $x$ is in the closure of $X_2'(\T)\setminus\pi^{-1}(Z)$. Let $f$ be the restriction of $\pi$ to $X_2'$. We have that $f^{-1}(\T^c)=X_2'(\T)$. This implies that elements of $\T^c$ have finite fiber under $f$. Therefore, the map $f$ is generically finite. Because, by equidimensionality, $\dim(X_2')=c$ and because $f$ is proper, this implies that $f$ is surjective. By semicontinuity of fiber dimensions, there exists an open subset $U\subseteq(\pp^1)^c$ which contains $\T^c$ such that every element of $U$ has finite fiber under $f$. Letting $V=f^{-1}(U)$ it suffices to show that $x$ is in the closure of $V(\T)\setminus\pi^{-1}(Z)$. The restriction of $f$ to $V$ is finite because it is proper with finite fibers. This implies that every semialgebraic continuous path $\alpha\colon(0,1)\to U(\T)$ with $\lim_{t\to0}\alpha(t)=f(x)$ lifts to $\deg(f)$ distinct semialgebraic paths in $V(\T)$. Because there are, counted with multiplicity, at most $\deg(f)$ many preimages of $f(x)$ under $f$, at least one of these lifts must converge to $x$. This shows the claim.
\end{proof}

\subsection{Examples}\label{ssec:lyex}
We first consider zero dimensional (strict) Lee--Yang varieties.

\begin{lem}\label{lem:zerodim}
    Let $X\subseteq(\pp^1)^n$ be a finite subset. The following are equivalent:
    \begin{enumerate}
        \item $X\subseteq\T^n$,
        \item $X$ is a Lee--Yang variety,
        \item $X$ is a strict Lee--Yang variety.
    \end{enumerate}
\end{lem}

\begin{proof}
    This immediately follows from the fact that $\var(a),\varbar(a)<n$ for every $a$ of length $n$.
\end{proof}
Next we provide an explicit way of constructing strict Lee--Yang varieties of dimension one and arbitrary high codimension.

\begin{Def}
  A \emph{separating curve} is a smooth irreducible real projective curve $X$ such that $X(\C)\setminus X(\R)$ has two connected components.
\end{Def}

\begin{ex}
  For example $\pp^1$ is a separating curve, the two connected components of $\pp^1(\C)\setminus \pp^1(\R)$ being the (open) upper and lower half-plane $H_+$ and $H_-$. More generally, if the smooth irreducible real projective curve $X$ has genus $g$ and $X(\R)$ has $g+1$ connected components, then $X$ is separating.
\end{ex}

Consider a separating curve $X$ and denote the two connected components of $X(\C)\setminus X(\R)$ by $X_+$ and $X_-$. We will explain how to embed $X$ to $(\pp^1)^n$ as a strict Lee--Yang variety. We say that a non-constant real rational function $f\colon X\to\pp^1$ is \emph{separating} if $f^{-1}(H_+)=X_+$. In this case one also has $f^{-1}(H_-)=X_-$ and $f^{-1}(\R\cup\{\infty\})=X(\R)$. It was shown by Ahlfors \cite{ahlfors}*{\S4.2} that every separating curve admits a separating rational function. Several methods for constructing separating functions are provided in \cites{gabard,KS20b}.

\begin{prop}\label{prop:leeyangconstr}
 Let $f_1,\ldots,f_n$ be separating real rational functions on the separating curve $X$. Let $\psi\colon X\to(\pp^1)^n$ be the morphism
 \begin{equation*}
     f=(f_1,-f_2,f_3,\ldots,(-1)^{n-1} f_n)\colon X\to(\pp^1)^n
 \end{equation*}
 composed with the coordinate-wise M\"obius transformation $z\mapsto\frac{z+i}{z-i}$. If the restriction of $f$ to $X(\R)$ is injective, then the image $\tilde{X}\subseteq(\pp^1)^n$ of $\psi$ is a one-dimensional strict Lee--Yang variety.
\end{prop}

\begin{proof}
 The M\"obius transformation $z\mapsto\frac{z+i}{z-i}$ maps the lower open half-plane to the open unit disc. Because $X$ is separating, so in particular a real curve, and the $f_i$ are real rational functions, this implies that $\tilde{X}$ is is invariant under the coordinate-wise involution $z\mapsto 1/\overline{z}$. Next we show that for every $z\in\tilde{X}\setminus\tilde{X}(\T)$ we have $\var(\log|z|)\geq n-1$. Let $z=\psi(x)$ for $x\in X$. Without loss of generality, assume that $x\in X_+$. Then $\Ima(f_i(x))>0$ for all $i=1,\ldots,n$ because $f_i$ is separating. This shows that
 \begin{equation*}
     \var (\Ima(f_1(x),-f_2(x),f_3(x),\ldots,(-1)^{n-1} f_n(x)))=n-1
 \end{equation*}
 which implies $\var(\log|z|)= n-1$. It remains to show that $\tilde{X}(\T)$ is contained in the smooth part of $\tilde{X}$. Letting $X'$ be the image of $f$, this is equivalent to $X'(\R)$ being contained in the smooth part of $X'$. We first observe that, because $f_1$ is separating, the preimage of $X'(\R)$ under $f$ is $X(\R)$. Therefore, because $X$ is smooth and $f$ is injective on $X(\R)$, it suffices to show that the differential of $f$ at every point $x\in X(\R)$ is injective. This follows because the differential of $f_1=\pi_1\circ f$ is injective at every point $x\in X(\R)$ by \cite{KS20a}*{Theorem~2.19}.
\end{proof}

\begin{ex}\label{ex:rational}
 Let $X=\pp^1$ and $f_1,\ldots,f_n$ be some real univariate polynomials of degree one with positive leading coefficients. These are clearly separating. Then the map $\psi$ from \Cref{prop:leeyangconstr} is an embedding and therefore $\tilde{X}$ is a strict Lee--Yang variety in $(\pp^1)^n$ of codimension $n-1$. The (closure in $(\pp^1)^3$ of the) curve considered in \Cref{ex:runningRational} is of this form and thus a strict Lee--Yang variety.
\end{ex}

In order to construct separating functions on curves other than $\pp^1$, the following criterion turns out to be convenient.

\begin{lem}[\cite{KS20b}*{Lemma~2.10}]\label{lem:interlcrit}
    Let $f$ be a real non-constant rational function on the separating curve $X$. Then $f$ or $-f$ is separating if and only if every connected component of 
    \begin{equation*}
        X(\R)\setminus\{P\mid f(P)=0\}
    \end{equation*}
    contains exactly one pole of $f$. In other words, on each connected component of $X(\R)$ zeros and poles of $f$ \emph{interlace}.
\end{lem}

Finally, we note that products of strict Lee--Yang varieties are Lee--Yang varieties.

\begin{ex}\label{ex:product}
If for each $j=1, \hdots, r$, $X_j \subseteq (\pp^1)^{n_j}$ is a strict Lee--Yang variety of codimension $1\leq d_j < n_j$, 
then the product $X=X_1\times \cdots \times X_r  \subseteq (\pp^1)^{n_1+\hdots +n_r}$ is a Lee--Yang variety of codimension $d = \sum_{i=1}^r d_j$. 
To see this, note that for any point $z = (z_j)_j\in X$ where $z_j\in X_j$, $\varbar(\log|z_j|)\geq d_j$. 
Therefore $\varbar(\log|z|)\geq \sum_{j=1}^rd_j =d$. The torus part $X(\T)$ is contained in the smooth part of $X$ but in general $X$ will not satisfy part (2) of \Cref{def:leeyang}.
\end{ex}

We will see more examples in \Cref{sec:concreteexample}.

\subsection{Transversality}\label{sec:trans}
We discuss some transversality properties of Lee--Yang varieties.
Throughout this subsection we fix a Lee--Yang variety $X\subseteq(\pp^1)^n$ of dimension $c=n-d$ and a matrix $L\in\R^{n\times d}$  with positive $d\times d$ minors. 
Further let
\begin{equation*}
    Y=\{z\in\C^{n}\mid \exp(2\pi iz)\in X\}
\end{equation*}
and $Y(\R)=Y\cap\R^n$.  We further let $X_{\rm sm}$ be the smooth part of $X$ and $Y_{\rm sm}$ the set of points $z\in Y$ such that $\exp(2\pi iz)\in X_{\rm sm}$. The following Lemma is a straightforward consequence of \Cref{prop:RR}.

\begin{lem}\label{lem:realtang}
    If $x\in\C^d$ such that $Lx\in Y$, then $x\in\R^d$.
\end{lem}
\begin{proof}
    The assumption $Lx\in Y$ implies that $\exp(2\pi i Lx)\in X$. Now the claim follows from \Cref{prop:RR}.
\end{proof}

\begin{lem}\label{lem:tangentspace}
    Let $x\in Y_{\rm sm}(\R)$. The following holds true:
    \begin{enumerate}
        \item The tangent space $T_xY\subseteq\C^n$ is the $\C$-span of $T_xY(\R)$.
        \item The orthogonal complement of $T_xY(\R)$ is in ${\rm Gr}_{\geq}(d,n)$.
    \end{enumerate}
\end{lem}

\begin{proof}
 Because $X$ is invariant under the coordinate-wise involution $z\mapsto 1/\overline{z}$, it follows that $Y$ is invariant under complex conjugation. This implies part (1). Now we prove part (2). Consider a non-zero tangent vector $w\in T_{x}Y(\R)$. By \Cref{thm:GK} we have to show that $\varbar(w)\geq d$. By part (1) we have $\xi=i\cdot w\in T_x Y$.  Let $ \gamma\colon[0,1]\to Y$ be a smooth path with $\gamma(0)=x$ and $\gamma'(0)=\xi$. For small $t\in(0,1]$ we have by the definition of Lee--Yang varieties that
  \begin{equation*}
   \varbar(\Ima[\frac{1}{t}\left(\gamma(t)-\gamma(0)\right)])=\varbar(\Ima[\gamma(t)])\geq d    
  \end{equation*}
  since $\Ima[\gamma(0)]=0 $ and $\exp(2\pi i\gamma(t))\in X$. Thus 
  \begin{equation*}
   \varbar(w)=\varbar(\Ima[\xi])\geq d   
  \end{equation*}
  since $\gamma'(0)=\xi$ and $\varbar$ is upper-semicontinuous. 
\end{proof}

\begin{lem}\label{lem:strictcovering}
    Let $X$ be a strict Lee--Yang variety, let $x\in X(\T)$ and $I\in\binom{[n]}{c}$. Then the projection map
            $\pi_I\colon X\to(\pp^1)^I$
    is unramified at $x$.
\end{lem}

\begin{proof}
    Let $\psi_1\colon(\pp^1)^n\to(\pp^1)^n$ and $\psi_2\colon(\pp^1)^I\to(\pp^1)^I$ be the coordinate-wise M\"obius transformations $z\mapsto\frac{z+i}{z-i}$. Let $\tilde{X}=\psi_1^{-1}(X)$ and $\tilde{x}=\psi_1^{-1}(x)$. We consider the morphism
    \begin{equation*}
        f=\psi_2^{-1}\circ\pi_I\circ\psi_1\colon \tilde{X}\to(\pp^1)^I.
    \end{equation*}
    Note that $f$ is just the projection from $\tilde{X}$ onto the coordinates from $I$. Because $\psi_1$ and $\psi_2$ are isomorphisms, it suffices to show that $f$ is unramified at $\tilde{x}$. Part (1) of \Cref{def:leeyang} implies that $\tilde{X}$ is invariant under $z\mapsto\overline{z}$, namely a real variety with $\tilde{X}(\R)=\psi_{1}^{-1}(X(\T))$. Here $\tilde{X}(\R)$ denotes the set of fixed points under $z\mapsto\overline{z}$, i.e. $\tilde{X}(\R)=\tilde{X}\cap(\R\cup\{\infty\})^n$. Furthermore, the point $\tilde{x}\in\tilde{X}(\R)$ is a smooth point of $\tilde{X}$. Now let $y\in \tilde{X}$ such that $f(y)$ is real. Then $\pi_I(\psi_1(y))\in\T^I$ and \Cref{cor:torus-fiberednonneg} implies that $\psi_1(y)\in X(\T)$ which in turn shows that $y\in\tilde{X}(\R)$. Thus $f$ is \emph{real-fibered} in the sense of \cite{KS20a}*{Definition~2.1} and \cite{KS20a}*{Theorem~2.19} implies that $f$ is unramified at $\tilde{x}$.
\end{proof}

\begin{cor}\label{cor:strictcovering}
    Let $X$ be a strict Lee--Yang variety and $I\in\binom{[n]}{c}$. Then the projection map
            $\pi_I\colon Y(\R)\to\R^I$
    is everywhere unramified.
\end{cor}

\begin{rem}\label{rem:definemI}
    In the case that $X$ is a strict Lee--Yang variety of dimension $c$, we can use \Cref{cor:strictcovering} to define an orientation on the smooth manifold $Y(\R)$. To this end, for $I\in\binom{[n]}{c}$ let $\pi_I\colon Y(\R)\to\R^I \cong\R^c$ be the projection onto the coordinates from $I$. By \Cref{cor:strictcovering} each $\pi_I$ is unramified. This implies in particular that for every subset $I=\{i_1,\ldots,i_c\}$ with $1\leq i_1<\cdots<i_c\leq n$ the top-dimensional differential form
    \begin{equation*}
       s(I)\cdot dx_I=s(I)\cdot dx_{i_1}\wedge\cdots\wedge dx_{i_c}
    \end{equation*}
    is nowhere vanishing on $Y(\R)$ and thus defines an orientation on $Y(\R)$. We claim that all these orientations agree. To see this, consider the tangent space of $Y(\R)$ at some point. We can write this tangent space as the range of an $n\times c$ matrix $M$ such the determinant of the first $c$ rows is positive. The differential form $dx_I$ evaluates at this basis to the determinant $M_I$ of the submatrix of $M$ whose rows are indexed by $I$. Because the determinant of the first $c$ rows is positive, the columns of $M$ form a reference frame for the orientation induced by $dx_{[c]}$. Because the orthogonal complement of the range of $M$ is in $\Gr_{\geq}(n-c,n)$ by \Cref{lem:tangentspace}, it follows that the sign of $M_I$ is $s(I)$ which implies the claim. 
    In the same way we can define an orientation on $X(\T)$. The measures on $Y(\R)$ and $X(\T)$ defined by these orientations and the differential forms $dx_I$ will play an important role in \Cref{sec:measuresandrestrictions}.
\end{rem}

Let $ z\cdot z':=(z_{1}z_{1}',\ldots,z_{n}z_{n}') $ denote point-wise multiplication in $ (\C^*)^{n}$.

\begin{thm}\label{thm:transv1}
Let $X\subseteq(\pp^1)^n$ be a  Lee--Yang variety of dimension $c=n-d$ and let $L\in\R^{n\times d}$ be a matrix with positive $d\times d$ minors. 
 We denote by $X_{\rm sm}$ the smooth part of $X$.
  For every $z_{0}\in\T^{n}$ the map $x\mapsto z_{0}\cdot\exp(2\pi iLx)$ is transverse to $X_{\rm sm}$ as a map from $\C^{d}$ to $\C^n$, and is transverse to $X_{\rm sm}(\T)$ as a map from $\R^{d}$ to $\T^{n}$.
\end{thm} 

\begin{proof}
  Since $z_{0}^{-1}\cdot X=\{z_{0}^{-1}\cdot z\mid z\in X\} $ is a Lee--Yang variety of the same codimension as $X$ by \Cref{lem:translate}, it is enough to only consider $z_{0}=(1,1,\ldots,1)$. For the first claim we need to show that if $x=Ly\in Y$ for some $y\in\C^{d}$ such that $\exp(2\pi ix)$ is a smooth point of $X$, then the tangent spaces $T_{x}Y$ and $T_{x}L\C^{d}=L\C^{d}$ add up to $\C^{n}$. \Cref{lem:realtang} implies that $y\in\R^d$ and thus $x\in Y(\R)$. Thus by part (1) of \Cref{lem:tangentspace} it suffices to show that  $T_xY(\R)+L\R^d=\R^n$ which then also implies the second claim. For dimension reasons, this is equivalent to showing that $T_{x}Y(\R)\cap L\R^{d}=\{0\}$. But this follows now directly from \Cref{cor:intsgr} and part (2) of \Cref{lem:tangentspace}.  
\end{proof}

For the rest of this subsection we assume that $X(\T)$ is contained in the smooth part of $X$. Then
\Cref{thm:transv1} implies in particular that for all $y\in\R^n$ the set
\begin{equation*}
    \Lambda_y=\{x\in\R^{d}\mid \exp(2\pi i(Lx+y))\in X \}
\end{equation*}
is discrete. The last goal of this subsection is to bound the Hausdorff distance of $\Lambda_y$ to $\Lambda_0$ in terms of $y$. 
To this end let $V=L\R^d\subseteq\R^n$ and $V^\perp\subseteq\R^n$ its orthogonal complement. We denote by $\pi_1\colon\R^n\to V$ and $\pi_2\colon\R^n\to V^\perp$ the orthogonal projections.

\begin{lem}\label{lem:inversecont}
    Let $K\subseteq V^\perp$ be simply connected and let $Y_K=\pi_2^{-1}(K)\cap Y(\R)$.
    \begin{enumerate}
        \item The set $Y_K$ has countably many connected components $(Y_i)_{i\in\N}$.
        \item The restriction $\pi_2|_{Y_i}\colon Y_i\to K$ is a homeomorphism for every $i\in\N$.
        \item The inverse $\psi_i\colon K\to Y_i$ of $\pi_2|_{Y_i}\colon Y_i\to K$ is continuously differentiable for every $i\in\N$.
    \end{enumerate}
\end{lem}

\begin{proof} By transversality, see
    \Cref{cor:intsgr} and part (2) of \Cref{lem:tangentspace}, the projection $ \pi_{2}\colon Y(\R)\to V^{\perp} $ is a covering map and using that $ Y(\R) $ is an analytic variety, each fiber, $ \pi_{2}^{-1}(w) $ with $ w\in V^{\perp} $, is a zero dimensional semianalytic set, hence discrete and countable. Therefore, $ Y_{K} $ has countably many connected components $ (Y_{j})_{j\in\N} $,  each homeomorphic to $ K $, proving (1) and (2). The inverse function theorem further implies part (3).
\end{proof}

\begin{thm}\label{lem: transversality}
Let $X\subseteq(\pp^1)^n$ be a  Lee--Yang variety of dimension $c=n-d$ such that $X(\T)$ is contained in the smooth part of $X$. Let $L\in\R^{n\times d}$ be a matrix with positive $d\times d$ minors. 
    There are 
    differentiable functions
    \begin{equation*}
        \varphi_i\colon \left[-\frac{1}{2},\frac{1}{2}\right]^n\to\R^d
    \end{equation*}
    for $i\in\N$ such that $\Lambda_y=\{\varphi_i(y)\mid i\in\N\}$ for all $y\in\left[-\frac{1}{2},\frac{1}{2}\right]^n$. Moreover, there exists a positive constant $C>0$ such that each $\varphi_i$ is $C$-Lipschitz continuous.
\end{thm}

\begin{proof}
    Let $f\colon V\to\R^d$ be the inverse of $\R^d\to V,\,x\mapsto Lx$. Let $K=\pi_2([-N,N]^n)$ for some large enough $N\in\N$ and use the notation from \Cref{lem:inversecont}. By construction we have that $\psi_i(\pi_2(y))-y\in V$ for all $y\in[-N,N]^n$. Thus the map
    \begin{equation*}
        \varphi_i\colon [-N,N]^n\to\R^d,\, y\mapsto f(\psi_i(\pi_2(y))-y)
    \end{equation*}
    is well-defined. By construction these maps are continuously differentiable and satisfy $\Lambda_y=\{\varphi_i(y)\mid i\in\N\}$ for all $y\in\left[-\frac{1}{2},\frac{1}{2}\right]^n$. Since $[-N,N]^n$ is compact, this implies that each $\varphi_i$ is $c_i$-Lipschitz continuous for some $c_i\geq0$. 
    In order to find a uniform Lipschitz constant, we let $m\in\N$ such that every $i\in\N$ with $Y_i\cap\left[-\frac{1}{2},\frac{1}{2}\right]^n\neq\emptyset$ satisfies $i\leq m$. Then for every $j\in\N$ there exists $k_j\in\Z^n$ and $i\leq m$ such that 
    \begin{equation*}
     Y_j\cap\pi_2^{-1}(\pi_2(\left[-\frac{1}{2},\frac{1}{2}\right]^n))\subseteq k_j+ Y_i.
    \end{equation*}
    This implies that the Lipschitz constant of $\varphi_j|_{\left[-\frac{1}{2},\frac{1}{2}\right]^n}$ can be bounded by the Lipschitz constant $c_i$ of $\varphi_i$. In particular, we can choose $C=\max_{i=1}^m(c_i)$.
\end{proof}

\begin{cor}\label{cor:hddist}
    There exists a positive constant $C>0$ such that 		for every $ y\in\left[-\frac{1}{2},\frac{1}{2}\right]^n$ the Hausdorff distance of $\Lambda_0$ and $\Lambda_y$ can be bounded as follows
 		\[\mathrm{dist}(\Lambda_0,\Lambda_{y})\le C\|y\|.\]
\end{cor}

\begin{proof}
    This follows directly from \Cref{lem: transversality}.
\end{proof}
\subsection{Topology of strict Lee--Yang varieties}
 In this short subsection we determine the topology of strict Lee--Yang varieties. We will not make use of this later on.
 For $I\subseteq[n]$ we denote by 
 \begin{equation*}
  \pi_I\colon (\pp^1)^n\to(\pp^1)^{I}\cong (\pp^1)^{|I|}   
 \end{equation*}
the projection onto the coordinates from $I$.

\begin{Def}
 Let $X\subseteq(\pp^1)^n$ be a closed subvariety of pure dimension $c$. The \emph{multidegree} $d(X)$ of $X$ is the tuple $(d_{[n]\setminus I})_{I\in\binom{[n]}{c}}$ of nonnegative integers such that $d_{[n]\setminus I}$ is the cardinality of the fiber of a generic point in $\pp^c$ under the map $\pi_{I}\colon X\to\pp^c$. 
\end{Def}

\begin{ex}
 If $X\subseteq(\pp^1)^n$ is a hypersurface defined by an $n$-variate polynomial having degree $d_i$ in the $i$-th variable, then the multidegree of $X$ is $(d_1,\ldots,d_n)$.
\end{ex}

For the rest of the subsection we fix a strict Lee--Yang variety $X\subseteq(\pp^1)^n$ of dimension $c$ with multidegree $(d_{[n]\setminus I})_{I\in\binom{[n]}{c}}$.

\begin{lem}\label{lem:coveringdIsheets}
For every $I\in\binom{[n]}{c}$ the projection $\pi_{I}\colon X(\T)\to\T^c$ is a covering map with $d_{[n]\setminus I}$ sheets. In particular, we have $d_{[n]\setminus I}>0$.
\end{lem}

\begin{proof}
 It was shown in \Cref{lem:strictcovering} that $\pi_I$ is everywhere unramified and thus it is a covering map. This implies in particular that $d_I>0$. Moreover, by \Cref{cor:torus-fiberednonneg} the number of sheets equals $d_I$.
\end{proof}

By \Cref{rem:definemI} and \Cref{lem:coveringdIsheets} we can choose an orientation on $X(\T)$ with
 \begin{equation}\label{eq:degint}
 d_{[n]\setminus I}=\frac{s(I)}{(2\pi i)^c}\int_{X(\T)}\frac{dz_{i_{1}}\wedge\cdots\wedge dz_{i_{c}}}{z_{i_{1}}\cdots z_{i_{c}}}
\end{equation}
for all $I=\{i_1,\ldots,i_c\}\subseteq [n]$ with $i_1<\cdots<i_c$.
Using this we can determine the topology of $X(\T)$ as a subset of $\T^n$. Namely, \Cref{lem:coveringdIsheets} implies that each connected component of $X(\T)$ is a  compact connected covering space of $\T^c$ and thus itself homeomorphic to $\T^c$.
 In order to the determine the homology class of $X(\T)$ in $H_c(\T^n)$ note that the differential $c$-forms
  \begin{equation*}
 \frac{s(I)}{(2\pi i)^c}\cdot\frac{dz_{i_{1}}\wedge\cdots\wedge dz_{i_{c}}}{z_{i_{1}}\cdots z_{i_{c}}}
\end{equation*}
are a basis of $H^c(\T^n)=\Hom(H_c(\T^n),\Z)$. Thus letting $(\beta_I)_{I\in \binom{[n]}{c}}\subseteq H_c(\T^n)$ be the dual basis, we obtain
\begin{equation*}
 [X(\T)]=\sum_{I\in \binom{[n]}{c}} d_{[n]\setminus I}\cdot \beta_I.
\end{equation*}

\subsection{The complex geometry of strict Lee--Yang varieties}
The goal of this section is to prove that certain integrals on strict Lee--Yang varieties vanish, that will play a role when computing Fourier transformations. Namely, we will show that for every strict Lee--Yang variety $X\subseteq(\pp^1)^n$ of dimension $c=n-d$
\begin{equation}\label{eq:intvanishes}
     \int_{X(\T)}z^{-k}\cdot\frac{dz_{i_1}\wedge\cdots\wedge dz_{i_{c}}}{z_{i_1}\cdots z_{i_{c}}}=0,
\end{equation}
whenever $k\in\Z^n$ with $\var(k)\geq d$ and $i_1,\ldots,i_c\in[n]$.
We start with the instructive example when the strict Lee--Yang is isomorphic to $\pp^1$ and embedded to $(\pp^1)^n$ as in \Cref{prop:leeyangconstr}.

\begin{ex}\label{ex:vanishingp1}
    Let $f_1,\ldots,f_n\colon\pp^1\to\pp^1$ be some real rational functions such that $f_j^{-1}(H_+)=H_+$ for $j=1,\ldots,n$ where $H_+\subseteq\C$ is the upper half-plane. In this case one also has $f_j^{-1}(H_-)=H_-$ and $f_j^{-1}(\R\cup\{\infty\})=\R\cup\{\infty\}$. We further let $g_j$ be the rational function obtained by post-composing $(-1)^{j+1}\cdot f_j$ by the M\"obius transformation $z\mapsto\frac{z+i}{z-i}$ which maps the lower half-plane to the open unit disc. This implies for even $j$ that $g_j$ has all its poles in the upper half-plane and all zeros in the lower half-plane, and vice versa for odd $j$.
    Let $X$ be the image of the map
    \begin{equation*}
        \pp^1\to(\pp^1)^n,\, z\mapsto (g_1(z),\ldots,g_n(z)).
    \end{equation*}
    In this situation \Cref{eq:intvanishes} is of the form
    \begin{equation}\label{eq:exint}
     \int_{X(\T)}z_1^{k_1}\cdots z_n^{k_n} \cdot\frac{dz_{j}}{z_{j}}=0
 \end{equation}
 for all $k\in\Z^n$ with $\var(k)=n-1$ and all $j=1,\ldots,n$. Here $z_1,\ldots,z_n$ are the coordinates on $(\pp^1)^n$. If we replace $k$ by $-k$, then the resulting integral will be the complex conjugate of the original integral. Therefore, without loss of generality, we can assume that $k_1\leq0$.
 The integral in \Cref{eq:exint} can be expressed as
    \begin{equation*}
     \int_{\R}g_1^{k_1}\cdots g_n^{k_n} \cdot\frac{dg_{j}}{g_{j}}.
 \end{equation*}
 The condition $\var(k)=n-1$ implies that $(-1)^ik_i>0$ for $i=1,\ldots,n$.
 Therefore, by our assumption on the location of the zeros and poles of $g_i$, the expression
 \begin{equation*}
     g_1^{k_1}\cdots g_n^{k_n} \cdot\frac{dg_{j}}{g_{j}}
 \end{equation*}
 has no pole in the closed lower half-plane which shows that the integral in question vanishes by Cauchy's integral theorem. We note that even when $\var(k)<n-1$, the integral can be efficiently computed using Cauchy's integral formula.
\end{ex}

The proof for the general case of a strict Lee--Yang variety $X\subseteq(\pp^1)^n$ is more technical but it follows the same line of thoughts as \Cref{ex:vanishingp1}. Instead of Cauchy's integral theorem we will use de Rham's theorem. To this end, we will realize $X(\T)$ as a submanifold of a smooth complex variety $\tilde{X}$ such that the fundamental class of $X(\T)$ is trivial in homology of $\tilde{X}$. We will show, by making use of $\var(\log|z|)\geq d$ for all $z\in X\setminus X(\T)$, that the integrands in question are top-dimensional regular differential forms on $\tilde{X}$ and therefore closed.

First, we need some preparations from algebraic topology and algebraic geometry. We will use basic notions and results from these areas. As references we recommend \cites{dold,hatcher,prasolov} and \cites{hartshorne,harris,mumford} respectively.
We start with a lemma from algebraic topology.

\begin{prop}\label{prop:boundary}
 Let $f\colon X\to Y$ be a proper (in the classical topology) and surjective morphism of smooth irreducible quasi-projective varieties of the same dimension $n=\dim(X)=\dim(Y)$. Further let $\cN\subseteq Y$ be a compact connected oriented smooth submanifold of dimension $k$ which is disjoint from the branch locus of $f$. Then we have
 \begin{enumerate}
     \item The preimage $\cM=f^{-1}(\cN)$ is a disjoint union of finitely many compact connected smooth submanifolds of dimension $k$.
     \item The restriction $f|_{\cM}\colon \cM\to \cN$ is a covering map. In particular $\cM$ has an orientation induced by the one of $\cN$.
     \item If $[\cN]=0$ in $H_k(Y,\C)$, then $[\cM]=0$ in $H_k(X,\C)$.
 \end{enumerate}
\end{prop} 

\begin{proof}
 First we note that, since $f$ is proper, the set $\cM=f^{-1}(\cN)$ is compact. Let $U\subseteq Y$ the complement of the branch locus. Then, letting $V=f^{-1}(U)$, we have that $f|_V\colon V\to U$ is a covering map. Since $\cN\subseteq U$, the restriction $f_{\cM}\colon \cM\to \cN$ is also a covering map. This implies $(1)$ and $(2)$. In order to prove $(3)$ we consider the Poincar\'e duality isomorphisms for cohomology with compact support $D_X\colon H^{n-k}_c(X,\C)\to H_{k}(X,\C)$ and $D_Y\colon H^{n-k}_c(Y,\C)\to H_{k}(Y,\C)$. We let $\alpha\in H^{n-k}_c(X,\C)$ and $\beta\in H^{n-k}_c(Y,\C)$ be the preimages of the fundamental classes of $\cM$ and $\cN$, respectively. By assumption we have $\beta=0$ and we need to show that $\alpha=0$.  To this end let $\varphi\in H_{n-k}(X,\C)$. We first note that $\alpha(\varphi)$ is the intersection number of $[\cM]$ with $\varphi$. Now since $\cN\subseteq U$ and because $f|_V\colon V\to U$ is a covering map, we can compute the intersection number of $[\cM]$ with $\varphi$ as the intersection number of $[\cN]$ with $f_*\varphi$. This shows that $\alpha(\varphi)=\beta(f_*\varphi)$ for all $\varphi\in H_{n-k}(X,\C)$. Thus $\alpha=f^*\beta=0$ and therefore $[\cM]=0$.
\end{proof}

For the rest of the subsection let $X\subseteq(\pp^1)^n$ be a strict Lee--Yang variety of dimension $c<n$ and let $k\in\Z^n$ such that $\var(k)\geq d=n-c$. We can then choose 
some subset  $J\subseteq[n]$ of $|J|=d+1$ indices on which the entries of $k$ are nonzero and alternate signs. 
That is, if the elements of $J$ are $1\leq j_0<\cdots< j_{d}\leq n$, then
\begin{equation*}
 \var(k_{j_0},\ldots,k_{j_{d}})=d.
\end{equation*}
Let $K=([n]\setminus J)\cup\{m\} $ for some $m\in J$ with $k_{m}<0$. 
(To make the choice of $K$ depend only on the vector $k$, we can choose $J$ to be the lexicographically smallest subset 
among all valid choices and $m$ to be the minimal $m\in J$ with $k_m<0$.)
Finally, let 
\begin{equation}\label{eq:pidefi}
 \pi\colon X\to(\pp^1)^{K}  \cong (\pp^1)^c
\end{equation}
 the projection on the coordinates indexed by $K$. 
 By \Cref{lem:strictcovering} the map $\pi$ is unramified at $X(\T)$.
We denote
\begin{equation*}
 B_{{K}}=\{z\in(\pp^1)^{{K}}\mid |z_i|=1\textrm{ for }i\neq m\textrm{ and }|z_m|\leq1\}.
\end{equation*}
We further denote $B=\pi^{-1}({B_{K}})\subseteq {X}$. 
\begin{lem}\label{lem:closuresubset}
 For $z\in B$ either $z\in \T^n$ or $\textnormal{sgn}(\log|z_j|)=\textnormal{sgn}(k_j)$ for all $j\in J$ and $|z_i|=1$ for all $i\not\in J$. 
\end{lem}

\begin{proof}
Let $z\in B$.  By construction, $|z_i|=1$ for all $j \not\in J$ and $|z_m|\leq 1$. 
 If $|z_m|=1$, then $\pi(z)\in \T^K$ and so $z\in \T^n$. 
 If $|z_m|<1$, then the vector $\log|z|$ is nonzero. Since $X$ is a strict Lee--Yang variety and $z\in X$, 
 we have $\var(\log(|z|))\geq d$.  Since $\log|z_j|=0$ for all $j\not\in J$ and $|J|=d+1$, 
 we see that the entries $\log|z_j|$ for $j\in J$ must be nonzero and alternate in sign in order to achieve $\var(\log|z|)\geq d$. 
Since the entries $(k_j)_{j\in J}$ also alternate sign and $\textnormal{sgn}(\log|z_m|)=\textnormal{sgn}(k_m)$, 
we see that the signs of $\log|z_j|$ and $k_j$ must agree for all $j\in J$. 
\end{proof} 

\begin{lem}\label{lem:pole1}
 For all $i\in[n]$ the rational differential $z_i^{-k_i}\cdot\frac{dz_i}{z_i}$ is regular on $B$.
\end{lem}

\begin{proof}
We go through the three cases $k_i<0$, $k_i>0$ and $k_i=0$.

If $k_i<0$, then this differential is regular except when $z_i=\infty$. In this case, by \Cref{lem:closuresubset}, we have $|z_i|\leq1$ for all $z\in B$ which thus shows that our differential is regular on $B$. 

If $k_i>0$, then our differential is regular except when $z_i=0$. In this case, by \Cref{lem:closuresubset}, we have $|z_i|\geq1$ for all $z\in B$ which thus shows that our differential is regular on $B$.

If $k_i=0$, then our differential is regular except when $z_i=0$ or $z_i=\infty$. In this case we have $i\not\in J$, so by \Cref{lem:closuresubset}, we have $|z_i|=1$ for all $z\in B$ which thus shows that our differential is regular on $B$.
\end{proof}

\begin{thm}\label{thm:intvanish}
    Let $X\subseteq(\pp^1)^n$ be a strict Lee--Yang variety of dimension $c=n-d$. Let $I=\{i_1,\ldots,i_c\}\subseteq[n]$ such that $1\leq i_1<\cdots<i_c\leq n$. Then for all $k\in\Z^n$ with $\var(k)\geq d$ we have 
    \begin{equation*}
     \int_{X(\T)}z^{-k}\cdot\frac{dz_{i_1}\wedge\cdots\wedge dz_{i_{c}}}{z_{i_1}\cdots z_{i_{c}}}=0.
 \end{equation*}
\end{thm} 

\begin{proof}
    We denote $\omega_I=\frac{dz_{i_1}\wedge\cdots\wedge dz_{i_{c}}}{z_{i_1}\cdots z_{i_{c}}}$. Let $\rho\colon \tilde{X}\to X$ be a resolution of singularities. This means that $\tilde{X}$ a smooth variety and $\rho$ is a proper birational surjective morphism which restricts to an isomorphism on $\rho^{-1}({X}_{\rm sm})$ where ${X}_{\rm sm}$ is the smooth part of $X$. Such a resolution of singularities exists by \cite{hironaka}*{Main Theorem I}. Now let $U\subseteq X$ be the maximal open subset on which the rational differential $c$-form $z^{-k}\omega_I$ is regular and let $Z={X}\setminus U$ its complement. Let $Y=(\pp^1)^{K}\setminus \pi(Z)$ where $\pi$ is the map from \Cref{eq:pidefi}. By \Cref{lem:pole1} we have 
    \begin{equation*}
     \T^{K}\subseteq{B_{K}}\subseteq Y   
    \end{equation*}
    and letting $\tilde{Y}=(\pi\circ \rho)^{-1}(Y)\subseteq\tilde{X}$ the $c$-form $z^{-k}\omega_I$ is regular on $\tilde{Y}$. We let $f\colon \tilde{Y}\to Y$ the restriction of $\pi\circ \rho$ to $\tilde{Y}$. Since $X(\T)$ is contained in the smooth part of $X$ and because $X(\T)=\pi^{-1}(\T^{K})$, we can identify $X(\T)$ with the subset $f^{-1}(\T^{K})$ of $\tilde{Y}$. Thus we have to prove
    \begin{equation*}
        \int_{f^{-1}(\T^{K})}z^{-k}\omega_I=0.
    \end{equation*}
    Because $z^{-k}\omega_I$ is a holomorphic $c$-form on $\tilde{Y}$ and $\dim(\tilde{Y})=c$, it is closed. On the other hand, since $\T^{K}$ is the boundary of ${B_{K}}\subseteq Y$, we have that $\T^{K}$ is homologous to zero in $Y$ and thus by \Cref{prop:boundary} also $f^{-1}(\T^{K})$ is homologous to zero in $\tilde{Y}$. Therefore, it follows from Stokes' direction of de Rham's theorem that the integral in question is zero.
\end{proof}

\section{Delone sets}\label{sec:delone}
We remind the reader that a set $A\subseteq\R^d$ is called \emph{uniformly discrete} if there exists $r>0$ such that $ \|x-x'\|\ge r $ for every pair of distinct points in $A$. It is called \emph{relatively dense} if there exists $R>0$ such that $B_{R}(x)\cap A\ne\emptyset$ for every $x\in\R^{d}$. Finally, one says that $A$ is \emph{Delone} if it is both both relatively dense and uniformly discrete. The goal of this section is to prove the following theorem.
\begin{thm}\label{thm:deloneiff} 
	Let $ X \subseteq (\mathbb{P}^1)^n $ be a Lee--Yang variety of dimension $ c = n - d $. Let $ X_{\rm sm} $ denote its smooth part, and let $ L \in \mathbb{R}^{n \times d} $ with range $ V = L\mathbb{R}^{d} \in \Gr_{+}(d,n) $. Then the following statements hold:
	
	\begin{enumerate}
		\item If $X(\T) \subseteq X_{\rm sm}$, then $\Lambda(X,L)$ is Delone.
		\item Conversely, if $x \mapsto \exp(2\pi iLx)$ has a dense image in $\T^n$, then $\Lambda(X,L)$ being Delone implies $X(\T) \subseteq X_{\rm sm}$.
	\end{enumerate}
	
\end{thm}

\begin{rem} Notice that if the $d\times d $ minors of $L$ are positive then 
 $V\in \Gr_{+}(d,n)$ holds (by definition), and if these minors are $\Q$-linearly independent then the assumption that the image of $x\mapsto \exp(2\pi iLx)$ is dense in $\T^{n}$ holds as well, by \Cref{lem:assL}. Finally $X(\T) \subseteq X_{\rm sm}$ holds when $X$ is a strict Lee--Yang variety. 
\end{rem}

Throughout this section we assume 
\begin{itemize}
	\item The variety $X\subseteq(\pp^1)^n$ is a Lee--Yang variety of dimension $c=n-d$.
	\item The real matrix $L\in\R^{n\times d}$ has range $V=L\R^{d}\subseteq\R^n$ that lies in $\Gr_{+}(d,n)$.  
\end{itemize} 

Furthermore, as in \Cref{sec:trans}, we let
\begin{equation*}
    Y=\{z\in\C^{n}\mid \exp(2\pi iz)\in X\},
\end{equation*}
and $Y(\R)=Y\cap\R^n$. Note that $Y(\R)$ is a $c$-dimensional real analytic, $\Z^n$-periodic subvariety of $\R^{n}$, and $\Lambda(X,L)$, as defined in \Cref{lambda}, can be written as   
\begin{equation*}
	\Lambda(X,L)=\{x\in\R^d\mid Lx\in Y(\R)\}.
\end{equation*}
\begin{rem}\label{rem:iso}
	Notice that $x\mapsto Lx$ is an isomorphism between $\R^{d}$ and $V$ that sends $\Lambda(X,L)$ to $V\cap Y(\R)$. 
\end{rem}
We prove the two parts of \Cref{thm:deloneiff} in two different subsections.
\subsection{Proof of \Cref{thm:deloneiff}, part (1)} 
Assuming that $X(\T)\subseteq X_{\rm sm} $, we prove that $\Lambda(X,L)$ is Delone. First, we show that $\Lambda_{y}$ is infinite for every $ y\in\R^n$, where
    \begin{equation*}
    \Lambda_y=\{x\in\R^{d}\mid \exp(2\pi i(Lx+y))\in X \}=\{x\in\R^{d}\mid Lx+y\in Y(\R) \}.
    \end{equation*}
    Suppose that $\tilde{L}\in\Q^{n\times d}$ has range $ \tilde{V}=\tilde{L}\R^{d}\in\text{Gr}_{+}(d,n)$, and let $y\in Y(\R)$. Then $\tilde{\Lambda}_{y}=\{x\in\R^{d}\mid \tilde{L}x+y\in Y(\R) \}$ contains the lattice $m\Z^{n}$, where $m$ is the smallest common denominator of the $L_{ij}$ since $ Y(\R)$ is $\Z^{n}$ periodic. Therefore, $ \tilde{V}\cap (Y(\R)-y)=\tilde{L}\tilde{\Lambda}_{y}$ is infinite. Since $X(\T)\subseteq X_{\rm sm}$, \Cref{thm:transv1} says that the map $x\mapsto \exp(2\pi i(Lx+y))$ is transversal to $X$ for every $y\in\R^{n}$ and every $L$ in the open set of matrices in $ \R^{n\times d} $ with positive $ d\times d $ minors. This means that $(Y(\R)-y)$ is a smooth manifold\footnote{in fact a real analytic manifold} that intersects $\tilde{V}$ transversally for every choice of $(y,\tilde{V})\in\R^{n}\times\text{Gr}_{+}(d,n) $. Since $\R^{n}\times\text{Gr}_{+}(d,n)$ is connected, and $\tilde{V}\cap (Y(\R)-y) $ is infinite for some choice of $(y,\tilde{V})$, it follows that $(\tilde{V}-y)\cap Y(\R) $ is infinite for all $(y,\tilde{V})\in\R^{n}\times\text{Gr}_{+}(d,n) $. In particular, $V\cap (Y(\R)-y)$ is infinite for all $y$, and it is also discrete because $\dim(Y(\R))+\dim(V-y)=\dim(\R^{n})$. Since $x\mapsto Lx$ is an isomorphism that sends $\Lambda_{y}$ to $V\cap(Y(\R)-y)$, then $\Lambda_{y}$ is discrete and infinite. 
    
    We now prove that $\Lambda=\Lambda_{0}$ is Delone, by introducing two functions, which are continuous due to the transversality:
	\begin{align*}
		r_{0}\colon[0,1]^{n}\to\R_{\ge 0},&\quad  r_{0}(y) =\min\{\|x\|\mid x\in\Lambda_{y}\}=\mathrm{dist}(0,\Lambda_{y}),\\
		r_{1}\colon Y\cap[0,1]^n\to\R_{+}, &\quad r_{1}(y)  =\min\{\|x\|\mid x\in\Lambda_{y}\setminus\{0\}\}=\mathrm{dist}(0,\Lambda_{y}\setminus\{0\}).
	\end{align*}
	Notice that $ r_{1} $ is positive and $ r_{0} $ is non-negative and is not the zero function. Since $ [0,1]^{n} $ is compact, this means that $r_{1}$ has a positive minimum and $r_{0}$ has a positive maximum, say  $ R_{min},R_{max} $ such that $ r_{0}(y)\le R_{max}  $ for all $ y\in[0,1] $ and $ r_{1}(y)\ge R_{min}  $ for all $ y\in Y\cap[0,1]^n $. Also notice that for every $x\in\R^{d}$, there exists $y_{x}\in[0,1]^{n}$ that satisfies $y_{x}-Lx=0\mod{1}$ such that $\Lambda-x=\Lambda_{y_{x}}$. This means that for all $ x\in\R^{d} $, the distance between $x$ and $\Lambda$ is bounded by
	\[\mathrm{dist}(x,\Lambda)=\mathrm{dist}(0,\Lambda_{y_{x}})=r_{0}(y_{x})\le R_{\max}.\]
	We conclude that $\Lambda$ is relatively dense. To see that it is also uniformly discrete, let $x$ and $x'$ be two distinct points in $\Lambda$, so that 
	\[\|x-x'\|\le \mathrm{dist}(x,\Lambda\setminus\{x\})=\mathrm{dist}(0,\Lambda_{y_{x}}\setminus\{0\})=r_{1}(y_{x}),\]
	where $x\in\Lambda$ implies that $y_{x}\in Y\cap [0,1]^{n} $. Hence $\Lambda$ is uniformly discrete.
	\qed
	
\subsection{Proof of \Cref{thm:deloneiff}, part (2)}
From here on we assume, in addition to the assumptions made at the beginnign of this section, that $x\mapsto \exp(2\pi i Lx)$ has dense image in $\T^{n}$, and that $\Lambda(X,L)$ is Delone. We need to show that $X\subseteq X_{\rm sm}$. We first prove a series of necessary lemmas. 

We denote the orthogonal complement of $V$ in $\R^n$ by $V^\perp$, and define $\pi_1,\pi_2$ to be the projections from $\R^n=V\oplus V^\perp$ onto $V$ and $V^\perp$ respectively. We will also denote by $\pi_1,\pi_2$ the projections from $\C^n$ onto the $\C$-spans $V_\C$ and $V_\C^\perp$ of $V$ and $V^\perp$ respectively.

\begin{lem}\label{lem:analrealfib}
    If $x\in Y$ and $\pi_2(x)$ is real, then $x\in Y(\R)$.
\end{lem}

\begin{proof}
	We can write $x=Ly+x_{2}$ in the $\C=V_\C\oplus V_\C^\perp$ decomposition, with $y\in\C^d$ and $x_{2}=\pi_2(x)$ which is real by assumption. It suffices to show that $y$ is also real. Let $a=\exp(-2\pi i x_2)$ and $m_{a}\colon\T^{n}\to\T^{n}$ defined by $m_{a}(\exp(2\pi i z)):=\exp(2\pi i (z-x_{2}))$. Then $m_a(X)$ is a Lee--Yang variety by \Cref{lem:translate} and $x\in Y $ is equivalent to $\exp(2\pi ix)\in X$, which implies $\exp(2\pi i Ly)\in m_a(X)$, and so $y\in\R^d$ by \Cref{prop:RR}.
\end{proof}

\begin{lem}\label{lem:projtrans}
    For $y\in Y_{\rm sm}(\R)$ we have $T_yY(\R)\cap V=\{0\}$ and $T_yY\cap V_\C=\{0\}$.
\end{lem}

\begin{proof}
    By part (2) of \Cref{lem:tangentspace} the orthogonal complement of $T_yY(\R)$ is in $\Gr_{\geq}(d,n)$. Thus by \Cref{cor:intsgr} we have $T_yY(\R)\cap V=\{0\}$. The second claim follows then from part (1) of \Cref{lem:tangentspace}.
\end{proof}

\begin{cor}\label{cor:contsopen}
    Let $y\in Y(\R)$ and $U\subseteq\C^n$ be an open neighbourhood of $y$. Then $\pi_2(U\cap Y)$ contains an open subset of $V_\C^\perp$.
\end{cor}

\begin{proof}
    By \Cref{prop:smoothdense}, $Y(\R)$ is the closure of $ Y_{\rm sm}(\R)$, so $Y_{\rm sm}(\R)\cap U\ne\emptyset$ and so there is a point $y'\in Y_{\rm sm}(\R)\cap U$. Applying \Cref{lem:projtrans} to $y'$ allows us to use the inverse function theorem and conclude that $\pi_{2}(O\cap Y)$ is open for some small enough $O\subseteq U$ neighborhood of $y'$.
\end{proof}

\begin{lem}\label{lem:analproj}
    Let $U\subseteq\C^{a+b}$ be a connected open subset and $Z$ an analytic subvariety of $U$. Let $\pi\colon\C^{a+b}\to\C^a$ be the projection onto the first $a$ coordinates and assume
    \begin{equation}\label{eq:realfib}
        \forall z\in Z\colon \pi(z)\in\R^a\Rightarrow z\in\R^{a+b}.
    \end{equation}
    Then, every point $y\in Z\cap\R^{a+b}$ has a connected open neighborhood $U'\subseteq U$ such that $\pi(Z\cap U')$ is an analytic subvariety of the open set $\pi(U')\subseteq\C^{a}$.
\end{lem}

\begin{proof}
    We prove the claim by induction on $b$. For $b=0$ the claim is trivial. Thus let $b>0$ and assume that it is true for $b-1$. Let $\pi'\colon\C^{a+b}\to\C^{a+b-1}$ be the projection onto the first $a+b-1$ coordinates and $\pi''\colon\C^{a+b-1}\to\C^{a}$ be the projection onto the first $a$ coordinates so that $\pi=\pi''\circ\pi'$. 
    
    Now let $y\in Z\cap\R^{a+b}$ be an arbitrary point. We first prove that there exists an open polydisc $U_0\subseteq\C^{a+b}$ around $y$ such that $Z':=\pi'(U_0\cap Z)$ is an analytic subvariety of $\pi'(U_0)\subseteq\C^{a+b-1}$.
    To this end, observe that $y$ lies in the one-dimensional fiber 
    \begin{equation*}
    	G=\{(\pi'(y),\lambda)\in U\mid \lambda\in\C\}.
    \end{equation*}
    The intersection $G\cap Z$ is a complex subvariety of complex dimension at most one, so its dimension over the reals is either $0$ or $2$. Since $y$ is real, \eqref{eq:realfib} implies that
    \begin{equation*}
        G\cap Z\subseteq\{(\pi'(y),\lambda)\in U\mid \lambda\in\R\}.
    \end{equation*}
    In particular, this shows that $G\cap Z$ is at most one-dimensional over $\R$ which means it must be a zero dimensional analytic variety, i.e. a discrete set points. 
    By the projection theorem for analytic varieties, see e.g. \cite{rudin}*{14.2.4}, there exists an open polydisc $U_0\subseteq\C^{a+b}$ around $y$ such that $Z':=\pi'(U_0\cap Z)$ is an analytic subvariety of the open polydisc $\pi'(U_0)\subseteq\C^{a+b-1}$. 

    Next, we will apply the induction hypothesis to the projection $\pi''$ and the analytic subvariety $Z'$ of the connected open subset $\pi'(U_0)$ of $\C^{a+b-1}$. To this end, we need to show that $\pi''$ satisfies (\ref{eq:realfib}). Thus let $z'\in Z'$ such that $\pi''(z')\in\R^a$. There exists $z\in U_0\cap Z$ such that $z'=\pi'(z)$. Because $\pi$ satisfies (\ref{eq:realfib}) and since $\pi(z)=\pi''(z')\in\R^a$, we have $z\in\R^{a+b}$ and thus $z'=\pi'(z)\in\R^{a+b-1}$. Therefore, by induction hypothesis there exists a connected open neighborhood $U_1\subseteq \pi'(U_0)$, containing $\pi'(y)$, such that $Z''=\pi''(Z'\cap U_1)$ is an analytic subvariety of the open set $\pi''(U_1)\subseteq\C^{a}$. Finally, we consider the open neighborhood $U'=(\pi')^{-1}(U_1)\cap U_0\subseteq\C^{a+b}$ of $y$. Because $U_0$ is a polydisc, this set is connected. We have
    \begin{equation*}
        \pi(Z\cap U')=\pi''(\pi'(Z\cap (\pi')^{-1}(U_1)\cap U_0))=\pi''(Z'\cap U_1)=Z''
    \end{equation*}
    which is an analytic subvariety of $\pi''(U_1)=\pi(U')$. This proves the claim.
\end{proof}

\begin{prop}\label{prop:locallsurj}
    Let $y\in Y(\R)$.
    \begin{enumerate}
        \item If $U\subseteq\C^n$ is an open neighbourhood of $y$, then there exists an open neighbourhood $U'\subseteq U$ of $y$ such that $\pi_2(Y\cap U')=\pi_2(U')$.
        \item If $U\subseteq\R^n$ is an open neighbourhood of $y$, then there exists an open neighbourhood $U'\subseteq U$ of $y$ such that $\pi_2(Y(\R)\cap U')=\pi_2(U')$.
    \end{enumerate}
\end{prop}

\begin{proof}
    Let $U_1\subseteq\C^n$ be a connected open subset containing $y$ and such that $U_1\subseteq U$.
    By \Cref{lem:analrealfib} and \Cref{lem:analproj} there exists a connected open subset $U'\subseteq U_1$ with $y\in U'$ such that $\pi_2(Y\cap U')$ is an analytic subvariety of $\pi_2(U')$. On the other hand, by \Cref{cor:contsopen}, the set  $\pi_2(Y\cap U')$ contains an open subset of $V_\C^\perp$. This shows that $\pi_2(Y\cap U')=\pi_2(U')$. Part (2) follows from part (1) and \Cref{lem:analrealfib}.
\end{proof}

\begin{prop}\label{prop:distanceforsmooth}
    There exists an open neighbourhood $B\subseteq V$ of the origin such that for every $y\in Y(\R)$ we have $(y+B)\cap Y(\R) = \{y\}$.
\end{prop}

\begin{proof}
    Because $\Lambda$ is Delone, there exists $r>0$ such that every pair of distinct points in $\Lambda$ has distance at least $r$. Let $B=LB_{r/2}(0)$ be the image under $L$ of the open ball of radius $\frac{r}{2}$ around the origin in $\R^d$. 
    We claim that $B$ has the desired properties. Assume for the sake of a contradiction that there exist distinct $y_1,y_2\in Y(\R)$ such that
     \begin{equation*}
        y_2\in (y_1+B).
    \end{equation*}    
    Let $U_1,U_2\subseteq\R^n$ be open neighbourhoods of $y_1$ and $y_2$ such that $U_1\cap U_2=\emptyset$ and $\pi_1(U_i)\subseteq y_1+B$. By \Cref{prop:locallsurj} there exists an open neighbourhood $U_i'\subseteq U_i$ of $y_i$ such that $\pi_2(Y(\R)\cap U_i')=\pi_2(U_i')$ for $i=1,2$. Consider the open subset $W=\pi_2(U_1')\cap\pi_2(U_2')$ of $V^\perp$. Because $\pi_2(y_1)=\pi_2(y_2)\in W$ the set $W_i=\pi_2^{-1}(W)\cap U_i'$ is an open neighbourhood of $y_i$ in $\R^n$ for $i=1,2$. We have $\pi_2(W_i\cap Y(\R))=W$.
    The assumption on $L$ implies that $V+\Z^n$ is dense in $\R^n$. Then
    \begin{equation*}
        \pi_2(V+\Z^n)=\pi_2(\Z^n)
    \end{equation*}
    is dense in $V^\perp$. Thus there exists $m\in\Z^n$ such that $\pi_2(m)\in W$ and we obtain $y_i'\in W_i\cap Y(\R)$ such that $\pi_2(y_i')=\pi_2(m)$, i.e. $y_i'-m\in V$, for $i=1,2$. We can therefore write $y_i'-m=Lx_i$ for some $x_i\in\R^d$. We have
    \begin{equation*}
        \exp(2\pi iLx_i)=\exp(2\pi i y_i'-2\pi i m)=\exp(2\pi i y_i')\in X(\T).
    \end{equation*}
    Thus $x_1,x_2\in\Lambda$. Finally, we have
    \begin{equation*}
        L(x_1-x_2)=y_1'-y_2'=\pi_1(y_1')-\pi_1(y_2').
    \end{equation*}
    By our choice of the $U_i$ and since $W_i\subseteq U_i$, we have $y_1'\neq y_2'$, which implies $x_1\neq x_2$, and 
    \begin{equation*}
     \pi_1(y_1')-\pi_1(y_2')\in B-B\subseteq 2B   
    \end{equation*}
    which implies that $\|x_1-x_2\|<r$ contradicting our choice of $r$.
\end{proof}

\begin{cor}\label{cor:localsheet}
    For every $y\in Y(\R)$ we can find an open neighbourhood $U\subseteq\C^n$ of $y$ such that 
    \begin{equation*}
        \psi\colon U\cap Y\to\pi_2(U),\, y\mapsto\pi_2(y)
    \end{equation*}
    is surjective and $|\psi^{-1}(x)|=1$ for every real $x\in V^\perp$. 
\end{cor}

\begin{proof}
    Let $B\subseteq V$ as in \Cref{prop:distanceforsmooth} and let $U_1\subseteq\C^n$ a neighbourhood of $y$ such that $\pi_1(U)\subseteq y+\frac{1}{2}B$. Then every $x\in V^\perp$ has at most one preimage in $U\cap Y(\R)$ under $\pi_2$. By \Cref{lem:analrealfib} we even have that such $x\in V^\perp$ has at most one preimage in $U\cap Y$ under $\pi_2$. Using \Cref{prop:locallsurj} we can achieve surjectivity by further shrinking $U$.
\end{proof}
We are now in position to prove part (2) of \Cref{thm:deloneiff}. 
\begin{proof}[Proof of \Cref{thm:deloneiff} part (2)]
    Assume that $\Lambda(X,L)$ is Delone. We need to show that $Y(\R)\subseteq Y_{\rm sm}$. Thus let $y\in Y(\R)$ and assume that $y\not\in Y_{\rm sm}$. By \Cref{cor:localsheet} we can find an open neighbourhood $U\subseteq\C^n$ of $y$ such that 
    \begin{equation*}
        \psi\colon U\cap Y\to\pi_2(U),\, y\mapsto\pi_2(y)
    \end{equation*}
    is surjective and $|\psi^{-1}(x)|=1$ for every real $x\in V^\perp$. Let $v\in V^\perp$ be a real nonzero vector which is not contained in the image of the tangent cone at $y$ of the singular locus of $Y$. This ensures that $y$ is an isolated singularity of $Y'=\psi^{-1}(G)$ where $G=\pi_2(y)+\C v$. Thus by \Cref{lem:projtrans} the preimage under $\psi|_{Y'}$ of every real point on $G$ in a neighbourhood of $y$, except for $y$, has cardinality one even when counted with multiplicities. Moreover, because $y$ a singular point of $Y'$, the preimage of $y$ under $\psi|_{Y'}$ has higher multiplicity. On the other hand, the map $\psi|_{Y'}$ is finite and flat at $y$. Indeed, it is finite by \cite{fischer}*{Lemma in \S3.2} and it is flat because the local ring of $Y'$ at $y$ is a finitely generated torsion-free module over the local ring of $G$ at $\pi_2(y)$ which is a principal ideal domain, hence the former is free over the latter. Thus by \cite{fischer}*{Corollary to Proposition 3.13} the fiber cardinality (counting multiplicities) of $\psi|_{Y'}$ is constant in a neighbourhood of $\psi(y)$ contradicting our observation above. 
\end{proof}
\section{Non-periodicity and almost periodicity}
In this section we will prove that, under certain conditions, our construction results in Bohr almost periodic sets that do not contain any periodic set.

\begin{Def}
	A Delone set $A\subseteq\R^{d}$ is called \emph{Bohr almost periodic}
	if for every $ \epsilon>0 $ there is a relatively dense set $ T(\epsilon)\subseteq\R^{d} $ such that for every $ \tau\in T(\epsilon) $, the Hausdorff distance between $ A $ and $ A+\tau=\{x+\tau:x\in A\} $ is at most $ \epsilon $:
	\begin{equation*}
		\mathrm{dist}_{H}(A,A+\tau)\le\epsilon .  
	\end{equation*}
\end{Def} 

\begin{thm}\label{lem:almostperiodic}
    Let $X\subseteq(\pp^1)^n$ be a strict Lee--Yang variety of dimension $c=n-d<n$ and $L$ be a real $n\times d$ matrix all of whose $d\times d$ minors are positive.
	The set $\Lambda(X,L)$ is Bohr almost periodic. 
\end{thm}
\begin{proof}
	By \Cref{cor:hddist} there exists a positive constant $C>0$ such that for every $ y\in\left[-\frac{1}{2},\frac{1}{2}\right]^n$ the Hausdorff distance of $\Lambda$ and \begin{equation*}
		\Lambda_y=\{x\in\R^{d}\mid \exp(2\pi i(Lx+y))\in X \}
	\end{equation*}
	can be bounded by $\mathrm{dist}(\Lambda_0,\Lambda_{y})\le C\|y\|$.
	Given $ \epsilon>0 $ small enough let \begin{equation*}
		T(\epsilon)=\{\tau\in\R^{d}\mid \mathrm{dist}(L\tau,\Z^{n})<\frac{\epsilon}{C}\}                          .
	\end{equation*}
	Then for each $ \tau\in T(\epsilon) $ there is $ y\in\left[-\frac{1}{2},\frac{1}{2}\right]^n $
    such that $\exp(2\pi iy)=\exp(2\pi iL\tau)$ and $ C\|y\|\le \epsilon $. Since 
	\[\Lambda+\tau=\{x+\tau\mid\exp(2\pi iLx)\in X(\T) \}=\{x\mid\exp(2\pi iLx-2\pi iL\tau)\in X(\T) \}=\Lambda_{-y},\] our choice of $C$ gives $ \mathrm{dist}_{H}(\Lambda,\Lambda+\tau)\le \epsilon $. To see that $T(\epsilon)$ is relatively dense we will show that it contains a product $ T_{1}(\delta)\times T_{2}(\delta)\times\ldots\times T_{d}(\delta) $ such that each $ T_{j}(\delta)\subseteq\R $ is relatively dense. Let $ L_{1},\ldots,L_{d} $ be the columns of $ L $ and let $ \delta>0 $ such that $ \mathrm{dist}(L\tau,\Z^{n})<\frac{\epsilon}{C} $ whenever $ \mathrm{dist}(\tau_{j}L_{j},\Z^{n})<\delta $ for all $ j $. Define $ T_{j}(\delta)\subseteq\R $ as the set of $ \tau_{j} $ for which $ \mathrm{dist}(\tau_{j}L_{j},\Z^{n})<\delta $. The fact that $ T_{j}(\delta)$ is a relatively dense set follows from Kronecker's Theorem, see \cite{almostperiodic}*{Corollary on page 34} for example.
\end{proof}

In order to prove non-periodicity statements, we shall assume that the $d\times d$ minors of $L$ are $\Q$-linearly independent. 
We record some immediate consequences of this assumption.
\begin{lem}\label{lem:assL}
	If $L\in\R^{n\times d}$, with $n>d\ge 1$, such that the $d\times d$ minors of $L$ are $\Q$-linearly independent, then the following holds:
	\begin{enumerate}
		\item For every non-zero $v\in L\R^d $ we have
		\begin{equation*}
			\dim_{\Q}(v):=\dim(\mathrm{span}_{\Q}\{v_{1},\ldots,v_{n}\})> n-d.
		\end{equation*}
		Equivalently, $ \det(AL)\ne 0 $ for every matrix $ A\in \Z^{d\times n}$ of rank $ d $.
		\item The map $\R^d\to\T^n,\, x\mapsto\exp(2\pi i Lx)$ is injective with dense image.
		\item The map $\Z^n\to\R^d,\, k\mapsto L^tk$ is injective with dense image.
	\end{enumerate}
\end{lem}

\begin{proof}
    For part (1), assume, by means of contradiction, that there exists $y\in\R^{d}$ such that $v=Ly\ne0$ and $ \dim_{\Q}(v)\le n-d $. This means the entries of $ v$ satisfy at least $d$ linear relations $\langle A_{j},v\rangle=0$ with $A_{j}\in\Z^{n}$ and $(A_{1},\ldots,A_{d})$ linearly independent. Taking the matrix $ A\in\Z^{d\times n} $ whose rows are these $A_{j}$'s provides an integer matrix of rank $d$ such that $ Av=ALy=0 $. 
	In particular, we have $ \det(AL) = \sum_{I}A_IL_I=0 $, summing over  $I\in\binom{[n]}{d}$, by Cauchy-Binet. Since $A$ has full rank, at least one minor is non-zero $ A_{I}\ne 0$, so this is a $\Q$-linear relation on the minors $L_I$, contradicting the assumption.
	
	For parts (2) and (3), the map in (2) is a continuous group homomorphism whose kernel consists of all $x\in\R^d$ such that $v=Lx\in\Z^n$. Because $\dim_{\Q}(v)=1\leq n-d$, part (1) shows that $v=Lx=0$ and thus $x=0$ since $L$ has full rank. This proves that the map in (2) is injective. By Pontryagin duality this is equivalent to the map in (3) having dense image, see e.g. \cite{stroppel}*{Proposition~23.2}. The map in (3) being injective directly follows from $\det(L^tA^t)= \det(AL)\ne 0 $ for every matrix $ A\in \Z^{d\times n} $ of rank $ d $. As before, this is equivalent to the map in (2) having dense image.
\end{proof}

\begin{rem}\label{rem:asslcoord}
	Let $S\in\R^{d\times d}$ be an invertible matrix.
	The matrix $L\in\R^{n\times d}$ has $\Q$-linearly independent $d\times d$ minors if and only if $LS$ does.
\end{rem}

\begin{rem}
	There exist explicit positive parametrizations of the positive Grassmannian $\Gr_+(d,n)$. 
	See e.g. \cite{Postnikov}. For example,
	 $\Gr_+(2,4)$ can be parameterized by $(w,x,y,z)\in\R_+^4$, defining the corresponding point in $\Gr_+(2,4)$ to be the range of  
	$L=L(w,x,y,z):= \begin{pmatrix} 1 & 0& -wy  & -w  \\ 0 & 1 & wxy + z & wx  \end{pmatrix} $.
	The $2\times 2$ minors $(L_{I})_{I\in\binom{[4]}{2}}$ are
	\[L_{\{1,2\}}=1,\ L_{\{1,3\}}=wxy + z,\ L_{\{1,4\}}=wx,\ L_{\{2,3\}}=wy,\ L_{\{2,4\}}=w, \ L_{\{3,4\}}=wz .\]
In particular, the minors are positive if and only if $(w,x,y,z)\in\R_+^4$, and if we take $w,x,y,z$ to be algebraically independent over $\Q$, or even just square-roots of distinct primes, then these minors will $\Q$-linearly independent. 
\end{rem}

\begin{Def}[Algebraic-torus cosets]
	We say that a subgroup $H\subseteq(\C^{*})^{n}$ is an \emph{algebraic subtorus of dimension $ m\le n $} if there is an $n\times m$ matrix with integer coefficients $A\in \Z^{n\times m}$ of rank $m $ such that
	\begin{equation*}
		H=\{\exp(2\pi iAy)\mid y\in\C^{m}\}.
	\end{equation*}
	Given a point $ z=\exp(2\pi ix)\in(\C^*)^n $ the \emph{algebraic-torus coset} $ zH $ is given by
	\begin{equation*}
		zH=\{\exp(2\pi i(x+Ay)\mid y\in\C^{d}\}.
	\end{equation*}
\end{Def}
\begin{thm}\label{thm: non-periodicity2}
	Let $X\subseteq (\C^{*})^{n} $ be an algebraic variety of dimension $c=n-d $, of total degree $D$, and assume that none of the irreducible components of $X$ is an algebraic-torus coset in $(\C^{*})^{n}$. Let $L\in\R^{n\times d}$ with $\Q$-linearly independent $d\times d$ minors. Then 
	$\Lambda(X,L)$
	intersects every $ m $-dimensional $ \Q $-vector space $W\subseteq\R^d$ in at most $r^{m+1}$ points, with $r=e^{(6D\binom{n+D}{D})^{(5D\binom{n+D}{D})}}$    
\end{thm}

\Cref{thm: non-periodicity2} is a consequence of \cite{Evertse00}*{Theorem~1.2}, which is Evertse's proof of a conjecture by Lang, commonly known as Lang's $G_{m}$ conjecture:   
\begin{thm}[\cite{Evertse00}*{Theorem~1.2}]\label{thm:langgm}
	Let $V\subseteq (\C^{*})^n$ be an algebraic variety of total degree $ D $ and let $ G $ be a multiplicative subgroup of $ (\C^{*})^n $ of finite rank $ m $. Let $ \overline{G}:=\{z\in(\C^{*})^n\mid \exists j\in\N:~z^{j}\in G\} $ be its division group. Then $ \overline{G}\cap V $ is contained in a union of at most $ R $ algebraic-torus cosets $ z_{j}H_{j}\subseteq V $ for $ R\le e^{(m+1)(6D\binom{n+D}{D})^{(5D\binom{n+D}{D})}} $. 
\end{thm}
We will need the following lemma.
\begin{lem}\label{lem: algebraic-torus cosets}
	If $L\in\R^{n\times d}$ has $\Q$-linearly independent $d\times d$ minors, and $ zH $ is an algebraic-torus coset of dimension $ m<n-d $, then there is at most one point $ y\in\R^{d} $ with $ \exp(2\pi iLy)\in zH $.
\end{lem}
\begin{proof}
	Assume that there are two points, $ y\ne y'$, such that both $ \exp(2\pi iLy)\in zH $ and $ \exp(2\pi iLy')\in zH $, then $ \exp(2\pi iL(y-y'))\in H $, that is, 
	\begin{equation*}
		v=L(y-y')=Au+2\pi k 
	\end{equation*}
	for some $ u\in\R^m $ and $ k\in\Z^{n} $, where $ A\in \Z^{n\times m}$. It follows that 
	\begin{equation*}
		\dim_{\Q}(v)\le\dim_{\Q}(Au)+1\le\dim_{\Q}(u)+1\le m+1\le n-d
	\end{equation*}
	in contradiction to part (1) of \Cref{lem:assL}.
\end{proof}
We can now prove \Cref{thm: non-periodicity2}.
\begin{proof}[Proof of \Cref{thm: non-periodicity2}]
 Let $\Lambda=\Lambda(X,L)$.
	Let $ W=\mathrm{span}_{\Q}(y_{1},\ldots,y_{m})\subseteq\R^{d}  $ be an $ m $-dimensional $ \Q $-linear subspace, then $ \exp(2\pi iW)=\bar{G} $ is the division group of the group $ G\subseteq (\C^{*})^{n}  $ generated by $ \exp(2\pi iy_{1}),\ldots,\exp(2\pi iy_{m}) $. According to Lang's $G_{m}$ conjecture (i.e. \Cref{thm:langgm}), there are $R=R(W,X)$ algebraic-torus cosets in $X$, say $ z_{j}H_{j}\subseteq X $, such that $ y\in W\cap\Lambda \Rightarrow \exp(2\pi iy)\in\cup_{j=1}^{R} z_{j}H_{j}  $. That is $W\cap\Lambda\subseteq\cup_{j=1}^{M}\Lambda_{j}$ where $\Lambda_{j}$ is the set of points $y\in\Lambda $ with $ \exp(2\pi iy)\in z_{j}H_{j}  $. By assumption, $ X $ has dimension $n-d$ and no irreducible component which is an algebraic-torus cosets, so the algebraic-torus cosets contained in $X$ must have lower dimension, $ \dim(z_{j}H_{j})<n-d $ for all $ j $. It follows from \Cref{lem: algebraic-torus cosets} that each $\Lambda_{j}$ contains at most one point, so $ W\cap\Lambda $ contains at most $ R $ points. \Cref{thm:langgm} also provides a bound on $R$ in terms of $m$ and $D$ the total degree of $X$, $R\le r^{m+1} $, with $ r=e^{(6D\binom{n+D}{D})^{(5D\binom{n+D}{D})}}$.
\end{proof}

\begin{cor}[Non-periodicity] \label{cor-np}
    Let $X\subseteq(\pp^1)^n$ be a strict Lee--Yang variety of dimension $c=n-d<n$ such that no irreducible component of $X$ is an algebraic-torus coset.
    Let $L$ be a real $n\times d$ matrix all of whose $d\times d$ minors are positive and linearly independent over $\Q$. Then the following holds:
    \begin{enumerate}
        \item $\Lambda(X,L)\subseteq\R^{d}$ intersects every lattice in at most $r^{d+1}$ points, and every discrete periodic set in finitely many points.
        \item $\Lambda(X,L)\subseteq\R^{d}$ intersects every set obtained by the ``cut and project'' method from a lattice in $\R^{d}\times\R^{N}$, for some $N\in\N$, in at most $r^{d+N+1}$ points.
    \end{enumerate}
    Here $r$ is the constant introduced in \Cref{thm: non-periodicity2}.
\end{cor}
\begin{proof}
	The $\Q$-span of a lattice in $\R^{d}$ has dimension $d$ as a $\Q$-vector space. Thus by \Cref{thm: non-periodicity2} it can intersect $\Lambda(X,L)$ in at most $r^{d+1}$ points. By the same argument $\Lambda(X,L)$ intersects the translate of a lattice in at most $r^{d+2}$ points. 
    Because every discrete periodic set is a finite union of translates of lattices, this proves part (1). 
	
	For the same reason, every projection of a lattice from $\R^{d+N}$ to $\R^{d}$ will have $\dim_{\Q}\le N+d$. Now recall that a set $ A\subseteq\R^{d} $ is called a ``cut and project'' set, if there is a lattice $\cL\subseteq \R^{d}\times\R^{N}$ for some $N\in\N$ and a bounded set $O\subseteq\R^{N} $ with non-empty interior such that $A=\{x\in\R^{d}\mid\exists y\in O\ \text{s.t.}\  (x,y)\in\cL\}$. This construction implies that $ \dim_{\Q}(A)\le \dim_{\Q}(\cL)\le d+N $.  Now the claim follows from \Cref{thm: non-periodicity2}.
\end{proof}

\section{Proof of \Cref{thm:main1} and \Cref{thm:main2}}\label{sec:measuresandrestrictions}
In this section, we will prove \Cref{thm:main1} and \Cref{thm:main2}. We fix a strict Lee--Yang variety $X\subseteq(\pp^1)^n$ of dimension $c=n-d$ and multidegree $d({X})=(d_J)_{J\in\binom{[n]}{d}}$ and a matrix $L\in\R^{n\times d}$ with positive $d\times d$ minors. We let $\Lambda=\Lambda(X,L)$. As in \Cref{sec:trans} we consider
\begin{equation*}
    \logX=\{y\in\R^{n}\mid \exp(2\pi iy)\in X\}
\end{equation*}
which is a $\Z^{n}$-periodic real analytic manifold of dimension $c=n-d$.  
\begin{Def}($\bm_{I}$)
	For any $I=\{i_{1},\ldots,i_{c}\}$, $1\leq i_1<\cdots<i_c\leq n$, define the (periodic) measure $\bm_{I}$ on $\logX$ given by the $c$-dimensional volume in the $I$ coordinates $d\bm_{I}(x)=|dx_{i_{1}}\wedge\ldots\wedge dx_{i_{c}}|$, and its Fourier transform $\widehat{\bm}_{I}\colon\Z^{n}\to\C$,
	\[\widehat{\bm}_{I}(k)=\int_{\logX\cap[0,1]^{n}}e^{-2\pi i\langle x,k\rangle}d\bm_{I}(x).\]
\end{Def}
By \Cref{rem:definemI} there is an orientation of $\logX$ such that the measure $\bm_I$ can be computed by integrating against the differential form $s(I)\cdot dx_{i_1}\wedge\cdots\wedge dx_{i_c}$. Moreover, the Fourier transform $\widehat{\bm}_{I}(k)$ can be expressed as
\begin{equation}\label{eq:ftonx}
    \widehat{\bm}_{I}(k)=\int_{\logX\cap[0,1]^{n}}e^{-2\pi i\langle x,k\rangle}d\bm_{I}(x)=\frac{s(I)}{(2\pi i)^c} \int_{X(\T)}z^{-k}\frac{dz_{i_{1}}\wedge\ldots\wedge dz_{i_{c}}}{z_{i_{1}}\cdot\ldots\cdot z_{i_{c}}}
\end{equation}
with the orientation on $X(\T)$ as in \Cref{rem:definemI}.
\begin{lem}\label{lem: FT of MI}
	The support of $\widehat\bm_{I}$ is contained in $\Z^{n}_{\var<d}=\{k\in\Z^{n}\mid \var(k)<d\}$, and
$\widehat{\bm}_{I}(0)=d_{[n]\setminus I}$. 
\end{lem} 
\begin{proof}
 By \Cref{eq:ftonx} the claim follows from \Cref{thm:intvanish} and \Cref{eq:degint}. 
\end{proof}

\begin{Def}
	For every Borel set $A\subseteq \logX$ we define:
	\begin{equation*}
		\bm_L(A):=\lim_{\epsilon\to 0}\frac{1}{\epsilon^d}\Vol_{n}\left(A+\epsilon L Q\right)
	\end{equation*}
	where $ \mathrm{Vol}_{n} $ is the Lebesgue measure in $ \R^{n} $, $Q=[0,1]^d\subseteq\R^{d}$  and $A+\epsilon LQ $ is the Minkowski sum of $A$ and $\epsilon LQ $ which is the re-scaled image of $Q$ under $L$. 
\end{Def}
\begin{lem}\label{lem: measL}
	The measure $\bm_{L}$ depends linearly on the $d\times d$ minors of $L$,
	\[\bm_{L}=\sum_{I\in\binom{[n]}{ d}}L_{I}\cdot\bm_{[n]\setminus I}.\]
	In particular, $\widehat{\bm}_{L}(k)=\sum_{I\in{\binom{[n]}{d}}}L_{I}\cdot\widehat{\bm}_{[n]\setminus I}(k)$ is supported inside $\Z^{n}_{\var<d}$, namely $ \widehat{\bm}_{L}(k)=0$ when $\var(k)\ge d$. Moreover, it is uniformly bounded by the zero coefficient 
	\[|\widehat{\bm}_{L}(k)|\le\widehat{\bm}_{L}(0)=\sum_{I\in{\binom{[n]}{d}}}L_{I}\cdot d_{I}.\]
\end{lem}
\begin{proof}
	If we assume that $\bm_{L}=\sum_{I\in\binom{[n]}{d}}L_{I}\bm_{[n]\setminus I}$ then we can use
	the triangle inequality to get $|\bm_{[n]\setminus I}(k)|\le \int_{\logX\cap[0,1]^{n}}d\bm_{[n]\setminus I}=\widehat{\bm}_{[n]\setminus I}(0)=d_{I}  $ where we used \Cref{lem: FT of MI} in the last equality, from which we get 
	\[|\widehat{\bm}_{L}(k)|\le\widehat{\bm}_{L}(0)=\sum_{I\in{\binom{[n]}{d}}}L_{I}d_{I}.\]
	To show that $\bm_{L}=\sum_{I\in\binom{[n]}{ d}}L_{I}\bm_{[n]\setminus I}$, we need to calculate the density $ d\bm_{L}(x_{0}) $ at every $ x_{0}\in\logX $. Given $ I=\{i_{1},\ldots,i_{c}\} $ with $ 0<i_{1}<\ldots<i_{c}\le n $ consider the $ c $-form $ dx_{I}=dx_{i_{1}}\wedge\ldots,dx_{i_{c}}$ and the sign $ s(I)=(-1)^{\sum_{i\in I}i-\sum_{i\notin I}i} $ so that 
	\[d\bm_{I}(x)=s(I)dx_{I}.\]
	Let $A$ an open neighborhood of $x_0\in\logX$ such that there exists an open subset $O\subseteq\R^c$ and a smooth map
	\begin{equation*}
		\varphi\colon O\to \logX
	\end{equation*}
	parametrizing $A$ such that $\varphi(0)=x_0$. We denote the standard coordinates on $\R^c$ by $y=(y_{1},\ldots,y_c)$. The tangent space $T_{x_{0}}\logX$ is transversal to the image of $L$ by \Cref{lem:tangentspace}. Thus, after shrinking $A$ if necessary we can parameterize $A+\epsilon LQ$ by 
	\begin{equation*}
		\psi\colon O\times(\epsilon Q)\to A+\epsilon LQ,\, (y,z)\mapsto\varphi(y)+Lz
	\end{equation*}
	where $z=(z_{1},\ldots,z_{d})$ are the standard coordinates on $\R^d$. Using the notations $ dy=dy_{1}\wedge\ldots\wedge dy_{c} $ and $ dz=dz_{1}\wedge\ldots\wedge dz_{d} $ we can calculate the Jacobian of $\psi$,  
	\begin{align*}
		\det(D\psi(y,z))dy\wedge dz = & \det\begin{pmatrix}
			D\varphi(y) & L
		\end{pmatrix}dy\wedge dz\\
	=& \sum_{I\in{ \binom{[n]}{c}}}s([n]\setminus I) L_{[n]\setminus I}\det\left(\frac{\partial (\varphi_{i_{1}},\ldots,\varphi_{i_{c}}) }{\partial(y_{1},\ldots,y_{c})}\right)dy\wedge dz\\
	 =& \sum_{I\in{ \binom{[n]}{c}}}s([n]\setminus I) L_{[n]\setminus I}dx_{I}\wedge dz,
	\end{align*}
where in the second line we use the Laplace expansion along the rows corresponding to $[c]$, and in the last line we used the chain rule for forms. Since the Jacobian is independent of $ z $, by taking $ \epsilon>0 $ small enough so that $ \psi $ is injective we get 
	\begin{align*}
		\epsilon^{-d}\Vol_{n}(A+\epsilon LQ) & =\epsilon^{-d}\int_{z\in \epsilon Q}\int_{x\in A}\sum_{I\in{ \binom{[n]}{c}}}s([n]\setminus I) L_{[n]\setminus I}dx_{I}\wedge dz\\
		= & \int_{x\in A}\sum_{I\in{ \binom{[n]}{c}}}s([n]\setminus I) L_{[n]\setminus I}dx_{I}= \sum_{I\in{ \binom{[n]}{c}}} L_{[n]\setminus I}\bm_{I}(A).		\qedhere
	\end{align*}
\end{proof}  

Let us now prove that the needed summation formula holds, from which we will be able to conclude \Cref{thm:main1} and \Cref{thm:main2}.
\begin{prop}[Summation formula]\label{prop: summation formula} 
	For every $f\in\mathcal{S}(\R^{d})$, 
	\[\sum_{x\in\Lambda}\hat{f}(x)=\sum_{k\in\Z^{n}_{\var<d}}\widehat{\bm}_{L}(k)f(L^{t}k).\]
\end{prop}

\begin{proof}
\Cref{thm:deloneiff} shows that $\Lambda$ is Delone, so the left sum is absolutely converging. Using \Cref{cor:imagediscrete}, \Cref{lem:bounded}, and $|\widehat{\bm}_{L}(k)|\le \widehat{\bm}_{L}(0)$, the right sum is also absolutely converging\footnote{We elaborate on that and show that $\sum_{k\in\Z^{n}_{\var<d};|L^{t}k|<R}|\widehat{\bm}_{L}(k)|=O(R^{n})$ in the proof of \Cref{thm:main2} which follows}. To prove the summation formula we approximate $ \bm_{L} $ and the counting measure of $ \Lambda $ with smooth compactly supported functions. We consider $ \bm_{L} $ as a singular measure on $ \R^{n}$ that is supported on $\logX$, i.e. $ \bm_{L}(A):=\bm_{L}(A\cap\logX) $ for every Borel set $A\subseteq\R^{n}$. For small enough $ \epsilon>0 $ let $ g_{\epsilon}\colon\R^{n}\to [0,\epsilon^{-d}] $ be a smooth $ \Z^{n} $-periodic function supported in an $ \epsilon^2 $ neighborhood of $ Y(\R)+\epsilon LQ $ that is equal to $ \epsilon^{-d} $ on $Y(\R)+\epsilon LQ $. Then $ g_{\epsilon}d\Vol_{n} $ converges in the vague topology to $ \bm_{L} $, i.e. for every continuous and compactly supported $ h\colon\R^{n}\to \R $,
	 \[\lim_{\epsilon\to 0} \int_{\R^{n}} h g_{\epsilon}d\Vol_{n}=\int_{\logX}hd\bm_{L}.  \]
	Since $ g_{\epsilon} $ is $\Z^{n}$-periodic it has a Fourier series $ g_{\epsilon}(y)=\sum_{k\in\Z^{n}}\hat{g}_{\epsilon}(k)e^{2\pi i\langle k,y\rangle} $ with $ \hat{g}_{\epsilon}(k)=\int_{[0,1]^{n}}e^{-2\pi i\langle x, k\rangle}g_{\epsilon}(x)d\Vol_{n}$, and since $g_{\epsilon}$ is smooth, $\hat{g}_{\epsilon}(k)$ is fast decaying as $|k|\to\infty$ for fixed $\epsilon$. On the other hand, for every fixed $k\in\Z^{n}$,
 \begin{equation}\label{eq:equationtowhichmariohasacomment}
	\lim_{\epsilon\to 0}\hat{g}_{\epsilon}(k)=\widehat{\bm}_{L}(k).    
	\end{equation}
	Similarly, the function $ x\mapsto g_{\epsilon}(Lx) $ is supported in an $ O(\epsilon^{2}) $ neighborhood of $ \Lambda+\epsilon Q $ and is equal to $ \epsilon^{-d}=\Vol_{d}(\epsilon Q) $ on $ \Lambda+\epsilon Q $. So for every $ f\in\mathcal{S}(\R^{d}) $, 
	\[\lim_{\epsilon\to 0} \int_{\R^{d}} f(x) g_{\epsilon}(Lx)dx=\sum_{x\in\Lambda}f(x).  \]
	The Fourier series of $ x\mapsto g_{\epsilon}(Lx) $ is induced from that of $g_{\epsilon}$, 
	\[g_{\epsilon}(Lx)=\sum_{k\in\Z^{n}}\hat{g}_{\epsilon}(k)e^{2\pi i\langle k,Lx\rangle}=\sum_{k\in\Z^{n}}\hat{g}_{\epsilon}(k)e^{2\pi i\langle L^{t}k,x\rangle}. \]
	For fixed $ \epsilon>0 $ and for every $ f\in\cS(\R^{d}) $ we have 
 \begin{align*}
 	\int_{\R^{d}} \hat{f}(x) g_{\epsilon}(Lx)dx= & \int_{\R^{d}} \hat{f}(x)\sum_{k\in\Z^{n}}\hat{g}_{\epsilon}(k)e^{2\pi i\langle L^{t}k,x\rangle}dx\\
 	= & \sum_{k\in\Z^{n}}\hat{g}_{\epsilon}(k)\int_{\R^{d}} \hat{f}(x)e^{2\pi i\langle L^{t}k,x\rangle}dx\\
 	 = & \sum_{k\in\Z^{n}}\hat{g}_{\epsilon}(k)f(L^{t}k),
 \end{align*}
where we used that $ f\in\cS(\R^{d}) $ and $ \hat{g}_{\epsilon}(k) $ is fast decaying in $k$ to interchange summation and integration. Using the dominant converging theorem, and $ \lim_{\epsilon\to 0}\hat{g}_{\epsilon}(k)=\widehat{\bm}_{L}(k) $, we can take the limit, 
	\[\sum_{k\in\Z^{n}}\widehat{\bm}_{L}(k)f(L^{t}k)=\lim_{\epsilon\to 0}\sum_{k\in\Z^{n}}\hat{g}_{\epsilon}(k)f(L^{t}k)=\lim_{\epsilon\to 0}\int_{\R^{d}} \hat{f}(x) g_{\epsilon}(Lx)dx= \sum_{x\in\Lambda}\hat{f}(x). \qedhere  \]
\end{proof}
\begin{rem}
	It is worth noting that $ \widehat{\bm}_{L} $ is in $\ell_{\infty}(\Z^{n})$ but, as seen in the proof above, not in $ \ell_2(\Z^{n}) $. By Parseval's theorem, $ \sum_{k\in\Z^{n}}|\hat{g}_{\epsilon}(k)|^{2}=\int_{[0,1]^{n}}|g_{\epsilon}|^{2}d\Vol_{n}\approx \epsilon^{-d}\widehat{\bm}_{L}(0)\to \infty$ as $\epsilon\to 0$. 
	
\end{rem}
We are now in position to prove \Cref{thm:main1} and \Cref{thm:main2}

\begin{proof}[Proof of \Cref{thm:main1}]
    The set $\Lambda$ is real by \Cref{prop:RR}. The set $\Lambda$ is Delone by \Cref{thm:deloneiff}. \Cref{cor:imagediscrete} states that the set 
    	\begin{equation}\label{eq:lprime}
		\Lambda'=\{L^{t}k\mid k\in\Z^{n},\ \var(k)<d\}
		\end{equation} 
  is discrete. Moreover, \Cref{cor:imagediscrete} shows that, for every $\xi\in\Lambda'$, the set of all $k\in\Z^{n}_{\var<d}$ with $L^{t}k=\xi$ is finite. Therefore, we can define the coefficient
  \begin{equation}\label{eq:cxi}
      c_{\xi}=\sum_{k\in\Z^{n}_{\var<d},\ L^{t}k=\xi}\widehat{\bm}_{L}(k)
  \end{equation}
   for every $\xi\in\Lambda'$. Now \Cref{prop: summation formula} implies that  for every Schwartz function $f\in\mathcal{S}(\R^{d})$, we have
	\begin{equation*}
		\sum_{x\in\Lambda}\hat{f}(x)=  \sum_{\xi\in\Lambda'}c_{\xi}f(\xi).
	\end{equation*}
  This proves that $\Lambda$ has property (1) of \Cref{def:FQ} of a Fourier quasicrystal.
  It remains to prove the polynomial growth bounds from part (2) of \Cref{def:FQ}. Because $\Lambda$ is Delone, we have $|\Lambda\cap B_{R}(0)|=O(R^{d})$. 
  Furthermore, using that $|\widehat{\bm}_{L}(k)|\le\int|e^{-2\pi i\langle x,k\rangle}|d\bm_{L}(x)=\widehat{\bm}_{L}(0)$, \Cref{cor:imagediscrete} gives 
	 \begin{equation*}
	 \sum_{\xi\in\Lambda'\cap B_R(0)}|c_{\xi}|\le \widehat{\bm}_{L}(0)\cdot|\{k\in\Z^{n}_{\var<d}\mid |L^{t}k|<R\}| =O(R^{n}).   
	 \end{equation*}
  This proves that $\Lambda$ is a Fourier quasicrystal with spectrum contained in $\Lambda'$.
\end{proof}

\begin{proof}[Proof of \Cref{thm:main2}]
    The set $\Lambda$ is Bohr almost periodic by \Cref{lem:almostperiodic}.
    The set $\Lambda'\subseteq\R^d$ and the coefficients $c_\xi\in\C$ for $\xi\in\Lambda'$ are given as above in \Cref{eq:lprime} and \Cref{eq:cxi}, respectively.
    \Cref{lem:bounded} implies that $k=0$ is the only solution to $L^{t}k=0$ with $k\in\Z^{n}_{\var<d} $, so $c_{0}=\widehat{\bm}_{L}(0)$. Therefore, \Cref{lem: measL} shows that
    \begin{equation*}
        c_0=\sum_{I\in{\binom{[n]}{d}}}L_{I}\cdot d_{I}.
    \end{equation*}
     The remaining statement of part (1) in \Cref{thm:main2}, namely that $c_0$ is the density of $\Lambda$, will be shown in \Cref{thm:autocorr} using that $\Lambda$ is a Fourier quasicrystal by \Cref{thm:main1}. For part (2) note that by definition of $\Lambda'$ it is clear that $\dim_{\Q}(\Lambda')\leq n$. The statement on its growth rate follows from \Cref{cor:imagediscrete}. For part (3) we let $c_{L,k}=\widehat{\bm}_{L}(k)$ for every $k\in\Z^n$. \Cref{eq:cxi} and \Cref{lem: measL} imply the formula for $c_{\xi}$; the sum in question being finite was already observed in the proof of \Cref{thm:main1} above. The stated dependence of $c_{L,k}$ on $L$ also follows \Cref{lem: measL}. Finally
     \begin{equation*}
         |c_{L,k}|=|\widehat{\bm}_{L}(k)|\le\int|e^{-2\pi i\langle x,k\rangle}|d\bm_{L}(x)=\widehat{\bm}_{L}(0)=c_0
     \end{equation*}
     gives the desired bound.
\end{proof}

\section{Reduction to strict Lee--Yang varieties}\label{sec:reduct}
{This section is motivated by \cite{AlonCohenVinzant}, whose main result can be stated as follows: If $X\subseteq(\C^*)^n$ is an algebraic hypersurface, i.e. of pure codimension one, and $\ell\in\R^n$ such that $\Lambda(X,\ell)$ is real, then there is a Lee--Yang hypersurface $\tilde{X}$ and a vector $\tilde{\ell}$ with positive entries, such that $\Lambda(X,\ell)=\Lambda(\tilde{X},\tilde{\ell}) $. 
The main result of this section, \Cref{thm:coordchange}, is of similar flavor. Namely, it gives sufficient criteria on an algebraic variety $X$ and a matrix $L$ under which there is a strict Lee--Yang variety $\tilde{X}$ and a matrix $\tilde{L}$ with positive maximal minors such that $\Lambda(X,L)=\Lambda(\tilde{X},\tilde{L})$. Before we formulate the precise statement, we prove some preparatory lemmas.}

\begin{lem}\label{lem:maptoopen}
 Let $U\subseteq{\Gr}(d,n)$ be a nonempty open subset. 
 Let $L\in\R^{n\times d}$ be a real $n\times d$ matrix whose range is in $U$,  $L\R^{d}\in U$. 
 There exists a non-singular $n\times n$ matrix $S$ with entries in $\Z$ and the following two properties:
 \begin{enumerate}
     \item The range of $SL$ is in $\Gr_{+}(d,n)$.
     \item The range of $S^{-1}L'$ is in $U$ for every $n\times d$ matrix $L'$ whose range is in $\Gr_{\geq}(d,n)$.
 \end{enumerate}
\end{lem}

\begin{proof}
 If a matrix $S$ satisfies (1) and (2), then every multiple of $S$ by a non-zero scalar also satisfies (1) and (2). Thus, it suffices to find such a matrix $S$ with rational entries. Furthermore, since (1) and (2) are open conditions on $S$, it even suffices to find such a matrix $S$ with entries in $\R$.
 Throughout the proof we will consider $\Gr(d,n)$ as a subset of $\pp(\wedge^d\R^n)$ via the Pl\"ucker embedding. Since $\GL_n(\R)$ acts transitively on $\Gr(d,n)$, we can assume without loss of generality that the range of $L$ is spanned by $e_1,\ldots,e_d$. We observe that it suffices to prove the claim for any open neighbourhood of the range of $L$ that is contained in $U$. Hence without loss of generality, we can assume that $U$ is contained in the open affine chart of $\pp(\wedge^d\R^n)$ where the Pl\"ucker coordinate corresponding to $[d]$ is not zero. This chart is naturally homeomorphic to $\R^N$ where $N=\binom{n}{d}-1$ and the range of $L$ corresponds to the origin of $\R^N$. Thus by further shrinking $U$ we can assume that $U$ is the intersection of $\Gr(d,n)$ with an open $\epsilon$-ball $B$ around the origin in this chart. Under the natural projection
 \begin{equation*}
     \wedge^d\R^n\setminus\{0\}\to\pp(\wedge^d\R^n)
 \end{equation*}
the ball $B$ is the image of the open convex cone
\begin{equation*}
    C=\left\{\sum_{I\in\binom{[n]}{d}}\lambda_I e_I\mid \sum_{I\neq[d]}\lambda_I^2<\epsilon\cdot\lambda_{[d]}^2,\, \lambda_{[d]}>0\right\}\subseteq\wedge^d\R^n
\end{equation*}
where 
$e_I=e_{i_1}\wedge\cdots\wedge e_{i_d}$ if $I=\{i_1,\ldots,i_d\}$ and $1\leq i_1<\cdots<i_d\leq n$.
 
Let $A$ be a $d\times n$ matrix all of whose $d\times d$ minors are positive. Let $A_{1}$ and $A_{2}$ be the $d\times(n-d)$ and $d\times d$ matrices such that $A=\left(\begin{array}{@{}c|c@{}}A_{1} & A_{2}\end{array}\right)$, and define the $n\times n$ matrix $T_{\nu}$, in the block decomposition $\R^n=\R^{n-d}\times\R^{d}$, as follows  
		\[
		T_{\nu}=\left(\begin{array}{@{}c|c@{}}
			A_{1}
			& A_{2} \\
			\hline
			0_{n-d\times d} &
			\nu\cdot I_{n-d}
		\end{array}\right),
		\]
 where $I_{n-d}$ is the $(n-d)\times(n-d)$ identity matrix and $0_{n-d\times d}$ the $n-d\times d$ zero matrix. Clearly, for $\nu\neq0$, the $n\times n$ matrix $T_\nu$ is invertible. On the other hand, for every $I\in\binom{[n]}{d}$ the linear subspace spanned by all $T_0(e_i)$ for $i\in I$ is just the range of $L$ and therefore in $B$ for every $I\in\binom{[n]}{d}$. Hence for small enough $\nu>0$ the linear subspace spanned by all $T_\nu(e_i)$ for $i\in I$ is also in $B$. By our positivity assumption on $A$, we even have that the induced linear map 
 \begin{equation*}
  \wedge^dT_\nu\colon\wedge^d\R^n\to\wedge^d\R^n  
 \end{equation*}
 sends $e_I$ to $C$ for all $I\in\binom{[n]}{d}$. As this is an open condition, it is also satisfied by
	\[
S=\left(\begin{array}{@{}c|c@{}}
	A_{1}
	& A_{2} \\
	\hline
	R &
	\nu\cdot I_{n-d}
\end{array}\right),
\]
 where $R$ is an $(n-d)\times d$ matrix sufficiently close to the zero matrix so that the span of the first $d$ columns of $S$ is in $\Gr_{+}(d,n)$. Such a matrix exists because the closure of $\Gr_{+}(d,n)$ is $\Gr_{\geq}(d,n)$ \cite{Karp}*{Theorem~3.15(ii)}. 
 Since $C$ is a convex cone, every nontrivial nonnegative linear combination of the $e_I$ is also mapped to $C$ by $\wedge^dS$. This implies in particular that $\Gr_{\geq}(d,n)$ is mapped to a subset of $U$. Hence the range of $S L$ is in $\Gr_{+}(d,n)$ and $S^{-1}L'\in U$ for every $d\times n$ matrix whose range is in $\Gr_{\geq}(d,n)$.
\end{proof}

We apply such coordinate changes to the \emph{amoeba} of the variety $X$. Formally, 
the \emph{amoeba} of an algebraic variety $X\subseteq(\C^*)^n$, denoted $\cA(X)$, is the image of $X$ under the map $(\C^*)^n \to \R^n$ taking $(z_1, \hdots, z_n)$ to $(\log|z_1|, \hdots, \log|z_n|)$. 
That is, $\cA(X)=\left\{a\in\R^{n}\mid \exists b \in\R^{n}\ \text{s.t.}\ \exp(a+ib)\in X\right\}$. See  \cites{GKZ,MikhalkinAmoeba} for more detail.

\begin{lem}\label{lem:amoeba1}
    Let $X\subseteq(\C^*)^n$ be a closed equidimensional algebraic subvariety of dimension $c=n-d$ which is invariant under the coordinate-wise involution $z\mapsto1/\bar{z}$ and $\cA(X)\subseteq\R^n$ be its amoeba. If $S\in\Z^{n\times n}$ is a nonsingular integer valued $n\times n$ matrix, then the set $X'=\{\exp(2\pi i Sx)\mid x\in\C^{n}, \exp(2\pi i x)\in X\}$ is a closed equidimensional algebraic subvariety of $(\C^*)^n$ of dimension $c=n-d$ which is invariant under the coordinate-wise involution $z\mapsto1/\bar{z}$. In particular, we have $\cA(X')=S\cA(X)$.
\end{lem}

\begin{proof}
    Let $S=(s_{ij})_{ij}$ with $s_{ij}\in\Z$. We consider the morphism
    \begin{equation*}
        f_S\colon(\C^*)^n\to(\C^*)^n,\, (z_1,\ldots,z_n)\mapsto(z_1^{s_{11}}\cdots z_n^{s_{1n}},\ldots,z_1^{s_{n1}}\cdots z_n^{s_{nn}}).
    \end{equation*}
    The set $X'=f_S(X)$ is invariant under the coordinate-wise involution $z\mapsto1/\bar{z}$ and  clearly satisfies $\exp(2\pi i x)\in X \Leftrightarrow \exp(2\pi i Sx)\in X'$ for all $x\in\C^n$. Hence it suffices to show that $f_S$ is a finite map. We can write $S$ in Smith normal form $S=RDT$ where $R,T$ are invertible integer matrices and $D$ is a diagonal integer matrix. The morphisms $f_R$ and $f_T$ are automorphisms of $(\C^*)^n$ so in particular finite. The morphism $f_D$ can be decomposed into morphisms that take the $m$-th power of the $i$-th entry and leave the remaining entries unchanged. These are clearly finite. Finally, the statement that $\cA(X')=S\cA(X)$ follows directly from $\exp(2\pi i x)\in X \Leftrightarrow \exp(2\pi i Sx)\in X'$ by taking the logarithm of the absolute value of both sides.
\end{proof}

\begin{lem}\label{lem:amoeba2}
     Let $X\subseteq(\C^*)^n$ be a closed algebraic subvariety of dimension $c=n-d$ which is invariant under the coordinate-wise involution $z\mapsto1/\bar{z}$ and $\cA(X)\subseteq\R^n$ be its amoeba. 
     If there exists an open subset $U\subseteq\Gr(d,n)$ such that
     \begin{enumerate}
     	\item[(i)] $\Gr_{\geq}(d,n)\subseteq U$, and
     	\item[(ii)] $\cA(X)\cap V=\{0\}$ for every $V\in U$,
     \end{enumerate} 
     then the closure $\bar{X}$ of $X$ in $(\pp^1)^n$ satisfies conditions (1) and (2) in \Cref{def:leeyang}. (Condition (3), that $X(\T)$ is smooth, may not be satisfied.)
\end{lem}

\begin{proof}
	We need to show that $\bar{X}$ is invariant under the involution $z\mapsto1/\bar{z}$ (condition (1) in \Cref{def:leeyang}) and that every point $z\in\bar{X} $, which is not in $\T^{n}$, has $\var(\log |z|)\ge d$ (condition (2) in \Cref{def:leeyang}). 
	
	For condition (1), by assumption, $X$ is invariant under the involution, so it also takes its closure $\bar{X}$ to itself, as needed. For condition (2), \Cref{lem:Karp} implies that every $0\neq v\in\R^n$ with  $\var(v) < d$ is contained in a linear subspace $V\in\Gr_{\geq}(d,n)$, so (i) implies $V\in U$ and we conclude that $v\notin \cA(X)$ from (ii). This means that every $z\in X\setminus\T^{n}$ has $\var(\log|z|)\ge d$, and we are left with showing that this is also the case for all $z\in \bar{X}\setminus X$.

    Assume that this is not the case, i.e.  there exists $z\in \bar{X}\setminus X$ with $\var(\log|z|)<d$, so there is a sequence $(z_i)_{i\in\N}\subseteq X\setminus \T^{n}$ converging to $z$. Consider the sequence $(a_i)_{i\in\N}$ of points on the sphere 
    \begin{equation*}
        a_i=\frac{\log|z_i|}{\|(\log|z_i|)\|}\in \mathbb{S}^{n-1}.
    \end{equation*}
    After passing to a subsequence if necessary, the sequence converges to the limit point $\lim_{i\to\infty}a_{i}=a\in\mathbb{S}^{n-1} $. Since every non-zero entry of $a$ has the same sign as the corresponding entry of $\log|z|$, we obtain 
    \begin{equation*}
        \var(a)\leq\var(\log|z|)<d.
    \end{equation*}
    By \Cref{lem:Karp} there exists $V\in\Gr_{\geq}(d,n)$ which contains $a$, and $V\in U $ by (i).
    We complete $a$ by $b_1,\ldots,b_{d-1}$ to a basis of $V$, and define $V_{i}=\mathrm{span}(a_{i},b_{1},\ldots,b_{d-1})$ for all $i\in\N$, so that $V_{i}\in \Gr(d,n)$ for large enough $i$ and so that the sequence $(V_{i})_{i\in\N}$ converges (as points in $\Gr(d,n)$) to $V\in U$. Therefore $V_{i}\in U$ for some $i\in\N$. By construction, $\log|z_i|\ne 0$ and $\log|z_i|\in V_i\cap\cA(X)$, in contradiction to (ii). We conclude that that $\bar{X}$ satisfies (2).
\end{proof}

\begin{thm}\label{thm:coordchange}
    Let $X\subseteq(\C^*)^n$ be a closed equidimensional subvariety of dimension $c=n-d$ which is invariant under the coordinate-wise involution $z\mapsto1/\bar{z}$. Let $L$ be an $n\times d$ matrix such that the following holds:
    \begin{enumerate}
        \item There is an open subset $U\subseteq{\Gr}(d,n)$ which contains $L\R^{d}$, the range of $L$, such that $\cA(X)\cap V=\{0\}$ for every $V\in U$.
    \item The set
        $\Lambda(X,L)$
    is Delone.
    \end{enumerate}
     Then there is a strict Lee--Yang variety $\tilde{X}\subseteq(\pp^1)^n$ and an $n\times d$ matrix $\tilde{L}$ whose range is in $\Gr_{+}(d,n)$ such that 
    $\Lambda(X,L)=\Lambda(\tilde{X},\tilde{L})$.
    In particular, the Delone set $\Lambda(X,L)$ is a Fourier quasicrystal.
\end{thm}

\begin{proof}
    First we note that after replacing $(\C^*)^n$ by a suitable algebraic subtorus and $X$ with the intersection of $X$ with this subtorus, we can assume without loss of generality that the rows of $L$ are linearly independent over $\Q$.
    By \Cref{lem:maptoopen} there exists a non-singular $n\times n$ matrix $S$ with entries in $\Z$ such that the range of $SL$ is in $\Gr_{+}(d,n)$ and the range of $S^{-1}L'$ is in $U$ for every $n\times d$ matrix $L'$ whose range is in $\Gr_{\geq}(d,n)$. Let $U'$ be the image of $U$ under the automorphisms $\Gr(d,n)\to\Gr(d,n)$ induced by $S$. Then $U'$ is an open subset of $\Gr(d,n)$ which contains $\Gr_{\geq}(d,n)$ such that the range of $S^{-1}L'$ is in $U$ for every $n\times d$ matrix whose range is in $U'$. The latter means that elements of $U'$ intersect $S\cA(X)$ only in the origin. By \Cref{lem:amoeba1} there exists a closed algebraic subvariety $X'\subseteq(\C^*)^n$ of codimension $d$ which is invariant under the coordinate-wise involution $z\mapsto1/\bar{z}$ such that $\exp(2\pi i x)\in X \Leftrightarrow \exp(2\pi i Sx)\in X'$ for all $x\in\C^n$ and whose amoeba is $\cA(X')=S\cA(X)$. By \Cref{lem:amoeba2} the closure $\tilde{X}$ of $X'$ in $(\pp^1)^n$ satisfies conditions (1) and (2) in \Cref{def:leeyang}. By construction we have
    $\Lambda(X,L)=\Lambda(\tilde{X},\tilde{L})$
    where $\tilde{L}=SL$.
    Finally, because $\Lambda$ is a Delone set, the torus part of $\tilde{X}$ is smooth by \Cref{thm:deloneiff} which implies the claim.
\end{proof}

\begin{cor}\label{cor:leeyangfqc}
    Let $X\subseteq(\pp^1)^n$ be a Lee--Yang variety of codimension $d$ such that $X(\T)$ is contained in the smooth part of $X$. Further let $L\in\R^{n\times d}$ be a real matrix all of whose $d\times d$ minors are positive.  Then 
 $\Lambda(X,L)$
is a Fourier quasicrystal and a Delone set.
\end{cor}

\begin{proof}
    By \Cref{prop:RR} the variety $X\cap(\C^*)^n$ satisfies condition (1) of \Cref{thm:coordchange} because $\Gr_{+}(d,n)$ is open. By \Cref{thm:deloneiff} it also satisfies part (2).
\end{proof}

\begin{cor}
   Let $\Lambda\subseteq\R^n$ be a Delone Fourier quasicrystal obtained as in \cite{lawton2024fourier}*{Theorem~3}. There is a strict Lee--Yang variety $\tilde{X}\subseteq(\pp^1)^n$ and an $n\times d$ matrix $\tilde{L}$ whose range is in $\Gr_{+}(d,n)$ such that $\Lambda=\Lambda(\tilde{X}, \tilde{L})$.
\end{cor}

\begin{proof}
    This follows from \Cref{thm:coordchange}. Indeed, by \cite{lawton2024fourier}*{Theorem~2} there is a  closed subvariety $X\subseteq(\C^*)^n$ of codimension $d$, cut out by $d$ equations, and an $n\times d$ matrix $L$ satisfying condition (1) in \Cref{thm:coordchange} such that
    $\Lambda=\Lambda(X,L)$. The variety $X$ is equidimensional because it is a complete intersection.
    Condition (2) is satisfied by assumption. Finally, after replacing $X$ by the Zariski closure of $X(\T)$, we can also assume that it is invariant under the coordinate-wise involution $z\mapsto1/\bar{z}$.
\end{proof}
\section{Genuine high-dimensionality} \label{Sec:non-triv}
The goal of this section is to prove that with our construction we can obtain Fourier quasicrystals that are genuinely high-dimensional in the sense that their supports do not contain one-dimensional Fourier quasicrystals nor subsets that are assymptotically close to such. This implies in particular that in $\R^2$ and $\R^3$ these Fourier quasicrystals do not contain the product of two lower dimensional Fourier quasicrystals. 

\begin{Def}[\cite{lawton-inverse}]
	A Fourier quasicrystal $\Lambda\subseteq\R^d$ is of \emph{toral type} if its spectrum is contained in a finitely generated subgroup of $\R^d$.
\end{Def}

\begin{ex}
	By \cite{olevskii2020fourier} every Fourier quasicrystal $\Lambda\subseteq\R$ is of toral type.
\end{ex}

\begin{lem}\label{lem:codimlemma}
	Let $\Gamma\subseteq\R^m$ be a Fourier quasicrystal of toral type which is also Delone. Let $M\in\R^{n\times m}$ be an $n\times m$ matrix such that $\psi(x)=\exp(2\pi i Mx)$ is a homomorphism from $\R^{m}$ to $\T^{n}$ with a dense image. Then the Zariski closure of 
	\[\psi(\Gamma)=\{\exp(2\pi i M x)\mid x\in\Gamma\}\]
in $\T^n$ has dimension at least $n-m$.
\end{lem}
\begin{proof}
    Let $\Gamma'$ be the spectrum of $\Gamma$ and consider the subgroup $G$ of $\R^m$ generated by the elements of $\Gamma'$ and the rows of $M$. The subgroup $G$ is finitely generated since $ \Gamma$ is toral and we denote its rank by $l\in\N$. 
    We define a real matrix $M_1\in\R^{m\times l}$ whose columns comprise a $\Z$-basis of $G$. Then there is a unique integer valued matrix $M_2\in\Z^{l\times n}$ such that $M=M_{2}^{t}M_{1}^{t}$. We consider the group homomorphisms $\psi_1\colon\R^m\to\T^l$ defined by $x\mapsto\exp(2\pi i M_{1}^{t}x)$ and $\psi_2\colon\T^l\to\T^n$ defined by 
    \begin{equation*}
     \psi_{2}(z)= z^{M_{2}^t}=\left(\prod_{i=1}^{l}z_{i}^{(M_{2})_{i,1}},\ldots,\prod_{i=1}^{l}z_{i}^{(M_{2})_{i,n}}\right),   
    \end{equation*}
    so that $\psi(x)=(\psi_{2}\circ\psi_{1})(x)=\exp(2\pi i M x)$. 
    The map $\Z^l\to\R^m,\, k\mapsto M_1k$ is injective by construction (since the columns of $M_{1}$ are a $\Z$-basis of $G$), so that by Pontryagin duality, see for example \cite{stroppel}*{Proposition~23.2}, the image of $\psi_1$ is dense in $\T^l$.
    It follows from \cite{lawton2024fourier}*{Proposition~2} that $\Gamma$ is Bohr almost periodic, so that by \cite{lawton-inverse}*{Theorem~5} the closure of $\psi_1(\Gamma)$ in $\T^l$ is a union of $l-m$ dimensional tori. In particular, the Zariski closure $Z$ of $\psi_1(\Gamma)$ has dimension at least $l-m$. Because $\psi$ has dense image, the group homomorphism $\psi_2$ is surjective. Thus $l\geq n$ and the fibers of $\psi_2$ are all of dimension $l-n$. This shows that the dimension of the Zariski closure of $\psi_2(Z)$ is at least 
	\begin{equation*}
		\dim(Z)-(l-n) \geq (l-m)-(l-n)=n-m.  
	\end{equation*}
	This proves the claim.
\end{proof}
	\begin{Def}
        If $W\subseteq \R^{d}$ is an affine $m$-dimensional subspace, we call a set $\Gamma\subseteq W$ a \emph{Fourier quasicrystal in $W$} if it is isometric to a Fourier quasicrystal $\Lambda\subseteq\R^{m}$, and we say $\Gamma$ is of toral type if $\Lambda$ is.
	\end{Def}
\begin{thm}\label{thm:nontriv23}Let $\Lambda=\Lambda(X,L)\subseteq\R^{d}$ as in \Cref{thm:main1}, with the further assumptions that:
	\begin{enumerate}
		\item The variety $X$ is a curve, i.e. $\dim(X)=1$, and so $d=n-1$. 
		\item Furthermore, $X$ is irreducible and $X(\T)=X\cap\T^{n}$ is not contained in a coset of a proper subtorus of $ \T^{n}$.
		\item The $d\times d$ minors of $L$ are $\Q$-linearly independent.
	\end{enumerate}
	Then, there is no affine subspace $W\subsetneq\R^d$ such that $W\cap\Lambda$ contains a Fourier quasicrystal of toral type in $W$. 
\end{thm}

\begin{proof}
	Assume that this is not the case, so that there are $m\in\{1,\ldots,d-1\},\ a\in\R^{d},\ M\in\R^{d\times m}$ of full rank, and a Fourier quasicrystal of toral type $\Gamma\subseteq\R^{m}$ such that $\{My+a\mid y\in\Gamma\}\subseteq \Lambda$. Multiplying $X$ by $ \exp(-2\pi i La)\in\T^{n}$ provides a new Lee--Yang curve $X'$ that also satisfies our assumptions such that $\Lambda(X',L)=\Lambda-a$. Thus we may assume that $a=0$. 
    We consider the group homomorphism
    \begin{equation*}
        \psi\colon\R^m\to\T^n,\, y\mapsto \exp(2\pi i LMy).
    \end{equation*}
    We will show that $\psi$ has dense image.
    By construction we have that
	\[\psi(\Gamma)=\{\exp(2\pi i LM y)\mid y\in\Gamma\}\subseteq X(\T).\]
	Since $\R^d\to\T^n,\, x\mapsto \exp(2\pi i Lx)$ is injective, by \Cref{lem:assL}, and $M$ has full rank, it follows that $\psi(\Gamma)$ is infinite. Therefore, its Zariski closure in $\T^{n}$ is 
    \begin{equation*}
        \overline{\psi(\Gamma)}^{\rm \,zar}=X(\T)
    \end{equation*}
    because $X$ is irreducible and one-dimensional.
	On the other hand, the closure $T_{M}=\overline{\psi(\R^{m})}\subseteq\T^{n}$ of the image of $\psi$ is a subtorus, so it is Zariski closed in $\T^{n}$. Because $T_M$ contains $\psi(\Gamma)$ and thus its Zariski closure $X(\T)$, our assumption (2) implies that $\T_{M}=\T^{n}$. This shows that $\psi$ has a dense image in $\T^{n}$. Now \Cref{lem:codimlemma} says that $\overline{\psi(\Gamma)}^{\rm \,zar}$ has dimension at least $n-m\ge 2$, in contradiction to $\dim(X(\T))=1$.
\end{proof}

Because every one-dimensional Fourier quasicrystal is of toral type by \cite{olevskii2020fourier}, we obtain the following.

\begin{cor}\label{cor:no1dimfqc}
	Consider $\Lambda=\Lambda(X,L)$ as in \Cref{thm:nontriv23}. Then, there is no one-dimensional affine subspace $W\subsetneq\R^n$ such that $W\cap\Lambda$ contains a Fourier quasicrystal in $W$.
\end{cor}

\begin{rem}
	In the proof of \Cref{thm:nontriv23} the Zariski closure of $\psi(\Gamma)$ does not change if we remove finitely many points. Thus in \Cref{cor:no1dimfqc} we can even deduce that there is no one-dimensional affine subspace $W\subsetneq\R^n$ and no finite set $S\subseteq W$ such that $(W\cap\Lambda)\cup S$ contains a Fourier quasicrystal in $W$. In \Cref{cor:no1dimfqc-asymp} we will prove an even stronger statement.
\end{rem}

\begin{cor} \label{cor:no1dimMain}
	Consider $\Lambda=\Lambda(X,L)$ as in \Cref{thm:nontriv23} and further assume that $d\in\{2,3\}$. Then $\Lambda$
	does not contain a set which is the product of two lower dimensional Fourier quasicrystals or a translate thereof.
\end{cor}

\begin{proof}
	If $\Lambda$ contained the translate of a product of two lower dimensional Fourier quasicrystals, then at least one of these lower dimensional Fourier quasicrystals would be one-dimensional and thus contained in some one-dimensional affine subspace $W$ of $\R^n$. Now claim follows from \Cref{cor:no1dimfqc}.
\end{proof}

Next we show that $ W \cap \Lambda $ cannot even contain any discrete set which is asymptotically close to a Fourier quasicrystal.

\begin{Def}
	Two discrete sets $ \{ s_n \}_{n= - \infty}^\infty $ and $ \{ k_m \}_{m= - \infty}^\infty $ enumerated monotonically 
	are called \emph{asymptotically close}
	if  there exist $ N_\pm \in \mathbb N $ such that
	$$ \lim_{n \rightarrow \pm \infty} (s_n - k_{n+N_\pm}) = 0 .$$
\end{Def}

Let us remind that two zero sets of analytic almost periodic functions, that are asymptotically close,
coincide:
\begin{thm}[\cite{KuSu20}*{Theorem~5}] \label{ThKuSu}
	Let $k_n$ and $l_n$ inside the strip $ \big\{z = x + iy : |y| < h \big\} $ be zeroes of two almost periodic functions
	$f_1$ and $f_2$, respectively, both holomorphic in a slightly larger strip $ \big\{z = x + iy : |y| < H \big\}, h < H$. If there
	exists a subsequence $l_{n_m}$ of $l_n$ such that
	\begin{equation} \label{cond}
		\lim_{n \rightarrow \infty} (k_n - l_{n_m} ) = 0, 
	\end{equation}
	then all the zeroes of $f_1$ inside the smaller strip are zeroes of $f_2$ with at least the same multiplicity.
\end{thm}

\begin{cor}\label{cor:no1dimfqc-asymp}
	If $\Lambda=\Lambda(X,L)\subseteq\R^{d}$ as in \Cref{thm:nontriv23}, and $d\in\{2,3\}$, then 
	there is no one-dimensional affine subspace $W\subsetneq\R^n$ such that $W\cap\Lambda$ is asymptotically close to
	a Fourier quasicrystal in $W$.
\end{cor}

\begin{proof}
Assume for the sake of a contradiction that there exists a one-dimensional affine subspace $W\subsetneq\R^n$ and a Fourier quasicrystal $\tilde{\Lambda}$ in $W$ which is asymptotically close to $W\cap\Lambda$.
Following  \cite{olevskii2020fourier} there exists a trigonometric polynomial whose zeroes coincides with $\tilde{\Lambda}$.
The set $ W \cap \Lambda $ is a common zero set for several trigonometric polynomials obtained as the restrictions of
the polynomials determining $ X $ to the curve $ \exp (2 \pi i L W)$. 
It is enough to consider zeroes on the real line.
Then Theorem \ref{ThKuSu} implies that
$ W \cap \Lambda $ coincides with $ \tilde{\Lambda}$. We may apply Corollary \ref{cor:no1dimfqc}.
\end{proof}

\section{Auto-correlation and the diffraction measure}\label{sec: autocorrelation}
Here we prove \Cref{thm: HU and AC}. For the reader’s convenience, we present a more explicit version.   
\begin{thm}\label{thm:autocorr}
Suppose that a set $ \Lambda\subseteq\R^{d}$ is a Fourier quasicrystal, with spectrum $\Lambda'$ and Fourier coefficients $ (c_{\xi})_{\xi\in\Lambda'} $, as in \Cref{def:FQ}. 
	Then, using the notation $N_{R}(x)=|\Lambda\cap B_{R}(x)|$, we have:
	\begin{enumerate}
		\item The auto-correlation $ \gamma $ of $ \Lambda $ exists and the diffraction measure is equal to $ \widehat{\gamma}=\sum_{\xi\in\Lambda'}|c_{\xi}|^{2}\delta_{k} $. Namely, for every $ f\in\mathcal{S}(\R^{d}) $,
		\[\lim_{R\to\infty}\frac{1}{\mathrm{Vol}(B_{R})}\sum_{x,y\in\Lambda\cap B_{R}}\hat{f}(x-y)=\sum_{\xi\in\Lambda'}|c_{\xi}|^{2}f(\xi).\]
				\item The Fourier coefficient $c_{0}$ is real, positive, and is the density of $\Lambda$
		\[\lim_{R\to\infty}\frac{N_{R}(0)}{\Vol(B_{R})}=c_{0}> 0,\]
		and there exists $C'>0$ such that for all $R>1$ \[\sup_{x\in\R^{d}}\left|N_{R}(x)-c_{0}\Vol(B_{R})\right|\le C' R^{d-1}.\]
	\end{enumerate} 
\end{thm}
\begin{proof} We first prove part (2) and then part (1). To see that $c_{0}>0$, let $\epsilon>0$ such that $|\xi|>\epsilon$ for all $\xi\in\Lambda'\setminus\{0\}$. Let $f$ be a smooth function supported inside $ B_{\epsilon/2}(0)$ and normalized in $L_{2}$. We can make sure that $\hat{f}(x)\ne 0$ for some $x\in\Lambda$ by changing $f(\xi)\mapsto f(\xi)e^{2\pi i\langle \xi,x_{0}\rangle}$ so that $\hat{f}(x)\mapsto\hat{f}(x-x_{0}) $, for appropriate shift $x_{0}$. Let $\tilde{f}(x)=\overline{f(-x)}$, so that the convolution $\phi=f*\tilde{f}$ is smooth, supported inside the ball  $B_{\epsilon}(0)$, and has a non-negative Fourier transform $\hat{\phi}=|\hat{f}|^{2}$ that is non-zero on some point of $\Lambda$. The $L_{2}$ normalization gives $\phi(0)=\int_{\R^{d}}|\hat{f}|^{2}dx=1$, and so the summation formula yields
	\[c_{0}=\sum_{\xi\in\Lambda'}c_{\xi}\phi(\xi)= \sum_{x\in\Lambda}\hat{\phi}(x)>0.\]
	To prove the rest of part (2), denote the inverse Fourier transform by $\check{f}(\xi)=\int_{\R^{d}}e^{2\pi i\langle \xi,x\rangle}f(x)dx=\hat{f}(-\xi)$, so that if $g=\hat{f}$ then  $\check{g}=f$. The summation formula can be written as   
\[		\sum_{x\in\Lambda}f(x) =\sum_{\xi\in\Lambda'}c_{\xi}\check{f}(\xi),\ \text{for all }\ f\in\mathcal{S}(\R^{d}).\]
	Notice that if $f$ is symmetric $f(x)=f(-x)$, then $\check{f}=\hat{f}$.
	We define a family of symmetric non-negative ``bump functions'' $g_{R}\in\mathcal{S}(\R^{d})$ for $R>1$, by
	\[g_{R}=\chi_{R}*\psi,\]
	where $\chi_{R}$ is the indicator function of $B_{R}(0)$ and $\psi\in\mathcal{S}(\R^{d})$ is some fixed non-negative smooth function supported in $B_{1}(0)$, which is normalized in the sense that $\int \psi(x)dx=1$, and symmetric, i.e. $\psi(-x)=\psi(x)$. Since both $\chi_{R}$ and $\psi$ are symmetric, then so does $g_{R}$, and therefore $ \check{g}_{R}(\xi)=\hat{g}_{R}(\xi)=\widehat{\psi}(\xi)\widehat{\chi}_{R}(\xi)$. 
	
	By construction, $ g_{R-1}(x)\le \chi_{R}(x) \le g_{R+1}(x)$ for all $x\in\R^{d}$ and all $R>1$, so the summation over $\Lambda$, shifted by $x_{0}\in\R^{d}$, gives
	\[\sum_{x\in\Lambda}g_{R-1}(x-x_{0})\le N_{R}(x_{0})=\sum_{x\in\Lambda}\chi_{R}(x-x_{0})\le\sum_{y\in\Lambda}g_{R+1}(x-x_{0}).\]
	The inverse Fourier transform of $x\mapsto g_{R}(x-x_{0})$ is $\xi\mapsto e^{2\pi i\langle\xi,x_{0}\rangle}\hat{g}_{R}(\xi)$, so the summation formula gives
	\[\sum_{\xi\in\Lambda'}c_{\xi}e^{2\pi i\langle \xi,x_{0}\rangle}\hat{g}_{R-1}(\xi)\le N_{R}(x_{0})\le\sum_{\xi\in\Lambda'}c_{\xi}e^{2\pi i\langle \xi,x_{0}\rangle}\hat{g}_{R+1}(\xi).\]
Using that $ \hat{g}_{R}(0)=\widehat{\psi}(0)\widehat{\chi}_{R}(0)=\Vol(B_{R})$ for all $R>1$, the above inequalities give 
	\begin{align*}
	\left|N_{R}(x_{0})-c_{0}\Vol(B_{R})\right|	& \le \max_{R'=R\pm 1}\left[c_{0}\left|\Vol(B_{R})-\Vol(B_{R'})\right|+\sum_{\xi\in\Lambda'\setminus\{0\}}\left|c_{\xi}\hat{g}_{R'}(\xi)\right|\right]\\ \le &	c_{0}\left(\Vol(B_{R+1})-\Vol(B_{R-1})\right)+	\max_{R'=R\pm 1}\sum_{\xi\in\Lambda'\setminus\{0\}}\left|c_{\xi}\hat{g}_{R'}(\xi)\right|.
	\end{align*}
	This bound is independent of $x_{0}$. Notice that $\Vol(B_{R+1})-\Vol(B_{R-1})=O(R^{d-1})$. To prove part (2), it is enough to show that $\sum_{\xi\in\Lambda'\setminus\{0\}}|c_{\xi}\hat{g}_{R'}(\xi)|=O(R^{d-1})$. The Fourier transform of the indicator of the unit ball is given by a Bessel function $ \widehat{\chi}_{1}(\xi)=\|\xi\|^{-d/2}J_{d/2}(\|\xi\|) $ which is known to satisfy the bound $ |J_{d/2}(r)|\le C_{J}r^{-\frac{1}{2}} $ for some $C_{J}>0$ (see for example \cite{abramowitz1968handbook}*{p.~364, 9.2.1}). Scaling gives
	\begin{equation*}
		|\widehat{\chi}_{R}(\xi)|=R^{d}|\widehat{\chi}_{1}(R\xi)|\le C_{J} R^{\frac{d-1}{2}}\|\xi\|^{-\frac{d+1}{2}},
	\end{equation*}
so for $\xi\ne0$ we can estimate
\begin{equation}\label{eq: ghat estimate}
	|\hat{g}_{R}(\xi)|=|\widehat{\psi}(\xi)\widehat{\chi}_{R}(\xi)|\le|\widehat{\psi}(\xi)| C_{J} R^{\frac{d-1}{2}}\|\xi\|^{-\frac{d+1}{2}}.
\end{equation}
    Since $\Lambda'$ is discrete there exists $\epsilon>0$ such that $\Lambda'\cap B_{\epsilon}=\{0\}$ which implies $\|\xi\|^{-\frac{d+1}{2}}<\epsilon^{-\frac{d+1}{2}}$ for all $\xi\in\Lambda'\setminus\{0\}$, which in turn gives 
    \[\sum_{\xi\in\Lambda'\setminus\{0\}}|c_{\xi}\hat{g}_{R}(\xi)|< R^{\frac{d-1}{2}}\left[C_{J} \epsilon^{-\frac{d+1}{2}}\sum_{\xi\in\Lambda'\setminus\{0\}}|c_{\xi}\widehat{\psi}(\xi)|\right].
    \]
Since $\psi\in\mathcal{S}(\R^{d})$ and $\Lambda$ is a Fourier quasicrystal, the polynomial growth implies that $\sum_{\xi\in\Lambda'}|c_{\xi}\widehat{\psi}(k)|<\infty$, and therefore $\sum_{\xi\in\Lambda'\setminus\{0\}}|c_{\xi}\hat{g}_{R}(\xi)|\le C'' R^{\frac{d-1}{2}}$ for all $R>1$, which proves part (2), since 
       \[\left|N_{R}(x_{0})-c_{0}\Vol(B_{R})\right|\le c_{0}\left(\Vol(B_{R+1})-\Vol(B_{R-1})\right)+C''(R+1)^{\frac{d-1}{2}}=O(R^{d-1}).\]

	For part (1), notice that $\sum_{x\in\Lambda}g_{R}(x)e^{-2\pi i\langle \xi_{0},x\rangle}=\sum_{\xi\in\Lambda'}c_{\xi}\hat{g}_{R}(\xi-\xi_{0})$. If we define the Fourier coefficient $c_{\xi_{0}}=0$ when $\xi_{0}\notin\Lambda'$, then similarly to part (2), we have 
	\begin{align*}
		&\left|\sum_{x\in\Lambda\cap B_{R}(0)}e^{-2\pi i \langle \xi_{0},x\rangle}-c_{\xi_{0}}\Vol(B_{R})\right|\\  \le &\left|\sum_{x\in\Lambda\cap B_{R}(0)}e^{-2\pi i \langle \xi_{0},x\rangle}-\sum_{x\in\Lambda}g_{R}(x)e^{-2\pi i \langle \xi_{0},x\rangle}\right|+\left|\sum_{\xi\in\Lambda'}c_{\xi}\hat{g}_{R}(\xi-\xi_{0})-c_{\xi_{0}}\Vol(B_{R})\right|\\
		 \le & N_{R+1}(0)-N_{R}(0)+\sum_{\xi\in\Lambda'\setminus\{\xi_{0}\}}\left|c_{\xi}\hat{g}_{R}(\xi-\xi_{0})\right|, 
	\end{align*}
	and the same argument as in part (2), using that $\psi\in\mathcal{S}(\R^{d})$ and that there is some positive lower bound $\|\xi-\xi_{0}\|\ge\epsilon>0$ for all $\xi\in\Lambda'\setminus\{\xi_{0}\}$, we conclude that there is some $C'''>0$ (that may depend on $\xi_{0}$) such that $\sum_{\xi\in\Lambda'\setminus\{\xi_{0}\}}\left|c_{\xi}\hat{g}_{R}(\xi-\xi_{0})\right|\le C'''R^{\frac{d-1}{2}}$ for all $R>1$. This means that summing $e^{-2\pi i \langle \xi,x\rangle}$ over $x\in\Lambda\cap B_{R}(0)$ is equal to $c_{\xi}\Vol(B_{R})+O(R^{d-1})$. Taking complex conjugate, for later use, this gives    
	\begin{equation}\label{eq: ck}
		\sum_{x\in\Lambda\cap B_{R}(0)}e^{2\pi i \langle \xi,x\rangle}=\overline{c_{\xi}}\Vol(B_{R})+O(R^{d-1}).
	\end{equation}
	Given $f\in\mathcal{S}(\R^{d})$ and $y\in\R^{d}$ let $f_{y}(x)=f(y-x)$. Its inverse Fourier transform is $\check{f}_{y}(\xi)=\hat{f}(\xi)e^{2\pi i\langle\xi,y\rangle}$, and the summation formula gives 
	\[|\sum_{x\in\Lambda}f(y-x)|=|\sum_{\xi\in\Lambda'}c_{\xi}\hat{f}(\xi)e^{2\pi i\langle \xi , y\rangle}|\le \sum_{\xi\in\Lambda'}|c_{\xi}\hat{f}(\xi)| ,\]
	 independently of $y$. This bound and \eqref{eq: ck} allows to switch the order of summation and take the limit inside in the next calculation
	\begin{align*}
		&\lim_{R\to\infty}\frac{1}{\mathrm{Vol}(B_{R})}\sum_{y\in \Lambda\cap B_R}\sum_{x\in\Lambda}f(y-x)\\  =&	\lim_{R\to\infty}\frac{1}{\mathrm{Vol}(B_{R})}\sum_{y\in \Lambda\cap B_R}\sum_{\xi\in\Lambda'}c_{\xi}\hat{f}(\xi)e^{2\pi i\langle \xi , y\rangle}\\
		 =& \sum_{\xi\in\Lambda'}c_{\xi}\hat{f}(\xi)\left [	\lim_{R\to\infty}\frac{1}{\mathrm{Vol}(B_{R})}\sum_{y\in \Lambda\cap B_R}e^{2\pi i\langle \xi , y\rangle}\right]\\
		 = &\sum_{\xi\in\Lambda'}|c_{\xi}|^{2}\hat{f}(\xi),
	\end{align*}
    using \Cref{eq: ck} in the last equality.
	It is left to show that for all $ f\in\mathcal{S}(\R^{d}) $
	\[\lim_{R\to\infty}\frac{1}{\mathrm{Vol}(B_{R})}\sum_{x,y\in \Lambda\cap B_R}f(y-x)=\lim_{R\to\infty}\frac{1}{\mathrm{Vol}(B_{R})}\sum_{y\in \Lambda\cap B_R}\sum_{x\in\Lambda}f(y-x),\]
	or equivalently, 
	\[\lim_{R\to\infty}\frac{1}{\mathrm{Vol}(B_{R})}\sum_{y\in \Lambda\cap B_R}\sum_{x\in\Lambda\setminus B_{R}}f(y-x)=0.\]
	Since compactly supported smooth functions are dense in $\mathcal{S}(\R^{d})$, it is enough to prove this convergence for $f\in C_{c}^{\infty}(\R^{d})$. Let $T$ large enough so that $ f $ is supported inside the ball $ B_{T}(0) $ and let $ R\gg T $. Using part (2) we can bound 
	\begin{align*}
		\left|\sum_{y\in \Lambda\cap B_{R}(0)}\sum_{x\in\Lambda\setminus B_{R}(0)}f(y-x)\right| & \le\|f\|_{\infty}\sum_{y\in \Lambda\cap B_{R}\setminus B_{R-T}}|\Lambda\cap B_{T}(y)|\\
		& \le \|f\|_{\infty} C^{2}T^{d-1}(\Vol(B_{R})-\Vol(B_{R-T}))\\&=O(R^{d-1}), 
	\end{align*}
	and therefore,
	\[\lim_{R\to\infty}\frac{1}{\Vol(B_{R})}\left|\sum_{y\in \Lambda\cap B_R}\sum_{x\in\Lambda\setminus B_{R}}f(y-x)\right|=0.\qedhere\]
\end{proof}

\section{Concrete examples of Delone Fourier quasicrystals}\label{sec:concreteexample}
Here we apply our construction to concrete examples of strict Lee--Yang varieties and examine the resulting Fourier quasicrystals. We start with a trivial example.

\begin{ex}[Lattices]
    If $X=\{(1,\ldots,1)\}\in(\pp^1)^n$, which is a strict Lee--Yang variety by \Cref{lem:zerodim}, and $L$ is an $n\times n$ matrix of positive determinant, then $\Lambda(X,L)$ is a lattice in $\R^n$. Every lattice in $\R^n$ can be obtained in this way.
\end{ex}

Next we apply the construction to one dimensional strict Lee--Yang varieties.

\begin{ex}\label{ex:rationalExplicit} 
For a concrete instance of \Cref{ex:rational}, we choose $f_1(t)=t-1$, $f_2(t)=t$ and $f_3=t+1$. Then we obtain
  \begin{equation*}
      \tilde{X} =   {\left\{\left(\frac{-1 + i + t}{-1 - i + t}, \ \ \frac{-i + t}{i + t},\ \ \frac{1 + i + t}{1 - i + t} \right) \mid t\in \C\cup\{\infty\} \right\}}
 \subseteq (\pp^1)^3
  \end{equation*}
  which is a strict Lee--Yang variety of codimension two. It is the closure of the curve considered in \Cref{ex:runningRational}. 
  Its multidegree is $(1,1,1)$. 
  One can check that $\{\theta \in [0,1]^3\mid \exp(2\pi i\theta)\in \tilde{X} \}$ affinely span all of $\R^3$. 
  Indeed, the image of $t=-1,0,1,2$ under the parametrization are affinely independent. 
  Therefore $\tilde{X}$ is not contained in any subtorus of $\T^3$. 
  One checks that the $2\times 2$ minors of
  \begin{equation*}
  L=\begin{pmatrix}
     1&0\\ 0 &1\\ -\sqrt{2}&\sqrt{3}
    \end{pmatrix}
 \end{equation*}
  are positive and linearly independent over $\Q$.
  Therefore, the set
 $\Lambda=\Lambda(\tilde{X},L)$
is a Delone Fourier quasicrystal as in \Cref{thm:main1} satisfying the properties from \Cref{thm:main2}. Its density is 
\begin{equation*}
 \Delta(\Lambda)=1+\sqrt{2}+\sqrt{3}\approx 4.14626.
\end{equation*}
The Fourier transform of its counting measure is supported on the set 
\[\Lambda' = \{L^tk \mid k\in \Z^3 \text{ with } \var(k)<2\}, \]
which is discrete, by \Cref{cor:imagediscrete}.
The curve $\tilde{X}$ and the discrete sets $\Lambda$, $\Lambda'$ appear in \Cref{fig:RationalCurveIntro}. Using \Cref{eq:ftonx} we can even compute the coefficients of the Fourier transform. For example the coefficient of $(0,0)\in\Lambda'$ equals the density $\Delta(\Lambda)$:
\begin{equation*}
	c_{(0,0)}=1+\sqrt{2}+\sqrt{3}\approx 4.14626.
\end{equation*}
For another example consider the point $\xi=\left(\sqrt{2}-1,-\sqrt{3}-1\right)\in\Lambda'$. Again using the integral formula in \Cref{eq:ftonx} we compute its Fourier coefficient as
\begin{equation*}
 c_\xi= \frac{4}{25} i \left((-2+i)+(2+i) \sqrt{2}+2 i \sqrt{3}\right)\approx -0.94053+0.132548 i.
\end{equation*}
\end{ex}

\begin{ex}\label{ex:elliptic}
    Consider the smooth projective curve $X$ of genus one defined by 
    \begin{equation}\label{eq:ellipticex}
       y^2= 1 + 4 x^2 - 8 x^3 + 4 x^4.
    \end{equation}
    By this we mean the normalization of the projective closure of the affine curve cut out by \Cref{eq:ellipticex}. Because $X(\R)$ has $2=g+1$ connected components, namely given by $y>0$ and $y<0$, the curve $X$ is separating. We let $X_+$ be the connected component of $X(\C)\setminus X(\R)$ that contains the point
    \begin{equation*}
        Q=\left(\frac{1+ \sqrt{2} + i}{2} , i\right).
    \end{equation*}
    We now consider the following three rational functions on $X$:
    \begin{equation*}
        f_1=x,\, f_2=-\frac{-1 + 4 x - 2 x^2 + y}{2 (-1 + x)},\, f_3=\frac{1 - 4 x + 2 x^2 - y}{4 (-1 + x) x}.
    \end{equation*}
    We claim that all three functions are separating. To see this, we first note that the imaginary part of $f_i(Q)$ is positive for $i=1,2,3$. Therefore, by \Cref{lem:interlcrit} it suffices to show that the zeros and poles of $f_i$ interlace on both connected components of $X(\R)$ for $i=1,2,3$. This amounts a straight-forward calculation. For example $f_1$ has zeros $(0,\pm1)$ and one pole at infinity on each connected component of $X(\R)$. The image of $X$ under the map
    \begin{equation*}
        X\to\pp^1\times\pp^1\times\pp^1,\, P\mapsto(f_1(P),-f_2(P),f_3(P))
    \end{equation*}
    is the zero set of the following two polynomials:
    \begin{equation*}
     P_1=2z_1w_2z_3+w_1z_2w_3\textrm{ and }P_2=2(z_1-w_1)(w_2-z_2)(z_3-w_3)-w_1z_2w_3.   
    \end{equation*}
    Here $w_i,z_i$ are the homogeneous coordinates on the $i$-th copy of $\pp^1$ for $i=1,2,3$. One checks that this is a smooth curve of genus one and hence the map is an embedding. Therefore, by \Cref{prop:leeyangconstr} its image $\tilde{X}$ under the coordinate-wise M\"obius transformation $z\mapsto\frac{z+i}{z-i}$ is a strict Lee--Yang variety. Its multidegree is $(2,2,2)$. As in \Cref{ex:rationalExplicit} one shows that$X$ is not contained in a proper subtorus.
    The entries of the matrix
    \begin{equation*}
    L=\begin{pmatrix}
        \frac{\exp(1)}{6} & \frac{\exp(\sqrt{5})}{3} & \frac{\exp(\sqrt{3})}{3}\\
        \frac{\exp(\sqrt{2})}{6} & \frac{\exp(\sqrt{11})}{3} & \frac{\exp(\sqrt{13})}{6}
    \end{pmatrix}
\end{equation*}
are algebraically independent over $\Q$ by Lindemann's theorem, so in particular the $2\times 2$ minors of $L$ are linearly independent over $\Q$. Therefore, the  set 
$\Lambda=\Lambda(\tilde{X},L)$
is a Delone Fourier quasicrystal as in \Cref{thm:main1} with the properties as in \Cref{thm:main2}. Its density is 
\begin{equation*}
 \Delta(\Lambda)\approx 10.6583.
\end{equation*}
\Cref{fig:elliptic} is a picture of the Voronoi diagram of $\Lambda$. The color of a cell determined by a point $x\in\Lambda$ is chosen according to the connected component of $\tilde{X}(\T)$ which contains $\exp(2\pi iLx)$.
\end{ex}

\begin{figure}
\begin{center}
\includegraphics[height=1.5in]{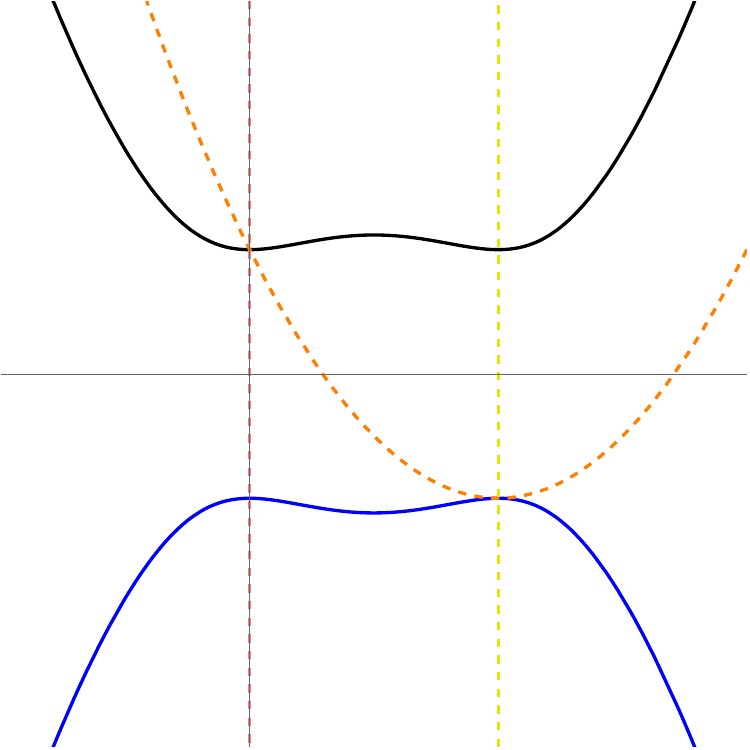} \ \ \ \  \includegraphics[height=1.5in]{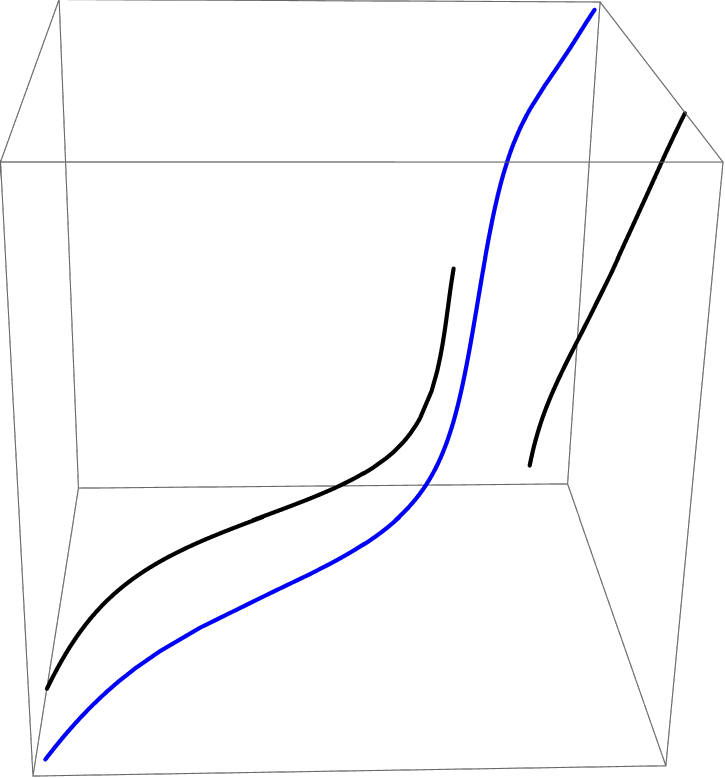}
\end{center}
\caption{The curves $X(\R)$ and $\{\theta \mid \exp(2\pi i\theta)\in \tilde{X} \}$ from \Cref{ex:elliptic}.}
\label{fig:ellipticCurve}
\end{figure}

\begin{figure}
\begin{center}
\includegraphics[height=1.8in]{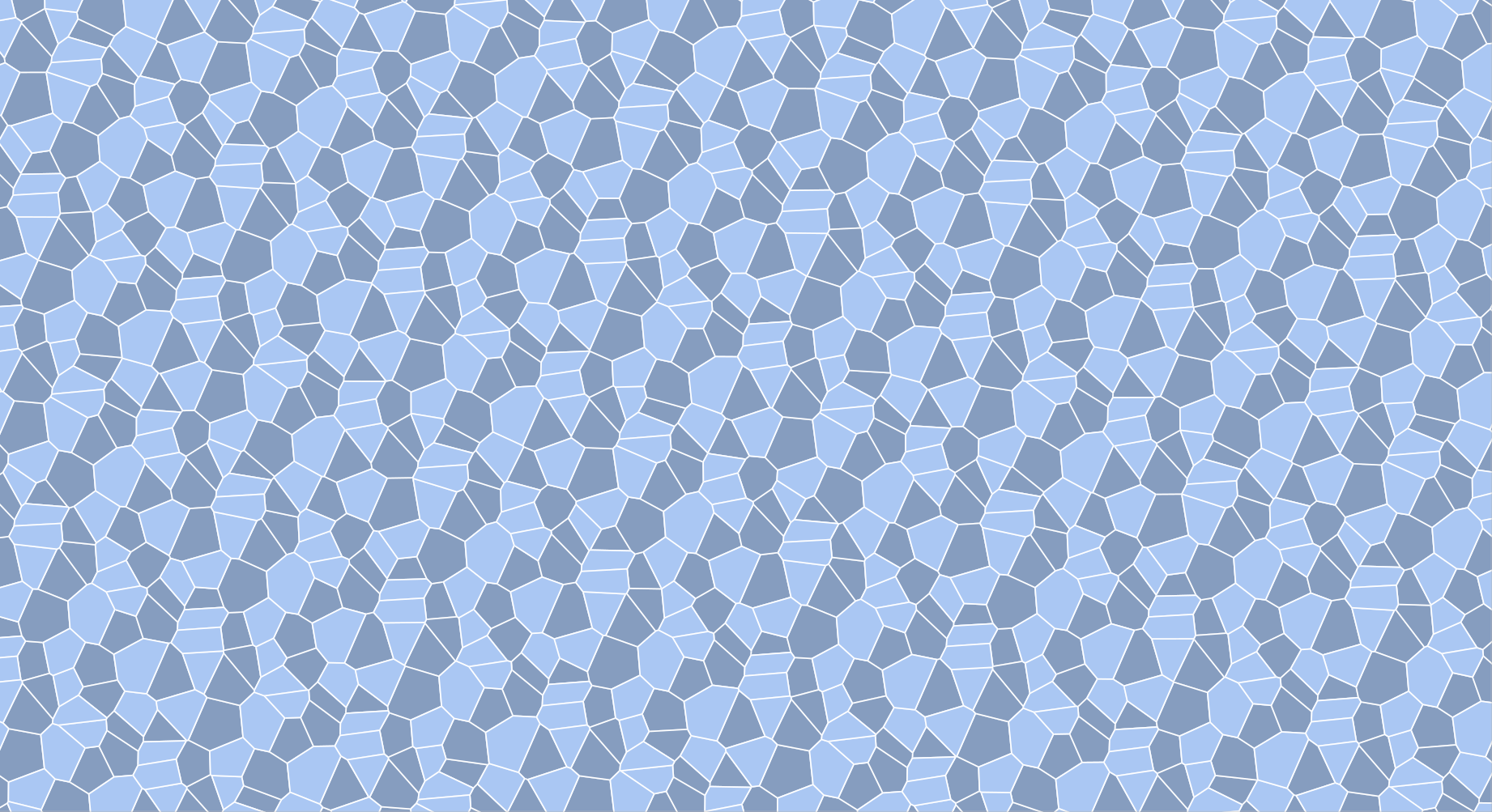}
\end{center}
\caption{The Voronoi diagram of the Fourier quasicrystal from \Cref{ex:elliptic}.}\label{fig:elliptic}
\end{figure}

\begin{ex}\label{ex:product2}
Here we consider a concrete instance of \Cref{ex:product}.
Let $X\subseteq (\pp^1)^4$ be  the variety defined by the two equations
\[-3w_1w_2 - z_1w_2 + w_1z_2 + 3 z_1 z_2 = 0 \text{ and } -3w_3w_4 - z_3w_4 + w_3z_4 + 3 z_3 z_4 = 0. \]
This is a product of two Lee--Yang curves in $(\pp^1)^2$ and has codimension two.  
The multidegree of $X$ is $(d_{12}, d_{13}, d_{14}, d_{23}, d_{24}, d_{34}) = (0,1,1,1,1,0)$.  
Let $L = \begin{pmatrix} 1 & 0& -ac  & -a  \\ 0 & 1 & abc + d & ab  \end{pmatrix}^t $ with $a =\sqrt{2}/2$, $b=\sqrt{3}/3$, $c = \sqrt{5}/5$ and $d = \sqrt{7}/7$. 
One can check that the minors of $L$ are positive and linearly independent over $\Q$. 
The resulting Fourier quasicrystal $\Lambda = \{x\in \R^2\mid \exp(2\pi i Lx)\in X\}$ (see \Cref{cor:leeyangfqc}) is defined by the 
trigonometric polynomial equations 
\[ \sin (\pi  (x_1- x_2))-3 \sin (\pi (x_1+ x_2) =0 \text{ and } 
 \sin (\pi(\ell_1- \ell_2))-3 \sin (\pi (\ell_1+ \ell_2))=0\]
 where $\ell_1 = -a c x_1 + (a b c + d)x_2$ and  $\ell_2 = -a x_1 + a b x_2$.
Figure~\ref{fig:prod} shows the zero sets of these trigonometric polynomial equations intersecting in $\Lambda$.
The discrete set 
\begin{equation*}
 \Lambda' = \{L^t k \mid k\in \Z^4, \var(k)<2\}   
\end{equation*}
supporting its Fourier transform is shown on the right.  
Note that the span of $\Lambda'$ over $\Q$ is a four-dimensional, whereas in \Cref{ex:rationalExplicit} and \Cref{ex:elliptic} 
the $\Q$-span of $\Lambda'$ was only three-dimensional. 
\end{ex}

\begin{figure}
\begin{center}
\includegraphics[height=1.8in]{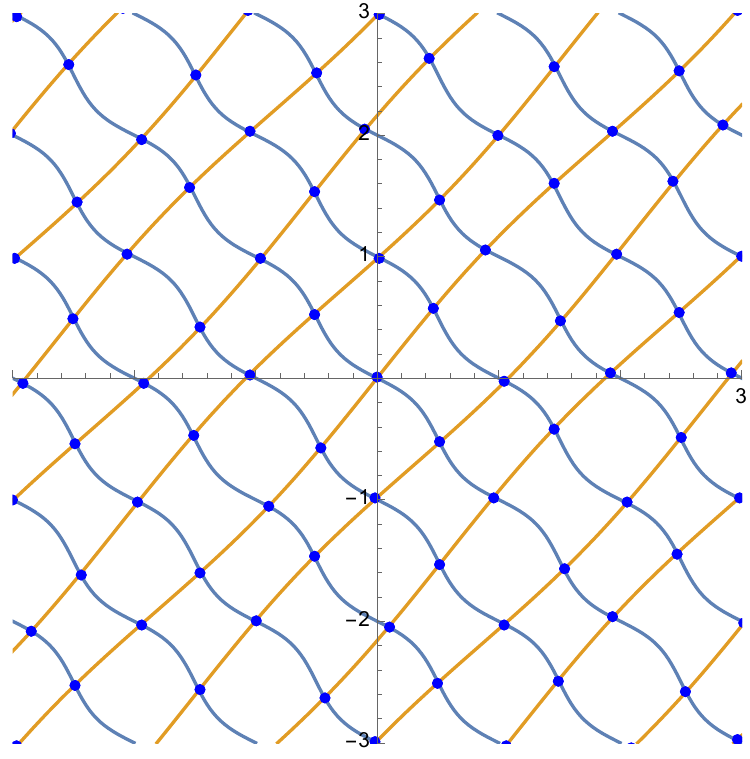}  \ \ \ 
\includegraphics[height=1.8in]{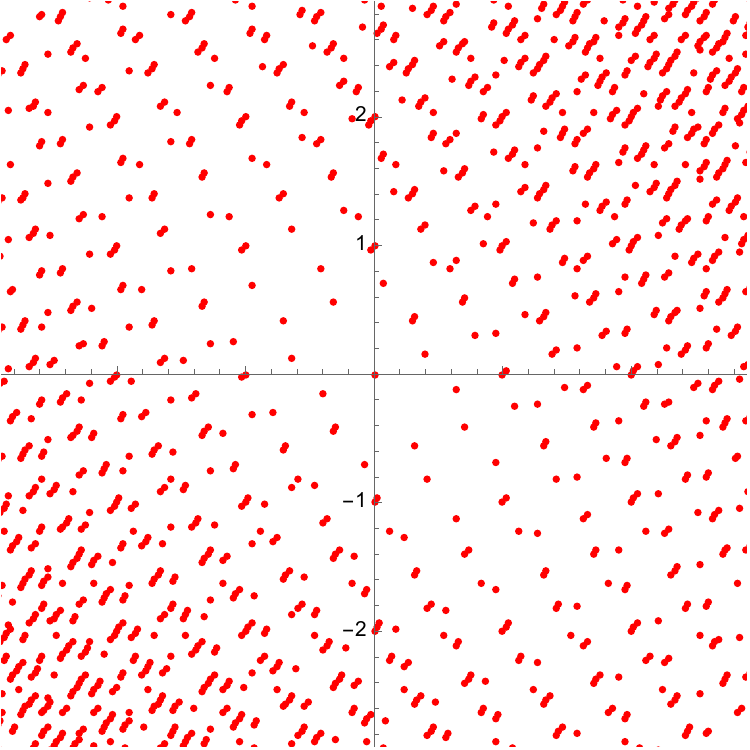} 
\end{center}
\caption{The Fourier quasicrystal $\Lambda$ from \Cref{ex:product2}  and discrete set $\Lambda'$ supporting its Fourier transform.}\label{fig:prod}
\end{figure}

\begin{ex}\label{ex:notcompleteintersection}
    Consider the following three rational functions on $\pp^1$:
    \begin{eqnarray*}
        f_1&=& \frac{\left(t^2-t-1\right) \left(t^2+t-1\right)}{2 t \left(2 t^2-3\right)},\\
        f_2&=&-\frac{t^4+3 t^3-3 t^2-5 t+1}{2 (t-1) \left(t^3-t^2-5 t-2\right)},\\
        f_3&=&\frac{2 t^4-5 t^2+1}{2 t \left(3 t^2-4\right)}.
    \end{eqnarray*}
    Each of them maps $i$ to the upper half-plane. By inspecting their zeros and poles and using \Cref{lem:interlcrit} we see that all three are separating. We  consider
        \begin{equation*}
        f=(f_1,-f_2,f_3)\colon \pp^1\to(\pp^1)^3.
    \end{equation*}
    First we compute $f_j-\frac{i}{2}$ for $j=1,2,3$:
    \begin{eqnarray}\label{eq:ihalf1}
        f_1-\frac{i}{2}&=& \frac{(t-i) \left(t^3-i t^2-2 t+i\right)}{2 t \left(2 t^2-3\right)}\\ \label{eq:ihalf2}
        f_2-\frac{i}{2}&=& -\frac{((1+i) t+(2-i)) \left(t^3-i t^2-2 t+i\right)}{2 (t-1) \left(t^3-t^2-5 t-2\right)}\\
        f_3-\frac{i}{2}&=& \frac{(2 t-i) \left(t^3-i t^2-2 t+i\right)}{2 t \left(3 t^2-4\right)}.
    \end{eqnarray}
    This shows that $f^{-1}(\frac{i}{2},-\frac{i}{2},\frac{i}{2})$ consists of the three distinct zeros of the polynomial 
    \begin{equation*}
     g=t^3-i t^2-2 t+i.   
    \end{equation*}
    In particular, the map $f$ is not injective. However, in order to apply \Cref{prop:leeyangconstr}, we will prove that the restriction of $f$ to $\pp^1(\R)$ is injective. To this end, we first show that $f^{-1}(f(i))$ has only one element. Assume for the sake of a contradiction that there exists $z\in f^{-1}(f(i))\setminus\{i\}$. Then $g(z)=0$ by (\ref{eq:ihalf1}). This in turn shows that $f_3(z)=\frac{i}{2}$. But since $f_3(i)\neq\frac{i}{2}$, this shows that $f^{-1}(f(i))=\{i\}$. Furthermore, because $i$ is a simple zero of $f_1-\frac{i}{2}$, this shows that the map $f$ is birational onto its image.
    Next we observe that if $f(z)=f(z')$ for $z\neq z'$, then 
    \begin{equation*}
     (f_1(z),f_3(z))=(f_1(z'),f_3(z'))   
    \end{equation*}
    is a singular point of the image of the planar curve in $(\pp^1)^2$ given as the image of $(f_1,f_3)$. This is cut out by the following bihomogeneous polynomial
    \[
        -{w_1}^4 {w_3}^4+92 {w_1}^4 {w_3}^2 {z_3}^2-1936 {w_1}^4 {z_3}^4-196 {w_1}^3 {w_3}^3 {z_1} {z_3}+8416 {w_1}^3 {w_3} {z_1} {z_3}^3
      \]\[  
        +92 {w_1}^2 {w_3}^4 {z_1}^2-13120 {w_1}^2 {w_3}^2 {z_1}^2 {z_3}^2+768 {w_1}^2 {z_1}^2 {z_3}^4+8640 {w_1} {w_3}^3 {z_1}^3 {z_3}\]\[-1600 {w_1} {w_3} {z_1}^3 {z_3}^3-2048 {w_3}^4 {z_1}^4+768 {w_3}^2 {z_1}^4 {z_3}^2.
    \]
    This curve has exactly three real singularities, namely   $(\pm \frac{3 \sqrt{\frac{11}{14}}}{8}, \pm \sqrt{\frac{2}{77}})$ and $(\infty,\infty)$. Each of them has two preimages under $(f_1,f_3)$ which are:
    \[\frac{1}{14} \left(-2 \sqrt{14}-\sqrt{154}\right)\textnormal{ and }\frac{1}{14} \left(2 \sqrt{14}-\sqrt{154}\right),\]
    \[\frac{1}{14} \left(\sqrt{154}-2 \sqrt{14}\right)\textnormal{ and }\frac{1}{14} \left(2 \sqrt{14}+\sqrt{154}\right),\]
    \[0\textnormal{ and }\infty.\]
    One verifies that $f_2$ takes different values on the two preimages of each real singularity. This shows that $f$ is injective on $\pp^1(\R)$ and therefore by \Cref{prop:leeyangconstr} the curve $\tilde{X}\subseteq(\pp^1)^3$ parametrized by
    \begin{equation*}
        \left( \frac{f_1+i}{f_1-i},\frac{-f_2+i}{-f_2-i},\frac{f_3+i}{f_3-i}\right)
    \end{equation*}
    is a strict Lee--Yang variety of multidegree $(4,4,4)$. Using a computer algebra system we find three Laurent polynomials that cut out $\tilde{X}\cap(\C^*)^3$. Under the map
    \begin{equation*}
        \R^2\to(\C^*)^3,\, x\mapsto\exp(2\pi i Lx)
    \end{equation*}
    these pull back to three exponential polynomials such that their common zero zero set in $\R^2$ is a Fourier quasicrystal whenever $L\in\R^{3\times 2}$ has only positive $d\times d$ minors. We have depicted their zero sets in \Cref{fig:notcompleteintersection} for     
    \begin{equation*}
  L=\begin{pmatrix}
     1&0\\ 0 &1\\ -\sqrt{2}&\sqrt{3}
    \end{pmatrix}.
 \end{equation*}
 For each two of the three exponential polynomials, there are points where these two vanish but the third one does not. Only those points in which all three vanish belong to the support of our Fourier quasicrystal. We even claim that there is no pair of Laurent polynomials that cuts out $\tilde{X}\cap(\C^*)^3$ scheme theoretically. To this end, we examine the tangent cone of $\tilde{X}$ at the point $P=(3,\frac{1}{3},3)$, which is the image of $(\frac{i}{2},-\frac{i}{2},\frac{i}{2})$ under the coordinate-wise M\"obius transformation $z\mapsto\frac{z+i}{z-i}$. It is the union of the tangent vectors of the three branches of $\tilde{X}$ meeting at $P$ corresponding to the three zeros of $g$. One checks that these three tangent vectors are linearly independent. Thus they correspond to three points in $\pp^2$ in general position. If $\tilde{X}\cap(\C^*)^3$ was cut out scheme theoretically by two Laurent polynomials $p_1,p_2$, then these three points would be the common zero set of the initial forms of $p_1$ and $p_2$ at $P$, each with multiplicity one. This is not possible by B\'ezout's theorem.
\end{ex}

\begin{figure}
\begin{center}
\includegraphics[height=1.8in]{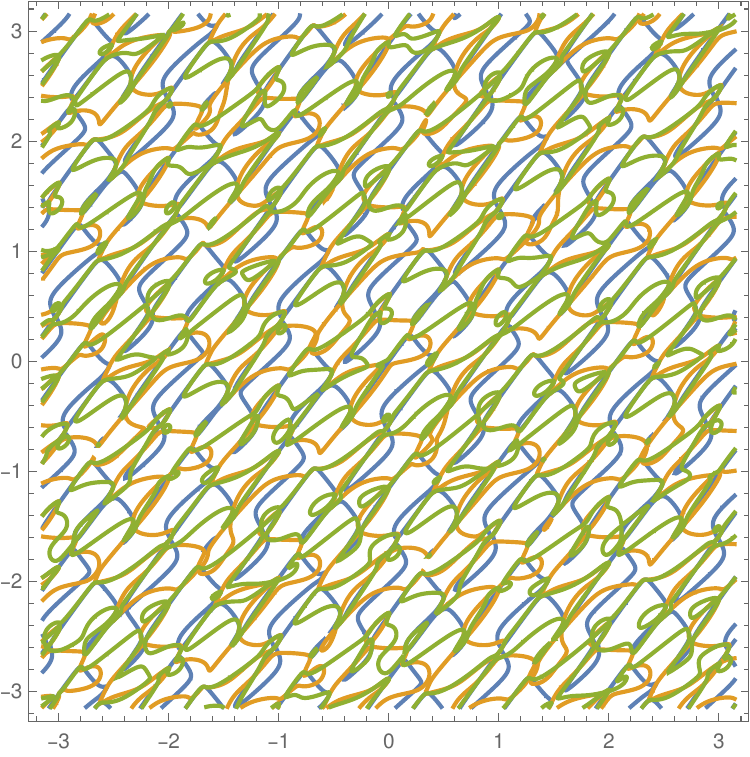}  \ \ \ 
\includegraphics[height=1.8in]{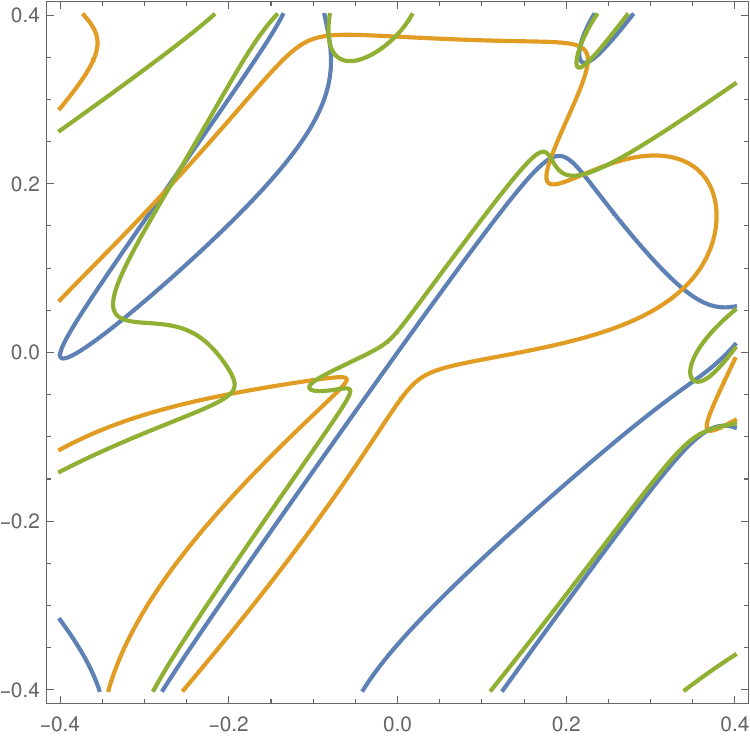} 
\end{center}
\caption{The zero sets of the three exponential polynomials from \Cref{ex:notcompleteintersection} in blue, green and orange. Only those points that lie on all three curves belong to the support of the Fourier quasicrystal.}
\label{fig:notcompleteintersection}
\end{figure}

 \noindent \textbf{Acknowledgements.} 
The authors would like to thank Philipp di~Dio, Larry Guth, Dominique Maldague and Peter Sarnak for interesting discussions on the question of high-dimensional FQs.  This project started at the Oberwolfach meeting on ``New Directions in Real Algebraic Geometry'' in March of 2023.

\bibliographystyle{plain}
\bibliography{biblio}

 \end{document}